\documentclass[12pt,english]{article}

\usepackage[T1]{fontenc}
\usepackage[margin=1.25in]{geometry}

\usepackage[utf8]{inputenc}
\usepackage{color}
\usepackage[dvipsnames]{xcolor}
\usepackage{enumerate}
\usepackage{enumitem}
\usepackage{hyperref}
\definecolor{darkblue}{rgb}{0, 0, 0.5}
\hypersetup{
     colorlinks = true,
     citecolor = Maroon,
     urlcolor = darkblue,
     linkcolor = darkblue
}

\usepackage{amsmath}
\usepackage{amsthm}
\usepackage{amssymb}
\usepackage{cases}
\usepackage{bm}
\usepackage[authoryear]{natbib}

\usepackage{setspace}
\usepackage{mathtools}

\AtBeginDocument{
  \setlength{\abovedisplayskip}{6pt plus 2pt minus 2pt}
  \setlength{\belowdisplayskip}{6pt plus 2pt minus 2pt}
  \setlength{\abovedisplayshortskip}{3pt plus 2pt minus 2pt}
  \setlength{\belowdisplayshortskip}{6pt plus 2pt minus 2pt}
  \setlength{\jot}{2pt}
}

\allowdisplaybreaks[2]

\usepackage{titlesec}

\usepackage{titlesec}

\titlespacing*{\section}
  {0pt}{1.0ex plus 0.3ex minus 0.2ex}{0.5ex plus 0.2ex}

\titlespacing*{\subsection}
  {0pt}{0.8ex plus 0.2ex minus 0.1ex}{0.4ex plus 0.1ex}

\titlespacing*{\subsubsection}
  {0pt}{0.6ex plus 0.2ex minus 0.1ex}{0.3ex plus 0.1ex}

\titlespacing*{\part}
  {0pt}{0.8ex plus 0.2ex minus 0.1ex}{0.4ex plus 0.1ex}

\usepackage{graphicx}
\usepackage{float}
\usepackage{caption}
\usepackage[flushleft]{threeparttable}

\usepackage{pdflscape}
\usepackage{multicol}
\usepackage{dsfont}
\usepackage{booktabs}
\usepackage{adjustbox}

\newcommand{\R}{\mathbb{R}}
\newcommand{\E}{\mathbb{E}}

\newcommand{\bydef}{\triangleq}

\newcommand{\weakto}{\rightsquigarrow}
\newcommand{\pTo}{\stackrel{p}{\rightarrow}}

\usepackage{algorithm}

\usepackage{comment}

\usepackage{multirow}

\usepackage{tikz}

\usetikzlibrary{arrows.meta,decorations.pathreplacing,calligraphy}

\makeatletter

\theoremstyle{plain}
  \newtheorem{rem}{\protect\remarkname}
  \newtheorem{example}{\protect\examplename}

  \theoremstyle{plain}
  \newtheorem{assumption}{\protect\assumptionname}

\theoremstyle{plain}
\newtheorem{thm}{\protect\theoremname}
\newtheorem*{thm*}{\protect\theoremname}
\newtheorem{cor}{\protect\corname}
  \theoremstyle{plain}
  \newtheorem{lem}{\protect\lemmaname}[section]
    
  \theoremstyle{remark}

\makeatother

\usepackage{babel}
  \providecommand{\assumptionname}{Assumption}
  \providecommand{\claimname}{Claim}
  \providecommand{\lemmaname}{Lemma}
  \providecommand{\remarkname}{Remark}
\providecommand{\theoremname}{Theorem}
  \providecommand{\examplename}{Example}
\providecommand{\corname}{Corollary}

\title{\vspace{-2cm}

Wild inference for wild SVARs with application to Volatility-based IV\thanks{  \footnotesize{
Bulat Gafarov: Department of Agricultural and Resource Economics, University of California, Davis.
Email: \href{bgafarov@ucdavis.edu}{bgafarov@ucdavis.edu};
Madina Karamysheva: HSE University. Email: \href{mkaramysheva@hse.ru}{mkaramysheva@hse.ru};
Andrey Polbin: Bank of Russia; Gaidar Institute. Email: \href{apolbin@gmail.com}{apolbin@gmail.com};
Anton Skrobotov: Centre for Big Data in Economics and Finance (CEBDA), HSE University. Email: \href{antonskrobotov@gmail.com}{antonskrobotov@gmail.com}.

The opinions expressed in the paper are solely those of the authors and may not reflect the official position of the
affiliated institutions.
We would like to thank Regis Barnichon, Colin Cameron, James Cloyne, Raffaella Giacomini, Yuriy Gorodnichenko, Lars Hansen, Jens Hilscher, Silvia Gon\c{c}alves, \'Oscar Jord\'a, Karel Mertens, Daniel Lewis, Serena Ng, Valerie  Ramey, Andres Santos, Sanjay Singh, Takuya Ura, Ke-Li Xu, participants of 2024 NorCal junior econometricians' conference at UC Santa Cruz,   2024 California econometrics conference, ICEBA 2024 conference, HSE Brown bag seminar, EEA Congress 2025, Sailing Macro Workshop 2025, World Congress 2025,  for constructive feedback and suggestions. The previous version of the paper was circulated under title ``Policymaker meetings as heteroscedasticity shifters: Identification and simultaneous inference in unstable SVARs''.}
}}
\author{
\setcounter{footnote}{1}
 Bulat Gafarov \and
	Madina Karamysheva\and Andrey Polbin\and Anton Skrobotov
}

\begin{document}
\maketitle
\thispagestyle{empty}

\vspace{-1cm}

\begin{abstract}

We propose a dependent wild bootstrap method based on local projections for computing the joint asymptotic distribution of parameter estimates in  structural vector autoregression models.
This procedure can be applied to the raw data in levels without pretesting while remaining robust   to  unit roots, cointegration, polynomial trends, and conditionally heteroscedastic shocks in a general form.
We show how knowledge of the joint asymptotic distribution in persistent data setups can improve the efficiency of impulse response function estimators  through smoothing, narrow multi-horizon confidence bounds, and deliver weak identification robust inference using external moments.
We illustrate these findings in simulations and  apply the method to US monetary policy shocks identified by FOMC-meeting-induced volatility.

\medskip

\medskip

\noindent \textbf{Keywords:}  SVAR, smoothed local projections, external instruments, conditional heteroscedasticity,  unit roots, cointegration, simultaneous inference,  monetary policy shocks.

\end{abstract}
\newpage

\section{Introduction}

Local projection (LP) estimators of impulse response functions (IRFs) \citep{jorda2005estimation,dufour2006short}  are increasingly popular in structural vector autoregression (SVAR) literature due to their computational simplicity and double robustness properties.
Namely, the   least squares regression coefficients corresponding to LP estimators are remarkably robust to the presence of unit roots and an unmodeled moving average component \citep{montiel2021local,Montiel2026double}.
The LP literature on robust inference has focused on IRFs at individual horizons, which leaves out a large class of modern structural identification techniques, such as identification through external moment restrictions \citep{rigobon2003identification, lewis2021identifying} or sign restrictions \citep{uhlig2005effects,gafarov2018delta,antolin2018narrative}.
Moreover, the classical LP estimators can potentially lose efficiency by ignoring smoothness restrictions on IRFs across horizons implied by the VAR model structure.\footnote{\cite{barnichon2018functional} proposed smoothing LP IRFs using the minimum distance (MD) method for better statistical efficiency, which requires knowledge of the joint asymptotic covariance matrix of LP estimates.
To our knowledge, there are no proposals on how to consistently estimate this matrix in the presence of unit roots. }
These limitations impose an uncomfortable trade-off between the robustness to unit roots of classical LP and the efficiency and structural identification versatility of iterative IRF estimators.

We develop a large sample theory for the joint distribution of LP estimators and study implications for inference on SVAR parameters in general, potentially nonstationary but nonexplosive cases.
Specifically, we propose a proof of the joint asymptotic normality of the local projection estimators across horizons and the corresponding asymptotic  covariance matrix formula  in the potential presence of trends, cointegration, any number of unit roots, and general weakly-dependent stationary conditionally heteroscedastic innovations.\footnote{The standard large sample least squares regression theory does not directly apply to nonstationary time series since the average Gram matrix of regressors $  T^{-1}X'X$ is unbounded as the sample size grows rather than having a  finite limit in probability. Our proof builds upon the results for  homoscedastic  VARs with unit roots \citep{sims1990inference,tsay1990asymptotic} and the functional central limit theorem for  stochastic integrals of conditionally heteroscedastic processes \citep{de2000functional}.}
Next, we extend the proof of the consistency of the joint heteroscedasticity and autocorrelation consistent (HAC) covariance matrix estimator  of
\cite{andrews1991heteroskedasticity} to accommodate nonstationary regressors with common finite order VAR dynamics.
Finally, we propose a version of the dependent wild bootstrap (DWB) procedure of \cite{shao2010dependent}, which provides a computationally convenient way to estimate the joint distribution of the local projection estimators  and the corresponding estimators of structural parameters.

The only regularity conditions we impose on the nonexplosive (possibly unit root)  VAR model are that the one-step-ahead forecast errors, along with the external variables used for structural identification, follow a stationary vector process with a mixing property and the existence of certain higher-order moments, which are necessary for the consistency of the HAC estimator.
Under this condition, the products of the multi-step forecast errors (the feasible analog of the statistical scores of LP estimators) and external variables have a coupling with a stationary mixing process, which makes CLT and DWB valid.
To our knowledge, this argument provides the first proof of the validity of \emph{joint} Gaussian inference on the IRF estimators that allows for the general case of unstable but nonexplosive VAR models with a general form of stationary conditional heteroscedasticity without pretests.

Among the multiple potential uses of DWB robust to nonstationary data, we focus on the following three: (i)  minimum distance (MD) smoothing of IRFs \citep{jorda2011estimation,barnichon2018functional, barnichon2019impulse}; (ii)  joint inference on multiple SVAR parameters using sup-$t$ statistic  \citep{jorda2009simultaneous,montiel2019simultaneous}; (iii) inference on structural IRFs identified using potentially weak higher order moments conditions \citep{lewis2022robust}.
We illustrate these benefits  in two  simulation designs and an empirical study of the impact of US monetary policy shocks in a six-variable VAR model.

\emph{Efficient estimation of IRFs.} The MD approach improves the efficiency of LP IRF estimates by imposing the underlying VAR structure on the IRF parameters.
The joint covariance matrix of LP estimates can be used either directly as an inverse weight matrix (if the number of horizons for smoothing is small) or indirectly (in the high-dimensional case).
In the latter case, one can use suboptimal  diagonal (precision-based) weights in MD estimators and then compute an efficient component-wise average of the resulting MD and the classical LP estimator. By construction, this efficient weighted estimator is guaranteed to have an asymptotic variance smaller than or equal to the minimum of the variances of the constituent LP or MD estimators.
Since, as we show, the lag-augmented VAR estimator of \cite{toda1995statistical} can be represented as a smoothed LP estimator with an equal number of VAR lags and smoothing horizons, our weighted estimator is weakly more asymptotically efficient than   popular regular estimators of IRFs with  unit-root robustness.
We illustrate this point in a Monte Carlo study based on an AR(1) design, showing that such a weighted estimator can adaptively match the confidence interval length  of the more efficient estimator. In designs with low persistence, it matches the length of  recursive IRF estimates (smoothed LP), while in near-unit root designs  it matches the classical LP estimators.
Moreover, in unit root designs, the weighted estimator also considerably improves the coverage probability for confidence intervals at long horizons, where classical LP demonstrates less than nominal coverage probability.

\emph{Efficiency gains in simultaneous confidence bands.} Using DWB, it is straightforward to compute critical values for sup-$t$ tests \citep{dunnett1955multiple} for testing multiple hypotheses about functions of the IRFs over multiple horizons and external moments.
\cite{montiel2019simultaneous} showed that critical values based on sup-$t$ result  in tighter simultaneous confidence bounds than alternatives that are agnostic about the joint covariance of the estimators, such as  the Bonferroni bounds  \citep{dunn1958estimation,inoue2023significance} or Scheff{\'e}  bounds \citep{scheffe1953method,jorda2009simultaneous}.
Using DWB the  efficient  sup-$t$ bounds are applicable without cumbersome specification searches for the exact nature of the violation of stationarity of the data.
We also discuss how sup-$t$ tests can be used for \cite{hausman1978specification}-type specification tests.

\emph{Anderson-Rubin bounds for structural IRFs based on external moments.} The robustness to the heteroscedasticity property of DWB allows the use of the informational content of heteroscedastic VAR innovations for identifying structural shocks.
Suppose that one observes an external variable  that is correlated with  the  volatility of a particular structural shock while being independent of the other structural shocks. We will refer to this variable as a volatility-based instrumental variable (VIV), which can be binary, as in the case considered in \cite{rigobon2003identification,rigobon2004impact}, or continuous, as in the case of lagged squared forecast errors considered in \citep{lewis2021identifying}.\footnote{The literature on identification through heteroscedasticity has a long history; see also \cite{feenstra1994new,brunnermeier2021feedbacks}.}
More interestingly, one can construct an external VIV, for example, using the impact of shock volatility on financial trading during the days of the FOMC meetings to identify monetary policy shocks.
We provide a proof that a VIV of this kind can uniquely identify the impact of a structural shock on the data and, together with the LP estimates, provide  the corresponding structural IRF estimates.
In our second Monte Carlo study based on a linearized \cite{smets2007shocks} DSGE model with  conditionally heteroscedastic monetary policy shocks, we show that DWB can successfully recover the true structural IRFs in realistic sample sizes if the volatility shifter is directly observed.\footnote{A linearized \cite{smets2007shocks} model is equivalent to a VARMA(3,2) model \citep{morris2016varma} which has a $VAR(\infty)$ representation. In this case, the VAR(8) model is only approximately correct, but LP-based DWB  delivers confidence bands for the structural IRFs with correct simulated coverage,  mimicking the theoretical findings about the effects of misspecification on LP of \cite{Montiel2026double}. }
Finally, following the observation of \cite{lewis2022robust} that many commonly used  higher order moment  schemes are only weakly identified in practice, we show how to use DWB to construct \cite{anderson1949estimation}-type confidence bands for the structural IRF identified with weak VIV.

We provide an empirical illustration of DWB to study the impact of monetary policy shocks in a  VAR model that includes six variables for the U.S.: real GDP, GDP deflator, commodity price index, non-borrowed reserves, the 1-year
bond rate, and the federal funds rate.
As VIV, we use monthly FOMC meeting counts \citep{rigobon2004impact,barnichon2025innovations}, including
meetings, telephone and video conferences, unscheduled meetings, and sequential day meetings.
This data is publicly available, allowing us to extend the sample to cover the period from January 1950 to December 2024 with monthly data.
We also consider a weighted meeting count  VIV, where we weight each meeting according to meeting-related surprises in trading activity (share volume or dollar volume) on the days of post-meeting announcements.
This approach allows us to distinguish between routine meetings and those that caused a significant shock to the financial market.

Our main empirical findings can be summarized as follows: (i) we document considerable efficiency gains from smoothing LP and sup-$t$ critical values in the classical Cholesky identification scheme \textit{{\`a} la }\cite{christiano1999monetary}; (ii) VIV results with the simple meeting-count instrument are statistically indistinguishable from those obtained using Cholesky; (iii) weak-IV robust bands show a significant negative impact of a contractionary  monetary policy shock on real GDP and a statistically significant price puzzle for the GDP deflator and commodity price index.\footnote{The price puzzle is sensitive to the choice of VIV variable. For the baseline VIV, it is statistically significant, while for the trading-volume weighted VIV, the puzzle is statistically insignificant.}

Information about policymaker meetings can  potentially be useful beyond monetary policy analysis. For example, \cite{kanzig2021macroeconomic} shows that OPEC meetings have high-frequency impacts on oil prices that can be used to identify oil supply shocks; \cite{janzen2018commodity} show that cotton storage levels change the volatility of precautionary demand shocks.
 Therefore, we believe that our VIV approach can be used more broadly.
Moreover, DWB can also be adopted to incorporate the standard external-IV moments \citep{stock2008what,mertens2013dynamic,montiel2021inference}
 or be used as input for sign-restriction frequentist inference \citep{gafarov2018delta,gafarov2016projection,drautzburg2021refining}.

The theoretical literature on inference for VAR models with unit roots goes back to \cite{dickey1979distribution}.
The classical textbook approach suggests doing a pre-test for unit roots, followed by first-differencing or de-trending \citep{hamilton1994time}.
Such a specification search can be challenging since it is sensitive to the presence of cointegration relations and polynomial trends that are popular among practitioners \citep{ramey2018government}. For example, \cite{bierens1997testing} shows that the unit root hypothesis for the US price level and interest rate is rejected for some popular tests once sufficiently flexible polynomial trends are included. Furthermore, seasonal unit roots pose additional challenges in the specification search \citep{ghysels2001econometric,cavaliere2019wild}.
If  stationary data is transformed into first differences without pretesting, there is a potential efficiency loss due to the introduction of a non-invertible moving average component from over-differencing \citep{plosser1977estimation}.
Moreover, pre-testing and specification search must be accounted for in the computation of the confidence bounds \citep{cavanagh1995inference}.
There is a long literature that aims to avoid this unit-root pre-test bias \citep{toda1995statistical,inoue2002bootstrapping,inoue2020uniform,mikusheva2012one,montiel2021inference, xu2023local}.
To our knowledge, our proof of the validity of the DWB method is more general than the earlier proposals since we consider proper multivariate VARs, any number of unit roots (including seasonal ones), polynomial trends, and general conditionally heteroscedastic innovations, while delivering tighter  (weighted LP) bands than the lag-augmented VAR recursive IRF estimates.

The paper is structured as follows. Section \ref{sec:model} discusses the setup. Section \ref{sec:estimation} provides the estimation and inference algorithms along with the corresponding theorems. Section \ref{sec:applicationDWB} discusses the uses of DWB. Section \ref{sec:symandempirics} summarizes the simulation evidence and  an empirical application to US monetary policy. Section \ref{sec:conclusion} concludes. All theorem proofs are provided in the Appendix. The Supplementary Appendix contains a technical lemma, additional figures, additional simulation designs, and details about data collection for the empirical illustration.

\section{``Wild'' SVAR model }\label{sec:model}

\subsection{Allowing for heteroscedastic and nonstationary data}

Many macroeconomic applications are concerned with  time series that feature smooth variation in mean over time  (for example, GDP) and also exhibit periods of elevated volatility.
Frequently, these features are not modeled explicitly.
Instead, they are corrected before the estimation stage, which may cause pre-testing bias.
In this paper, however, we follow an explicit modeling approach.
Suppose that we observe a $n-$vector time series $y_t$ that follows a VAR(p) model,
\begin{align}\label{var1}
    y_t =   V \mu_t + A_1 y_{t-1}+\cdots + A_p y_{t-p}+  \eta_t ,
\end{align}
where $y_t$ is a $n$-dimensional time series,  $  V \mu_t$ is a deterministic polynomial time trend (or time-varying drift for integrated processes) with $\mu_t=(1,t,t^2,\dots,t^k)'$ for some $k$; $ A_1,\dots,A_p$ are $(n\times n)$, $V$ is $(n\times (k+1))$ parameter matrices, and $\eta_t$ are serially uncorrelated forecast errors (innovations). We will assume that $p$ is finite and known.
In addition to the main time series $y_t$,  applications frequently feature external series $Z_t$ that are not part of the VAR model.
In this paper, we consider an external series $Z_t$, which is not part of the VAR model \eqref{var1} but is used as a volatility-based instrumental variable (VIV) for structural identification through additional moment restrictions.

 \begin{assumption} Regularity of innovations and VIV,
\label{ass:innovations}
\begin{enumerate}

    \item $\eta_t$ is a martingale difference sequence (MDS), i.e. $\forall$ \(t\) \(\mathbb{E}[\eta_t|\mathcal{F}_{t-1}] = 0\). The filtration $\mathcal{F}_{t-1}$ includes sigma algebra $\sigma(y_{t-1},y_{t-2},...)$.
    \item $(\eta_t,Z_t)$ is a  stationary sequence with $\alpha$-mixing coefficients satisfying $\alpha(j)= O(j^{-3(1+\epsilon)/\epsilon})$  for some $\epsilon>0$.
    \item  $\E||\eta_t||^{8+2\epsilon}<\infty$ and $\E |Z_t|^{4+\epsilon}<\infty$ for some $\epsilon>0$.
    \item $\Sigma_\eta = \E (\eta_t \eta'_t)$ is a full-rank matrix.
    \item $\E Z_t  \eta_{t}=0$ and $\E Z_t \eta_t\eta'_{t-\ell} = 0$  for all $\ell\geq 1$.
\end{enumerate}
\end{assumption}
This set of assumptions is commonly imposed in the literature.
Assumption \ref{ass:innovations}.1  implies that innovations are white noise, which is required for the consistency of LS estimators in stationary VAR models.
Assumption \ref{ass:innovations}.2 assumes that values  $(\eta_t,Z_t)$ and  $(\eta_{t-j},Z_{t-j})$  are  approximately independent for sufficiently large lags $j$ which is required for asymptotic normal approximations.
It allows for many stationary models of conditional heteroscedasticity.\footnote{See, for example, Assumption 2.1.(i-iii) in \cite{bruggemann2016inference}.
This assumption, for example, allows for asymmetric conditional heteroscedasticity models like the ones considered in \cite{rabemananjara1993threshold}.
Such models are ruled out by the popular residual wild bootstrap methods \citep{gonccalves2004bootstrapping}.} For simplicity, in this paper we do not consider unconditional trends in heteroscedasticity. Assumptions~\ref{ass:innovations}.2 together with \ref{ass:innovations}.3 are also sufficient for the consistency of the HAC estimators.
 The remaining two components, Assumptions~\ref{ass:innovations}.4 and \ref{ass:innovations}.5, are instrumental for structural identification and the asymptotic normality of the structural impulse responses.

Consider $n\times n$ matrices $C_h$ of $h$-step-ahead causal impulse responses with components
\begin{equation}\label{eq:IRFdefinition}
  C_{h,i,j}\bydef\E(y_{i,t+h}|\eta_t=e_j,y_{t-1},y_{t-2},...)-\E(y_{i,t+h}|\eta_t=0,y_{t-1},y_{t-2},...),
\end{equation}
where $e_j\in R^n$ is a vector of zeros with the exception of the $j-$th position, which is equal to 1.
We will refer to coefficients $C_{h}$ as \emph{reduced-form impulse responses} to distinguish them from the structural impulse responses to the structural shocks of interest.
Coefficients $C_{h}$ are always identified under the  martingale difference assumption.
In the case where $y_t$ is a weakly stationary vector series,  coefficients $C_{h}$ coincide with the MA coefficients of the $VMA(\infty)$ representation of $y_t$ \citep{wold1938study,zasuhin1941theory,whittle1953analysis}.
In the general and possibly  nonstationary  case, it is a priori unclear if the conditional expectations of the right hand side of \eqref{eq:IRFdefinition} depend on $t$.
To show that $C_{h}$ are time invariant, we can   use an equivalent representation of   $C_h$
\begin{equation*}
  C_{h,i,j}\bydef\E(\eta^{(h)}_{i,t+h}|\eta_t=e_j,y_{t-1},y_{t-2},...)-\E(\eta^{(h)}_{i,t+h}|\eta_t=0,y_{t-1},y_{t-2},...),
\end{equation*}
where for each integer pair $(s,h)$ the vector series $\eta_{s}^{(h)}$ are defined as errors in the $h$-step ahead version of equation \eqref{var1},
\begin{align}\label{varForecast}
    y_{s} &=  V^{(h)} \mu_{s-h} + \sum_{i=1}^p A_i^{(h)} y_{s-h-i} + \eta_{s}^{(h)}.
\end{align}
The VAR model \eqref{var1} implies an explicit representation $\eta_{s}^{(h)}=\sum_{h'=0}^h C_{h'}    \eta_{s-h'}$ with
 $C_0=I$, $  C_h = \sum_{\ell=1}^h   A_\ell  C_{h-\ell}$.\footnote{
Assuming $  A_\ell=0$ for $\ell>p$, matrices $(V^{(h)},A_1^{(h)},\dots,A_p^{(h)})$ in \eqref{varForecast} are  polynomial functions of $(V,A_1,\dots,A_p)$ that can also be computed recursively.}
By Assumption \ref{ass:innovations}.2, $\eta_{s}^{(h)}$ are stationary, and thus $C_h$
 is indeed time invariant and well defined even for nonstationary VAR models.

Reduced-form impulse response  $C_h$ can be estimated directly using the local projection estimator of \cite{jorda2005estimation}.
Although $C_h$ is well defined for arbitrary values of lag polynomial coefficients $A(L)$, for the purpose of our novel local projections inference framework, we  restrict our attention to
\begin{assumption}\label{ass:SSW} $det(I_n-A_1x-\dots-A_px^p)\neq0$  for $|x|<1$.
\end{assumption}
  Assumption \ref{ass:SSW} allows for an arbitrary number of unit roots, including complex unit roots, and cointegrating relations between $y_t$. It only rules out explosive roots.
We will show in the next section that one can apply our inference procedure to local projections estimated with or without a trend in levels, without preliminary seasonal adjustment.
It can be done without any pretests for trend stationarity, the (absence of) cointegration, or unit roots.
We allow both a polynomial trend as in \cite{sims1990inference} and complex unit roots as in \cite{tsay1990asymptotic}.
Under this assumption, the least squares estimator for the coefficients in equations \eqref{var1} and \eqref{varForecast} is consistent and has well-defined distribution limits.

When conducting counterfactual analysis, it is common to impose an underlying linear structure on innovations, $\eta_t = B \varepsilon_t$.
The $n-$vector of structural shocks $\varepsilon_t$ collects primal economic factors such as demand and supply shocks or unexpected policy changes.
Following the standard approach, we assume that the components of $\varepsilon_t$ are uncorrelated with each other.
Suppose we are interested in the causal impulse response functions of the structural shock $\varepsilon_{1,t}$,
\begin{equation*}
\beta_{h}\bydef \E(y_{t+h}|\varepsilon_t=e_1,\mathcal F_{t-1})-\E(y_{t+h}|\varepsilon_t=0,\mathcal F_{t-1}).
\end{equation*}
Immediate computation implies $\beta_{h} =  C_{h}\beta$,
where $\beta$ denotes the first column of $B$.
When $\beta_1\neq 0$, we denote the corresponding unit-impact-normalized structural impulse response by $\psi_h\bydef\beta_h/\beta_1=C_h\beta/\beta_1$.
The vector $\beta$ is not uniquely identified based on \eqref{var1} alone, but it can be identified using some external moment restriction that depends on the distribution of $(\eta_t,Z_t)$.
In this paper, we consider two cases. Vector $\beta$ could be $\beta_\Sigma$, a  column in the Cholesky decomposition of $\Sigma_\eta = \E \eta_t  \eta_{t}^\prime $, or it can be  a function of VIV,  $\beta_Z =  \E(\eta_{t} \eta_{1t}-\E\eta_{t} \eta_{1t})(Z_t - \E Z_t)$. The latter moment restriction will be considered   in Section~\ref{sec:HetIV}.

\subsection{Local projection estimators for reduced-form IRFs}
In the previous section, we saw that the structural IRFs  $\psi_h$ for $h=1,\dots,H$ can be represented as functions of the vector of parameters
\begin{equation*}
  \theta\bydef vec(  C_{1},\dots, C_{H}, \Sigma_\eta,\beta_Z ).
\end{equation*}
First, consider   $ C_{1},\dots, C_{H}, \Sigma_\eta$.
Following \cite{jorda2005estimation}, matrices $C_{1},\dots, C_{H}$ can be estimated using least squares regression (local projection) for each $h=1,\dots,H$
    \begin{align}\label{eq:hlagVAR}
    y_{t+h} = C_h y_{t} + \tilde V^{(h)} \mu_{t} + \sum_{i=1}^p  \tilde A_i^{(h)}  y_{t-i} +  u_{t+h}^{(h)},
\end{align}
 where matrices $\tilde V^{(h)}, \tilde A_1^{(h)},\dots,\tilde A_p^{(h)}$ are functions of  $ V ,  A_1 ,\dots,A_p$; $ u_{t+h}^{(h)}$  is a moving average  of $ \eta_{t+1},\dots,\eta_{t+h}$.
For the purpose of  joint inference, it is convenient to use an alternative equivalent representation based on the Frisch-Waugh-Lovell theorem \citep{lovell1963seasonal}.
\begin{thm}\label{thm:LPconsistency}
    Suppose Assumptions \ref{ass:innovations} and \ref{ass:SSW} hold.
The local projection estimators $\widehat C_h$ are consistent for $ C_h$ and have a representation
   $  \widehat C_h' =    (\sum_{t=p}^{T-h} \hat \eta_t  \hat \eta_{t}^\prime )^{-1}\sum_{t=p}^{T-h} \hat \eta_t  (\hat\eta_{t+h}^{(h)})',    $
where $\hat\eta_{t}^{(0)}\bydef \hat \eta_{t}  $ and $ \hat\eta_{t+h}^{(h)}$ are regression residuals from $h$-step ahead autoregressions \eqref{varForecast}, i.e.
       $\hat\eta_{t+h}^{(h)} =   y_{t+h} -  \hat V^{(h)} \mu_t - \sum_{i=1}^p \hat A_i^{(h)} y_{t-i}.$

Similarly, the matrix $\Sigma_\eta$ has a consistent estimator
   $       \hat \Sigma_\eta =       \frac{1}{T-p} \sum_{t=p}^{T} \hat \eta_t  \hat \eta_{t}^\prime.$
\begin{proof}
    See Appendix \ref{sec:app_proof_LP}
\end{proof}
\end{thm}
Theorem \ref{thm:LPconsistency}  claims that local projection estimators $\widehat C_h$ are equivalent to estimating regressions between $\hat\eta_{t+h}^{(h)}$ and $\hat\eta_{t}$, which have an interpretation of detrended, seasonally adjusted, and ``prewhitened'' $ y_{t+h}$ and $ y_{t}$.\footnote{It is common practice to detrend the data before IRF analysis. Some of the recent recommendations are summarized in \cite{hamilton2018you}. Prewhitening is a process of transforming autocorrelated time series into a white noise form, which has a long history in spectral estimation that goes back to \citep{press1956power}. For more recent use, see \cite{andrews1992improved}. Note also that a VAR(p) model with sufficiently large $p$ (for example, $p\geq12$ for monthly data) also nests seasonal unit roots, which automatically adjusts for the stochastic seasonal component.}
As a result, the consistency of $\widehat C_h$ essentially follows from the ergodic properties of
 innovations $\eta_t$ even for nonstationary $y_t$.
 The practical implication of this result is that it is unnecessary
  to detrend, seasonally adjust,   perform first-differencing, or use an error-correction model representation of $y_t$ before applying LP estimators.
  All the corresponding pretest procedures can thus be avoided, which radically simplifies inference.
 As long as the data generating process can be represented in VAR(p) form \eqref{var1}, the local projection estimator automatically transforms the data into an ergodic  form.

A consistent estimate of $ \Sigma_\eta$ allows one to compute structural impulse responses under recursive (Cholesky) identification, $C_h\beta_\Sigma$, where $\beta_\Sigma$ is a column of the Cholesky decomposition of $\Sigma_\eta$.
This classical scheme was used in \cite{sims1980comparison,christiano1999monetary} among others.
Alternatively, one can use external moment conditions $\beta_Z$ for structural IRFs.
The following theorem shows that the feasible estimator of $\beta_Z$  uses residuals $\hat\eta_{t}$ defined in Theorem \ref{thm:LPconsistency}.
\begin{thm}\label{thm:IVconsistency}
Under Assumptions \ref{ass:innovations}\&\ref{ass:SSW},  $\hat\beta_Z\bydef\frac{1}{T-p}   \sum_{t=p}^T  (\hat\eta_{t}\hat\eta_{1t}-\overline{\hat\eta_{t}\hat\eta_{1t}})(Z_t - \bar Z_T) \pTo \beta_Z$.

\end{thm}
\begin{proof}
    See Appendix Section \ref{sec:app_proof_IV}.
\end{proof}
Theorem \ref{thm:IVconsistency} uses the first component of the vector of forecast errors $\hat\eta_{1t}$ without loss of generality.
We will discuss the interpretation and additional normalization procedure for $\beta_Z$ in Section \ref{sec:HetIV} after the exposition of the joint inference procedure.

\section{Analytical and bootstrap inference on IRFs}\label{sec:estimation}

\subsection{Multiple coefficient inference for nonstationary data}

The most popular procedure for inference on multiple reduced-form
 coefficients in SVARs is the residual bootstrap \citep{gonccalves2004bootstrapping} or the residual-block bootstrap, a version for general heteroscedasticity innovations \citep{bruggemann2016inference}.
The proof of validity of these procedures assumes stationarity $y_t$.
The extension to nonstationary data is non-trivial, as the following example shows.

\begin{example}\label{exa:ar1}
   Consider a special case of model \eqref{var1}, a univariate AR(1) model
   \begin{equation}
       y_t= A_1 y_{t-1}+\eta_t,\label{eq:ar1}
   \end{equation}
with $\eta_t$ being i.i.d. Gaussian sequence with zero mean and unit variance. The impulse response sequence $C_h$ has explicit form $C_h=A_1^h$. It is well known that the OLS estimator of $ A_1$ has a non-Gaussian Dickey-Fuller (DF) distribution when $ A_1=1$ \citep{dickey1979distribution}. The common bootstrap procedures, such as residual bootstrap, are not valid for inference on $A_1$ and therefore functions of it when $A_1=1$. However, as proposed in \cite{inoue2002bootstrapping}, one can add a time trend or add a redundant lag to the model \eqref{eq:ar1},
   \begin{equation}
       y_t= A_1 y_{t-1}+A_2 y_{t-2}+\eta_t, \label{eq:ar2}
   \end{equation}
where the restriction $A_2=0$ is not imposed. The addition of such an irrelevant regressor results in OLS estimates that have an asymptotic Gaussian marginal distribution of individual components but are not jointly Gaussian.  \cite{sims1990inference} showed that there exist ``forbidden'' linear combinations of OLS estimates, corresponding to the so-called generalized cointegrating vectors.
Indeed, one can rewrite the model in ADF form
   \begin{equation*}
       y_t=( A_1+ A_2) y_{t-1} - A_2   \Delta y_{t-1}+\eta_t.
   \end{equation*}
Then, the OLS estimator of the corresponding population parameter $f(A) \bydef  A_1+ A_2$ will be super-consistent and will have a Dickey-Fuller (DF) asymptotic distribution, while $ \hat A_2$ will be asymptotically Gaussian.
By the linearity property of the OLS estimators, $\hat A_1=f(\hat A) -\hat A_2$. Due to the super-consistency of $f(\hat A)$ for 1, the asymptotic distribution of $\hat A_1$ is dominated by the asymptotic Gaussian distribution of $\hat A_2$.
As a result, both $\hat A_1$ and $\hat A_2$ have marginal distributions for which the residual bootstrap is valid.
Similarly, IRF estimators based on the  augmented model \eqref{eq:ar2}
\begin{align*}
     \widehat C_1 - C_1 &= \hat A_1 -1 =  O_P(\frac{1}{\sqrt{T}})  , \\
     \widehat C_2 - C_2&= \hat A_2 +\hat A^2_1 -1 =   (\hat A_1-1)+  (f(\hat A)-1) +    (\hat A_1-1)^2= \hat A_1-1 +  O_P(\frac{1}{{T}})= O_P(\frac{1}{\sqrt{T}}),
\end{align*}
have marginal asymptotic Gaussian distributions with the standard rate of consistency (since the asymptotic distribution of $\widehat C_2$ is dominated by the asymptotically Gaussian component $  \hat A_1-1 $).
However, the IRF vector is not jointly Gaussian, since the following difference has a non-Gaussian distribution
\begin{align*}
     \widehat C_2  -     \widehat C_1   =  (f(\hat A)-1)+   (\hat A_1-1)^2 = O_P(\frac{1}{T}).
\end{align*}
This difference is asymptotically equal to a linear combination of correlated DF random variables and $\chi^2(1)$ and is super-consistent for $C_2-C_1=0$. As a result, the residual bootstrap will not approximate the distribution of the difference $  \widehat C_2  -  \widehat C_1 $ even in the pointwise sense. \qed
\end{example}

There are ways to sidestep this problem at the cost of losing efficiency.
The standard textbook proposal is to test for unit roots. If a unit root is detected,  the data is converted to stationary form, which requires conservative adjustments to confidence intervals for pre-testing \citep{cavanagh1995inference}.
If  stationary data is transformed into first differences without pretesting, there is a potential efficiency loss due to the introduction of a non-invertible moving average component from over-differencing the innovations \citep{plosser1977estimation}.

Alternatively, one can estimate the model in levels and then use the projection of the product of marginal Gaussian confidence sets for $C_1,\dots,C_H$ with Bonferroni correction \citep{inoue2023significance}.
This intuitive approach also results in unnecessarily wide intervals since the  Bonferroni-type bounds are more conservative than  the sup-$t$ bounds \citep{montiel2019simultaneous}.
The sup-$t$ approach, however, requires an estimator of the covariance between the components of the IRF vector or the ability to jointly bootstrap the entire vector of estimators.

Yet another approach is to replace estimates of redundant parameters ($A_2$ in Example \ref{exa:ar1}) with their theoretical value, i.e., 0, when computing $f(\hat A)$ \citep{toda1995statistical}.
This approach is called lag-augmentation, as it requires one to add as many lags to the model as the potential number of unit roots in the model (as many as $p$ extra lags  under Assumption \ref{ass:SSW}).
This approach results in an asymptotic Gaussian distribution for smooth functions of $A(L)$ such as $f(\hat A)$, but it leads to very wide confidence intervals at large horizons $h$, even in the simplest case of AR(1) with one unit root, as shown in \cite{montiel2021local}.

In their pioneering work, \cite{montiel2021local} proved the asymptotic Gaussian marginal distribution for LP estimators $C_h$ for each individual horizon in the presence of unit roots and demonstrated that LP estimators have significantly shorter confidence sets than those based on lag-augmented recursive VAR estimates.
In the following subsections, we establish the joint asymptotic normality of LP estimators across horizons under Assumptions \ref{ass:innovations}-\ref{ass:SSW} and prove the validity of a novel dependent wild bootstrap procedure.
This inference procedure provides a way to conduct simultaneous inference on structural IRFs for all horizons without pre-tests.
Moreover, the joint inference procedure suggests more efficient and robust weighted estimators that fully utilize the underlying VAR model structure.

\subsection{Analytical HAC-inference}

Using representations from Theorems \ref{thm:LPconsistency}-\ref{thm:IVconsistency} for components of $\hat \theta$ as approximate averages of   stationary and mixing series, we can now apply the Central Limit Theorem (CLT) to derive the joint asymptotic distribution of  $\hat \theta$. We introduce relevant scores and their feasible analogs
\begin{equation*}
    \xi_t  \bydef \begin{pmatrix}
         vec(  \Sigma_\eta^{-1}  \eta_t \eta_{t+1}^\prime    )\\
         \cdots\\
         vec(  \Sigma_\eta^{-1} \sum_{s=0}^{H-1}  \eta_t\eta_{t+H-s} ^\prime  C_s ^\prime  )\\
                  vec( \eta_t  \eta_{t} ^\prime - \Sigma_\eta )\\
          (\eta_{t} \eta_{1t}-\E\eta_{t} \eta_{1t})(Z_t - \E Z_t)-  \beta_Z
         \end{pmatrix}\text{ and } \hat\xi_t \bydef \begin{pmatrix}
         vec(  \hat\Sigma_\eta^{-1}  (\hat\eta_t\hat\eta'_{t+1}-\overline{\hat\eta_t\hat\eta'_{t+1}}) )\\
         \cdots\\
         vec(  \hat\Sigma_\eta^{-1}\sum_{s=0}^{H-1}  (\hat\eta_t\hat\eta'_{t+H-s}-\overline{\hat\eta_t\hat\eta'_{t+H-s}}) \widehat C_s ^\prime  )\\
          vec( \hat\eta_t  \hat\eta_{t} ^\prime - \hat\Sigma_\eta )\\
         ( \hat\eta_{t} \hat\eta_{1t}-\overline{ \hat\eta_{t} \hat\eta_{1t}})(Z_t - \bar Z_T)- \hat\beta_Z
         \end{pmatrix},
\end{equation*}
where
$\eta_{1t}$ is the first element of $\eta_{t}$.
\begin{thm}\label{thm:inferenceCH}
    Suppose Assumptions \ref{ass:innovations} and \ref{ass:SSW}  hold. Then
\begin{equation}
\lim_{T\to\infty} \sup_{x\in \mathbb{R}, \|a\|=1}
\left|
P\left(\sqrt{T}a'(\hat\theta-\theta)<x\right)
-
\Phi\left(
\frac{x}{\sqrt{a'\Omega a}}\right)
\right|
=0,
\label{eq:CLTmainText}
\end{equation}
with  $\Omega   \bydef \sum^{\infty}_{\ell=-\infty} \E \xi_t  \xi'_{t+\ell} $.
\end{thm}
\begin{proof}
    See Appendix Section \ref{sec:app_proof_CLT}
\end{proof}

The variance-covariance matrix $\Omega$ can be estimated using the HAC estimator along the lines in \cite{andrews1991heteroskedasticity}.
One, however,  needs  to modify the original proofs   to account for the nonstationarity of $y_t$.
Consider
\begin{equation*}
     \widehat\Omega_T  = \sum^{B_T}_{\ell=-B_T}K\Big(\frac{\ell}{B_T}\Big) \overline{  \hat \xi_t \hat \xi'_{t+\ell} }
\end{equation*}
for some kernel functions $K(\cdot)$ and bandwidth sequences $B_T$.
\begin{thm}\label{thm:hac} Suppose Assumptions  \ref{ass:innovations}-\ref{ass:SSW} hold. Then, for any bandwidth sequence satisfying $ B_T^2 /T\to 0$, we obtain $ \widehat\Omega_T\pTo\Omega$.
\end{thm}
\begin{proof}
    See Appendix Section \ref{sec:app_proof_hac}
\end{proof}
Since we are estimating $\Omega$ of a stationary mixing process $\xi_t$, we suggest using the conventional rule for HAC inference, $ B_T= \{c \sqrt[3]{T}\}$, which is appropriate for the Bartlett kernel.\footnote{In our simulations, we use  $c=0.75$  as the rule of thumb \citep{stock_watson_2021_intro_econometrics}.}

  Recent literature emphasizes that when data is very persistent, HAC estimators do not perform well \cite[see, for example][]{lazarus2021size}.
  In particular, in the near-to-unity asymptotics, one should use a very large bandwidth $B_T$ to control the size of $t$ tests.
  This is not the case under our Assumption \ref{ass:innovations}.2 (weak dependence) for the innovations $\eta_t$.
If, in addition, the innovations are two-sided MDS, then \cite{montiel2021local} showed that one can use $B_T=0$ for inference on individual coefficients $\widehat C_h$.
However, even under these stronger assumptions, joint inference on the entire vector $\theta $ still requires $B_T$ to grow with the sample size to account for the potential presence of conditional heteroscedasticity of unknown form.

\subsection{Dependent wild bootstrap}

A HAC variance estimator could be used for inference on functions of $\theta$, such as smoothed local projection estimators, within a plug-in delta-method framework.

From a computational perspective, though, it may be preferable to use simulation-based methods.
The dependent wild bootstrap (DWB) approach of \cite{shao2010dependent} can be used as a numerical implementation of the delta method that allows making joint inference on multiple functions of $\theta$ simultaneously and without explicit analytical derivatives.
DWB is also computationally preferable to generating random draws directly from the asymptotic normal distribution with estimated $\Omega$, which can be high-dimensional \citep{chang2023testing}.
A specific  implementation is presented in the following algorithm.

\begin{algorithm}
\caption{Dependent wild bootstrap}\label{algo:DWB}

 \begin{enumerate}
    \item[Step 1.]  Select bandwidth $B_T$ according to Theorem \ref{thm:hac}. For example, $B_T= 0.75 \sqrt[3]{T} $.
   \item[Step 2.] Generate i.i.d. $\zeta^{s}_t\sim N(0,\frac{1}{B_T})$ for $t=-B_T+1,-B_T+2,...,0,1,...T$.
     \item[Step 3.] Recursively compute $MA(B_T)$ process  $\nu^s_t=\zeta^{s}_t+\zeta^{s}_{t-1}+...+ \zeta^{s}_{t-B_T+1}$ , $t=1,...T$.
     \item[Step 4.]  Compute        $ \hat\theta^s  = \hat\theta  + (T-H-p)^{-1}\sum_{t=p}^{T-H}\hat\xi_t \nu^s_t.$
\end{enumerate}

\end{algorithm}

Algorithm \ref{algo:DWB} is consistent for estimating the asymptotic distribution of $\hat\theta$ using bootstrap draws $ \hat\theta^s$ as the following theorem shows.

 \begin{thm}\label{thm:bootstrap}
 Suppose Assumptions \ref{ass:innovations}-\ref{ass:SSW} hold, and $B_T^2 /T\to 0$. Then $\hat{\theta}^s$ from Algorithm~\ref{algo:DWB} satisfies
 \begin{equation*}
      \lim_{T\to\infty} \sup_{x\in R, \|a\|=1}  |P(\sqrt{T}a'(\hat \theta^s -\hat \theta )<x|y_1,...,y_T)-P(\sqrt{T}a'(\hat \theta - \theta )<x) |=0.
 \end{equation*}

 \end{thm}
 \begin{proof}
     See Appendix Section \ref{sec:app_proof_bootstrap}.
 \end{proof}
To our knowledge, this is the first joint bootstrap resampling procedure for LP estimators  that allows  for nonexplosive VARs with unit roots, polynomial trends, cointegration, and weakly dependent conditional heteroscedasticity without pretesting.
The validity of  DWB is based on the fact that it mimics the asymptotic Gaussian distribution of $\hat \theta $ with a covariance matrix that coincides with the HAC estimator given in Theorem \ref{thm:hac}.
This argument does not impose any parametric model on $\eta_t$ and $Z_t$ and is generally applicable under the maintained Assumption \ref{ass:innovations}.2 (weak dependence of innovations).
Unlike the popular recursive residual-based bootstrap procedures  \citep{gonccalves2004bootstrapping,inoue2020uniform}, we allow for asymmetric conditional heteroscedasticity in $\eta_t$ and for a general (multivariate) VAR model with any number of unit roots.

It should be possible to replace the DWB procedure with a block multiplier bootstrap or other resampling methods for dependent data.
However, since    $\hat\xi_t$   depends on possibly nonstationary processes, $y_t$, the validity of these alternative methods requires independent analysis.

 The delta method yields the following immediate corollary to Theorem \ref{thm:bootstrap}  which will be useful in applications of DWB.
\begin{cor}\label{cor:DeltaMethod}
 Suppose that a continuously differentiable mapping $f(\cdot)$ maps $\R^ {dim(\theta)}$ to $\R^k$ such that the Jacobian $\nabla f(\theta)$ at the true value $\theta$ has a rank equal to $k$.
Then under the assumptions of Theorem \ref{thm:bootstrap}, we have
 \begin{align*}
      \lim_{T\to\infty} &\sup_{x\in \R,a\in \R ^k ,\|a\|=1}  |P(\sqrt{T}a'(f(\hat \theta^s) -f(\hat \theta) )<x|y_1,...,y_T)-P(\sqrt{T}a'(f(\hat \theta) - f(\theta) )<x) |=0,\\
\lim_{T\to\infty} & \sup_{x\in \mathbb{R}, a\in \R ^k , \|a\|=1}
\left|
P(\sqrt{T}a'(f(\hat \theta) - f(\theta) )<x)
-
\Phi\left(
\frac{x}{\sqrt{a' \nabla f(\theta)\Omega \nabla f(\theta)' a}}
\right)
\right|
=0.
 \end{align*}

\end{cor}

\begin{rem}
    One can use the delta method to     estimate the forecast-error variance decomposition for each $r=1,\dots,n$ and $H$ using the formula
 \begin{equation*}
     \widehat {FEVD}_r(H) = \frac{\sum_{h=0}^{H-1}  (e'_r\widehat C_h\hat\beta_Z)^2 }{(\hat\beta_Z^\prime \hat\Sigma^{-1}_\eta \hat\beta_Z)\sum_{h=0}^{H-1}  e'_r \widehat C_h \hat \Sigma_\eta \widehat C'_h e_r }.
 \end{equation*}
 Corollary \ref{cor:DeltaMethod} shows that DWB is  a valid inference procedure for this estimator.
\end{rem}

\section{Applications of dependent wild bootstrap } \label{sec:applicationDWB}
Theorem \ref{thm:bootstrap} has multiple practical corollaries. In this section, we consider applications for efficient smoothing of local projections, simultaneous confidence bounds on impulse response functions (IRFs), inference on structural IRFs identified through external VIV robust to weak identification, and specification tests.

\subsection{Efficient estimation using minimum distance approach}\label{subsec:MD}

Local projection estimators of $C_h$ do not impose the parametric restrictions implied by the underlying VAR model. One can explore a potential efficiency gain from constructing smoothed (restricted) LP estimators using the minimum distance (MD) approach. This idea was explored in the literature \citep{jorda2011estimation,barnichon2018functional} under the assumption of stationary $Y_t$.\footnote{\cite{barnichon2019impulse} introduced smooth local projections by adding a ridge-type penalty that enforces the smoothness of the estimated impulse response function across horizons. Their focus is primarily on point estimation, so they adopt a heuristic approach to constructing confidence intervals.}
Using DWB, which is robust to unit roots and time trends, we can now construct the optimal MD estimator of IRFs that  exploits the VAR($p$) recursion restrictions that link $C_h$ across horizons.
Because such a minimum-distance estimator is a smooth transformation of $\hat{\theta}$ under the maintained regularity conditions, Corollary~\ref{cor:DeltaMethod} directly yields valid pointwise and simultaneous inference for the resulting smoothed impulse responses $\widehat C_h^{MD}$, even under conditional heteroscedasticity and weak dependence.

This approach has several practical advantages: (i) MD improves the efficiency of $\widehat C_h^{MD}$ at short horizons;  (ii) it allows one to extrapolate IRFs to arbitrarily long horizons using the VAR structure; (iii) it provides estimates of lag polynomials $A_1,\dots,A_p$ as a byproduct.

To construct the MD estimate using $H'$ LP horizons ( $p\leq H'\leq H$), take any positive definite matrix $\hat W $ that converges in probability to a  positive definite matrix $W$ of dimension $(H'n^2)\times (H'n^2) $ and minimize
   \begin{equation}
        ( \hat A^{MD}_1,\dots, \hat A^{MD}_p) =  \underset{(A_1,\dots,A_p)\in \mathbb{R} ^{p n^2} }{\operatorname{argmin}}  (   g(A_1,\dots,A_p)-\hat g^{LP})' \hat W  (   g(A_1,\dots,A_p)-\hat g^{LP}),\label{eq:MD}
\end{equation}
where $\hat g^{LP} = vec (\widehat C_1,\dots,\widehat C_{H'})$ denotes the vector of local projection estimators and $g(A_1,\dots,A_p) = vec (C_1(A_1,\dots,A_p),\dots,C_{H'}(A_1,\dots,A_p))$ denotes the corresponding vector of recursively computed IRF coefficients for any values of the matrices $(A_1,\dots,A_p)$.
From the theory of minimum distance estimators (see, for example, \cite{hansen2022econometrics}), the optimal choice of $\hat W$ corresponds to the inverse of the long-run covariance matrix, as given in Theorem \ref{thm:inferenceCH} earlier in this paper.
In small samples or in cases where $H'$ is large, however, the HAC estimator given in Theorem \ref{thm:hac} may result in a non-invertible matrix.
In this paper, we instead use a diagonal matrix $\hat W$ with the vector of inverse variances of the components of $\hat g^{LP} $ on the diagonal.
These variances are estimated using the DWB sample.

Based on the minimum distance estimates $(\hat A^{MD}_1,\dots,\hat A^{MD}_p)$,
we can compute the smoothed impulse responses for any horizon $H$ using the recursion
$    \widehat C^{MD}_0 \bydef I$,$\widehat C^{MD}_h \bydef \sum_{\ell=1}^{\min(p,h)} \hat A^{MD}_\ell \widehat C^{MD}_{h-\ell}$, for $ h=1,\dots,H$.

Under the correct VAR($p$) specification, the limit of \eqref{eq:MD} coincides with the true parameters $A_1,\dots,A_p$.\footnote{See Remark \ref{rem:HausmanTest} in the next subsection for a test of misspecification. }
Moreover, the estimator $(\hat A^{MD}_1,\dots,\hat A^{MD}_p)$ is a smooth transformation of $\hat \theta$. So, by  Corollary \ref{cor:DeltaMethod}, the DWB from Theorem \ref{thm:bootstrap} remains valid for $\widehat C^{MD}_i$.
By implication, we can use it for both pointwise inference and simultaneous inference on $C_h$ across horizons.

The optimization program in equation \eqref{eq:MD} can be solved using quasi-Newton methods. The gradient of the objective function can be computed analytically. As the initial point for parameters $(A_1,\dots,A_p)$, one can use the   corresponding LS estimator for model \eqref{var1}. To compute the DWB sample $C^{s,MD}_{h}$, one can use $A^{MD}_{1},\dots,A^{MD}_{p} $ as the initial point in \eqref{eq:MD} and compute only a small number of quasi-Newton steps \citep{robinson1988stochastic,andrews2002equivalence}.

In the special case where the number of smoothed LP coefficients $H'=p$, the objective function \eqref{eq:MD} can be set exactly to $0$.
For each $i=1,\dots,p$  compute
   \begin{equation*}
       \hat A^{MD}_i =    \widehat C_i - \sum_{\ell=1}^{i-1}   \hat A^{MD}_\ell  \widehat C_{i-\ell}
   \end{equation*}
   assuming $\widehat C_{0} = I$ and $\hat A^{MD}_\ell=0$ for $\ell\leq0$.
This algorithm, in the case of $p=1$, further simplifies to  $\hat A^{MD}_1 = \widehat C_1$.
In other words, if  $p=1$, the local projection estimator $\widehat C_1$ from \eqref{eq:hlagVAR} coincides with the lag-augmented VAR estimator of IRF proposed in \cite{toda1995statistical}.
Since any VAR(p) model can be represented as VAR(1) by expanding the vector of observations $y_t$ to include the $p$ lags, one can see that MD with the number of smoothed horizons $H'=p$ nests lag-augmented VAR estimators that are robust to $p$ unit roots.
Although such estimators can be used for valid inference of the IRF, they may result in very wide confidence bands in the nonstationary case.
The choice of smoothed horizons $H'$ is explored later using simulations in Section \ref{subsec:MC}.

To address potential  efficiency loss from the  choice of suboptimal diagonal $\hat W$, we employ an optimal component-wise linear combination of $\widehat C_{h,i,j}$ and $\widehat C^{MD}_{h,i,j}$. For each  $(h,i,j)$, let
the weighted estimator be
\begin{equation}
    \widehat C^{\omega}_{h,i,j}
    =
    \widehat \omega_{h,i,j} \widehat C^{MD}_{h,i,j}
    +
    (1-\widehat \omega_{h,i,j})\widehat C_{h,i,j},\label{eq:WMD}
\end{equation}
where the scalar weight $ \widehat \omega_{h,i,j}$ is computed from the bootstrap variances and covariances of the two estimators,
\begin{equation*}
    \widehat \omega_{h,i,j}
    =
    \frac{
        \widehat{\operatorname{Var}}( \widehat C^{s}_{h,i,j})
        -
        \widehat{\operatorname{Cov}}(\widehat C^{s}_{h,i,j},\widehat C^{s,MD}_{h,i,j})
    }{
        \widehat{\operatorname{Var}}(\widehat C^{s,MD}_{h,i,j})
        +
        \widehat{\operatorname{Var}}(\widehat C^{s}_{h,i,j})
        -
        2\widehat{\operatorname{Cov}}(\widehat C^{s,MD}_{h,i,j},\widehat C^{s}_{h,i,j})
    }.
\end{equation*}
 The weighted estimator $ \widehat C^{\omega}_{h,i,j}$ shrinks the unrestricted LP estimate toward the smoother VAR-recursive estimate when the latter is empirically less variable (typically at short horizons)
, while retaining the classical LP estimate when the smoothing restriction appears comparatively noisy (typically at long horizons).
Whenever both raw and smoothed LP estimators are asymptotically normal and consistent by Corollary \ref{cor:DeltaMethod}, this optimal weighted combination results in an asymptotic variance that is smaller than or equal to the minimal variance of the two alternatives.\footnote{
Asymptotic normality of smoothed estimators may fail at longer horizons if  the largest characteristic root of the lag polynomial  is close to zero.
In that case,  Corollary \ref{cor:DeltaMethod} is not applicable to smoothed IRF $\widehat C^{MD}_{h,i,j}$, since the corresponding Jacobian is rank deficient.
Instead, $\widehat C^{s,MD}_{h,i,j}$ becomes super-consistent and has a non-standard asymptotic distribution.
The variance of $\widehat C^{MD}_{h,i,j}$ becomes asymptotically negligible when compared to $\widehat C_{h,i,j}$.
To address this degeneracy, we regularize the weights $ \widehat \omega_{h,i,j}$ by replacing the variance estimate $  \widehat{\operatorname{Var}}(\widehat C^{s,MD}_{h,i,j})$ with   $   \max \{\widehat{\operatorname{Var}}(\widehat C^{s }_{h,i,j})c,\widehat{\operatorname{Var}}(\widehat C^{s,MD}_{h,i,j})\}$ for some $c\leq 1$ and rescaling the covariance estimator accordingly. }

\subsection{Joint inference using sup-$t$ critical values} \label{subsec:supt}

Another  use of the joint inference procedures outlined in Theorems \ref{thm:inferenceCH}-\ref{thm:bootstrap} is simultaneous confidence bounds that cover the entire IRF for a particular time series over all horizons with a given nominal probability.
The problem has been extensively studied in stationary VAR models \citep[see, for example,][]{jorda2009simultaneous,inoue2016joint,montiel2019simultaneous,inoue2022joint, arias2023uniform}.
 DWB can be used to construct sup-$t$ confidence bounds, which are robust to unit roots.
The following algorithm can be used for joint inference on any smooth functions $f_h(\theta)$ (for example,  IRF($h$) over multiple  $h=1,\dots,H$):

\begin{algorithm}
\caption{sup-$t$ simultaneous confidence bounds}
\label{algo:supt}
\begin{enumerate}

\item[Step 1.]
Use Algorithm \ref{algo:DWB} to compute
$\{\hat{\theta}^s\}_{s=1}^S$.

\item[Step 2.]
Estimate the standard error for each
$h=1,\dots,H$:
\begin{equation*}
\hat\sigma_h =
\frac{
q(f_h(\hat{\theta}^s),\Phi(1))
-
q(f_h(\hat{\theta}^s),1-\Phi(1))
}{2},
\end{equation*}
where $q(f_h(\hat{\theta}^s),\alpha)$ is the order statistic of
$\{f_h(\hat{\theta}^s)\}_{s=1}^S$
corresponding to the quantile $\alpha$.

\item[Step 3.]
For each bootstrap draw $s=1,\dots,S$, compute
$\tau^s =
\sup_{h=1,\dots,H}
\frac{
|f_h(\hat{\theta}^s)-f_h(\hat{\theta})|
}{
\hat\sigma_h
}.$

\item[Step 4.]
Compute $ CV_\alpha = q(\tau^s,1-\alpha).$

\item[Step 5.]
For all $h$, compute
$
\left[
f_h(\hat\theta)-\hat\sigma_h CV_\alpha,
\;
f_h(\hat\theta)+\hat\sigma_h CV_\alpha
\right].
$

\end{enumerate}
\end{algorithm}

In large samples, by the Corollary \ref{cor:DeltaMethod} (delta method), the bootstrap distribution of $f_h(\hat{\theta}^s)$ is approximately Gaussian.
As a result, the difference between the $q(f_h(\hat{\theta}^s),\Phi(1))$ and $q(f_h(\hat{\theta}^s),1-\Phi(1))$ quantiles is converging in probability to 2 standard deviations of $f_h(\hat{\theta})$ as $T,S\to\infty$.
Step 2 can be replaced with a conventional bootstrap standard error estimator if the sample size is sufficiently large.\footnote{ Substituting the interquartile range for the standard deviation is frequently used in the bootstrap literature \citep[e.g.][]{hansen2022econometrics}.
This procedure is robust to the potential non-existence of second moments of functions $f_h(\hat\theta)$ in small samples.}
Another advantage of using the standard error estimator based on the quantiles is that the simultaneous bounds coincide with the $68\%$ pointwise bounds if we set the critical value $CV_\alpha=1$.

\begin{rem}\label{rem:HausmanTest}
   Another application of the sup-$t$ test is as a model specification test for the $VAR(p)$ model, an analog of the \cite{hausman1978specification} test. By Corollary \ref{cor:DeltaMethod}, DWB can correctly approximate the asymptotic distribution of the difference $\widehat C_h-\widehat C^{MD}_h$, which is  a smooth transformation of $\widehat C_h$.
 The null hypothesis of correct specification is rejected if $\sup_{i=1,n;j=1,n;h=1,H;}|\widehat C_{h,i,j}-\widehat C^{MD}_{h,i,j}|$ exceeds the $(1-\alpha)$ quantile of the   DWB distribution, $\sup_{i=1,n;j=1,n;h=1,H;}|\widehat C^s_{h,i,j}-\widehat C_{h,i,j}-(\widehat C^{s,MD}_{h,i,j}-\widehat C^{MD}_{h,i,j})|$.\footnote{\cite{chernozhukov2017clt} shows that high-dimensional sup-$t$ inference works for estimators with component-wise Gaussian limits, even in the case where the number of components grows at a polynomial rate with sample size.}
\end{rem}

\subsection{Inference on structural IRFs identified using potentially weak  VIV}
\label{sec:HetIV}

The  applications of the DWB considered in the previous sections were concerned with the reduced-form parameters in SVAR models.
It can also be useful for inference on structural impulse response $C_h \beta$, where $\beta$ is identified using additional external moment restrictions.
To illustrate the core of this problem, consider a  VAR(0) model or a simultaneous equations  model
\begin{equation*}
    \eta_t =  B \varepsilon_t,
\end{equation*}
where observations $\eta_t\in \R^n$ and $\varepsilon_t$ are i.i.d. zero-mean processes with $\E \varepsilon_t \varepsilon_t'=I_n$ that are referred to as structural shocks.\footnote{The simultaneous equations model with lags was first proposed and studied in  \cite{haavelmo1943statistical,mann1943statistical}.} Suppose that we are interested in measuring the impact of the structural shock $\varepsilon_{1t}$ on the vector of observables $\eta_t$, which is given by the first column of $B=(b_{\cdot1},\dots, b_{\cdot n})$.  We can estimate $\Sigma_\eta \bydef \E \eta_t \eta_t' $ using the sample analog and then solve equation $\Sigma_\eta= B  B'$ for the first column of $B$.
However, there is a continuum of solutions to this equation.
Therefore, additional assumptions or data are required to uniquely select $B$.

Many policy counterfactual questions only require one to identify a single column  $\beta = b_{\cdot1}$ of $B$.
Recently, several papers have proposed using moment conditions for the heteroscedasticity of the innovation process to identify the whole matrix $B$,  \citep{rigobon2003identification,lewis2021identifying}.
When a VAR model has many variables, this full identification approach may be challenging to interpret.
To make the interpretation of any single shock more transparent, one can reformulate this identification approach  as an external VIV problem.
Before we generalize the novel identification proposal, let us consider a simple motivating example.
\begin{example}\label{ex:VAR0}
Suppose $n=2$ and  that structural shocks have a representation
\begin{align*}
    \varepsilon_{it}=\zeta_{it}\sigma_{it},
\end{align*}
such that $\zeta_{it}$ is an i.i.d. sequence with   $\E \zeta_{t}=0$ and $\E \zeta_{t}\zeta'_{t}=I_2$; the process $\sigma_{it}$ is a stationary ergodic process that is independent of $\zeta_{it}$.
The volatility process $\sigma^2_{1t}$ for the shock of interest may be correlated with some observable time series  of    $Z_t$.
Suppose that this time series $Z_t$ only shifts the volatility  of the first structural shock but does not change volatility for other shocks.
We can then use it as a VIV for the identification of $b_{\cdot1}$.

Indeed, suppose that
  $\E  (Z_{t}-\E Z_{t} )(\sigma^2_{it}-\E\sigma^2_{it})= \rho_i,$
such that $\rho_1 \neq 0$ and $\rho_2=0$.

Then since
\begin{align*}
      \eta^2_{1t}&= (b_{11} \varepsilon_{1t}+ b_{12} \varepsilon_{2t})(b_{11} \varepsilon_{1t}+ b_{12} \varepsilon_{2t}), \\
       \eta_{1t}\eta_{2t}&= (b_{11} \varepsilon_{1t}+ b_{12} \varepsilon_{2t})(b_{21} \varepsilon_{1t}+ b_{22} \varepsilon_{2t}),
\end{align*}
we get, assuming that $b_{11}$ is not zero,

\begin{align*}
 \frac{ b_{21}}{ b_{11}}=    \frac{  \E   \eta_{1t}\eta_{2t}(Z_t -\E Z_{t})}{\E \eta^2_{1t}(Z_t -\E Z_{t})} .
\end{align*}\qedsymbol

\end{example}
Consider the general case where $\beta_Z  \bydef  \E(\eta_{t}\eta_{1t}-\E\eta_{t}\eta_{1t})(Z_t - \E Z_t)$. The identification conditions in  Example \ref{ex:VAR0} can be generalized to
\begin{assumption} \label{ass:instuments} VIV and shocks satisfy $\beta_{Z1} \neq 0$ and   for some $\rho_1\neq0$,   \\
  $\E \varepsilon_{it}\varepsilon_{jt} (Z_t - \E Z_t) =\rho_1 \mathds{1}\{i=j=1\}$
.
\end{assumption}
This assumption implies that the heteroscedasticity of the innovations for this time series must be correlated with the VIV $Z_t$.
In the monetary policy application, the federal funds rate is used as the normalizing variable, so the unit-impact normalization requires that the contemporaneous effect of the monetary policy shock on this variable be nonzero. The VIV can only affect the conditional variance of the shock of interest and cannot impact the conditional correlation of the other structural shocks.
\begin{thm}\label{thm:identification}
Suppose that Assumption \ref{ass:instuments} holds.
Then the impact vector for $\varepsilon_{1t}$ on $\eta_t$
is $\beta_Z/\beta_{Z1}$ with normalization of a unit impact on time series $i=1$ .
\end{thm}
\begin{proof}
    See Appendix \ref{app:proof_ID}.
\end{proof}
How can one find $Z_t$ for a particular shock? A recent paper by \citet{lewis2021identifying} shows that when the heteroscedasticity of a structural shock is persistent, one can use \emph{internal} VIV $Z_t = \eta^2_{t-1}$.
However, the approach can be applied more generally if appropriate \emph{external} VIV are available, as in our empirical application in Section \ref{sec:hetiv_id} which generalizes
\cite{rigobon2004impact} to a standard monthly monetary SVAR.
Specifically, in our application,  we use measures of FOMC meeting intensity as external VIV for changes in monetary-policy shock volatility.
Using Corollary \ref{cor:DeltaMethod}, one can conduct pointwise and simultaneous inference on the normalized structural impulse response functions $\psi_h=\frac{1}{\beta_{Z1}}C_h\beta_Z $, which is a smooth function of $\theta$ under Assumption \ref{ass:instuments}. The reduced-form IRF $C_h$ can be estimated using the MD approach from Section \ref{subsec:MD}.\footnote{Formula \eqref{eq:WMD} for the weighted estimator can be directly adapted to the product $\psi_h $ by substituting pairs of plug-in  estimators of $\hat\psi_h=\frac{1}{\hat\beta_{Z1}}\widehat C_h \hat \beta_Z$ and $\hat\psi^{MD}_h=\frac{1}{\hat\beta_{Z1}}\widehat C^{MD}_h \hat \beta_Z$ for $\widehat C_h$  and $C^{MD}_h$. }

\begin{rem}
 Suppose that we have two competing identification assumptions, for example, VIV and Cholesky.
We can construct a joint confidence set for the vector of differences in IRFs under these two schemes, $\frac{1}{\hat\beta_{Z1}}\widehat C_h \hat \beta_Z - \frac{1}{\hat\beta_{{\Sigma},1}}\widehat C_h \hat \beta_{\Sigma}$. Algorithm \ref{algo:supt} can then be used to compute sup-$t$ critical values for a test of hypothesis $\frac{1}{\beta_{Z1}} C_h \beta_Z = \frac{1}{\beta_{{\Sigma},1}} C_h  \beta_{\Sigma}$.
\end{rem}

In the context of full structural shock identification, \cite{lewis2022robust} shows that the original moment conditions in \cite{rigobon2004impact} have a weak identification problem, which complicates inference in realistic sample sizes.
The standard approach for weak-IV robust confidence bounds is  based on AR-type test inversion \citep{anderson1949estimation,montiel2021inference,lewis2022robust}.
We can use DWB to implement confidence sets that are robust to weak IV and an arbitrary number of unit roots (Assumption \ref{ass:SSW}).

Weak IV situations occur when $\beta_{Z1}$ is close to zero, so the Assumption \ref{ass:instuments} is nearly violated.
In such a situation, plug-in delta method inference becomes unreliable since $\hat\psi_{h}$   is a non-differentiable function of $\hat \theta$.
Nevertheless, it is possible to construct an asymptotically  valid confidence interval for $\psi_h$ using the Anderson-Rubin (AR) approach.
The idea is to invert the two-sided t-test of an equivalent linear hypothesis for each $k=1,\dots,n$.
\begin{equation*}
   H_0: e_k'  C_h  \beta_Z =\psi_{hk} \beta_{Z1}.
\end{equation*}
Using DWB, one can construct confidence sets for each structural IRF $\psi_{hk}$ for time series $k$ at the horizon $h$.
\begin{algorithm}
\caption{Anderson-Rubin confidence bounds}
\label{algo:AR}
        \begin{enumerate}
             \item[Step 1.] Pick $\psi^0_{hk}$  on a grid
             \item[Step 2.] Compute statistic $\zeta(\psi^0_{hk}) = e_k'  \widehat C_h  \hat\beta_Z -\psi^0_{hk} \hat\beta_{Z1}  $
             \item[Step 3.] Compare $|\zeta(\psi^0_{hk})|$ with DWB bootstrap $1-\alpha/2$ quantile  of
            \begin{equation*}
               \zeta^s(\psi^0_{hk}) = |  e_k'  (\widehat C^s_h  \hat\beta^s_{Z}-\widehat C_h  \hat\beta_Z) -\psi^0_{hk} (\hat\beta^s_{Z1} -\hat\beta_{Z1})|
            \end{equation*}
             \item[Step 4.] Collect all values $\psi^0_{hk}$ on a fine grid for which statistic  $|\zeta(\psi^0_{hk})|$ falls below the corresponding critical values. Take convex combination of the included grid points.
        \end{enumerate}
\end{algorithm}
This algorithm can be modified to improve efficiency by substituting minimum distance and weighted analogs for $\widehat C_h$.
The validity of Algorithm \ref{algo:AR} follows immediately from Corollary \ref{cor:DeltaMethod} and the fact that we are using a linear hypothesis.

Despite the necessity to recompute critical values for each hypothetical  $\psi^0_{hk}$,   Algorithm \ref{algo:AR} is relatively easy to compute.
The key parameter of interest is a scalar that requires only a univariate grid search.
The bootstrap sample   $\{\widehat C^s_h, \hat \beta^s_Z\}^{S}_{s=1}$ only needs to be computed once.

\section{Simulations and Empirical Application} \label{sec:symandempirics}
\subsection{Efficiency gains from smoothing}\label{subsec:MC}

We simulate the AR(1) model
\begin{align}
    y_t&=\rho y_{t-1}+\varepsilon_t, \qquad t=1,\ldots,T,\label{eq:AR-ARCH}\\
    \varepsilon_t&=(1+\gamma \mathbf{1}\{\zeta_{t-1}\le 0\})\zeta_t,\qquad \zeta_t\sim i.i.d. N(0,1)\qquad t=1,\dots,T. \label{eq:ARCH}
\end{align}
We consider the following grid $
\rho\in\{0,\,0.5,\,0.95,\,1\}$,
$T\in\{800,2400\}$,
$h\in\{1,\ldots,60\}$,
$\gamma\in\{0,\,1\}$,
where $\gamma=0$ corresponds to homoscedastic innovations and $\gamma=1$ introduces  conditional heteroscedasticity.
The asymmetric heteroscedasticity model \eqref{eq:ARCH} satisfies Assumption \ref{ass:innovations}, but violates the two-sided MDS innovations assumption maintained in \cite{montiel2021local}. A simple residual bootstrap without blocks is not guaranteed to consistently approximate the asymptotic distribution. This model is, however, empirically relevant, as it mimics a possible volatility boost of macro variables following a large negative structural shock.
Grid values for $\rho$ are chosen to represent four scenarios: (i) White noise  ($\rho=0$); (ii) Weakly persistent data ($\rho=0.5$); (iii) Near-unit root   ($\rho=0.95$); (iv) Unit root   ($\rho=1$).

The goal of the exercise is to compare the coverage probability and median length of confidence bounds for IRF $C_h=\rho^h$ in model \eqref{eq:AR-ARCH} for our novel procedures and the existing alternatives. We compare five lag-augmented procedures: (i) LP with confidence bounds robust to conditional heteroscedasticity (LP-DWB); (ii) Smoothed LP with percentile DWB (LP-Smoothed) (iii) Efficient weighted average of LP and Smoothed LP (LP-Weighted) (iv) LP with simple Eicker-Huber-White standard errors proposed in \cite{montiel2021local} (LP-White) (v) Iterative estimates using a lag augmented AR model with percentile residual bootstrap proposed in \cite{inoue2020uniform} (AR-LA-Residual).
We focus on the lag augmented procedures since, without lag augmentation, the unit root case requires pretesting and first differencing.
We use 500 bootstrap draws and 1000 Monte Carlo simulations for all the designs and the competing methods.
\begin{figure}[H]
\begin{center}
 \includegraphics[width=0.9\linewidth]{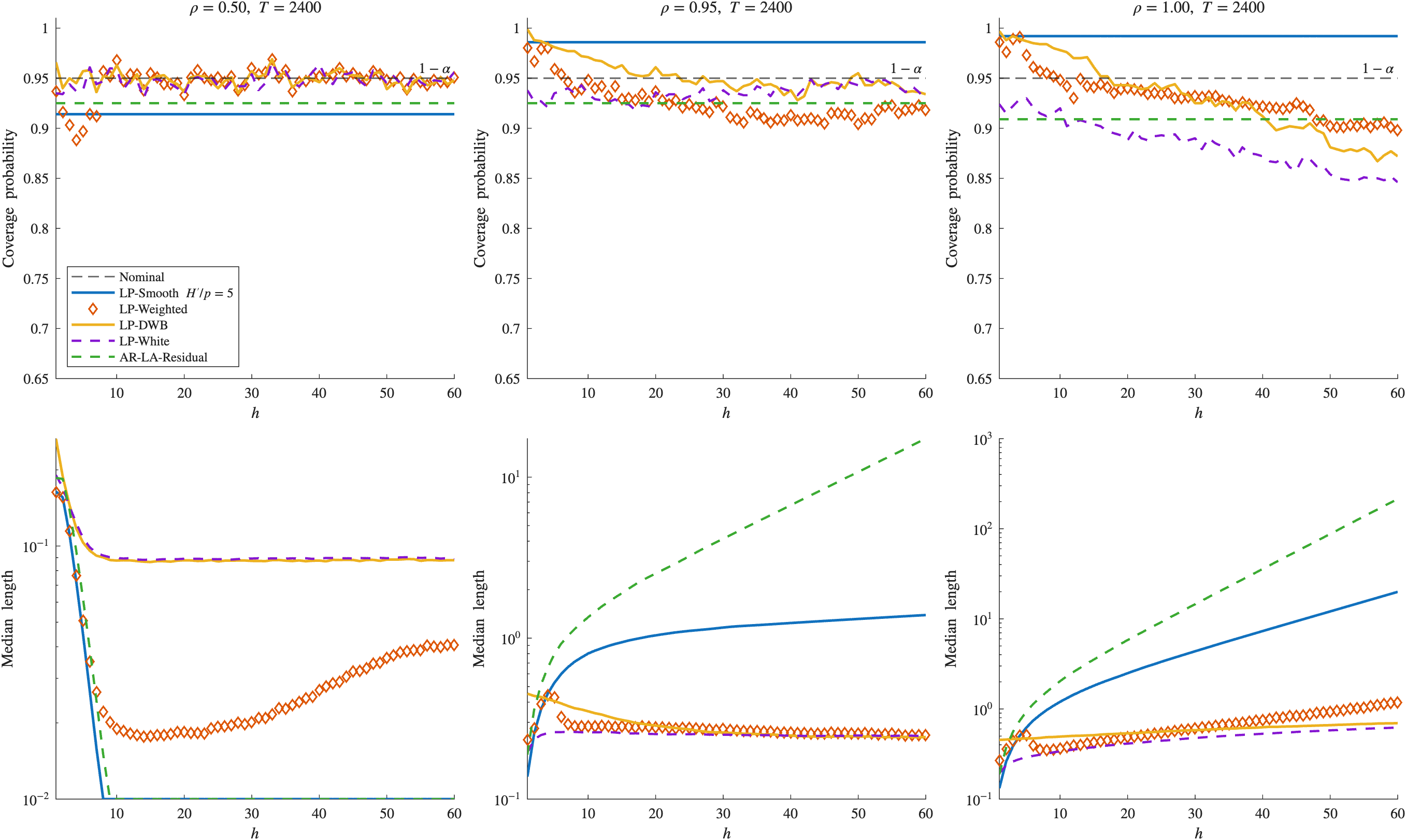}
\caption{Coverage probability and median interval length in AR(1) design with  asymmetric conditional heteroscedasticity ($\gamma=1$) for large sample ($T=2400$).
}\label{fig:MC_2400_main}
\end{center}
\end{figure}
Figure \ref{fig:MC_2400_main} compares the coverage probability and confidence interval (CI) length in large samples with persistent data, where the asymptotic efficiency comparison becomes evident. This figure shows the smoothed and weighted LP estimates for $H'=5$, which correspond to the ratio of $H'/p=5$ used in our empirical application.
We focus on the asymmetric heteroscedasticity designs in the main text while delegating the homoscedastic case to Supplementary Appendix  Figure \ref{fig:MC_gamma_0}.
For weakly persistent data ($\rho=0.5$), all procedures are relatively close to the nominal size for all horizons, with LP-DWB and LP-White coverage being the most precise. The corresponding CI length for non-linear plug-in estimators (AR-LA and LP-Smoothed) is orders of magnitude smaller than that for LP procedures, with LP-Smoothed having the shortest intervals. It is noticeable that LP-Weighted delivers much shorter CIs than basic LP estimators, with nearly identical coverage probabilities for $h\geq8$.
For near-unit root design ($\rho=0.95$), the order of median length flips, and the AR-LA and LP-Smoothed become one order of magnitude wider than other procedures for $h\geq H'$.
In the unit root case ($\rho=1$), the gap becomes even wider. In this design, the LP-Weighted estimator shows the most precise coverage probability of all the procedures, with a length comparable to simple LP-White without compromising the coverage.
Supplementary Appendix Figure \ref{fig:MC_rho_0} shows corresponding results for the white noise case ($\rho=0$). Since the first order derivative of $C_h=\rho^h$ is zero for $h\geq1$, the first order delta method is only applicable to $h=1$. As a result, AR-LA and LP-Smoothed show very low coverage probability for every even $h$. In contrast, LP-DWB and LP-White do not rely on the delta method and show coverage close to $95\%$ uniformly over $h$.  LP-Weighted procedure has a built-in regularization that bounds the weight on LP-Smoothed component away from 1. As a result, the LP-Weighted procedure preserves robustness to delta-method failure while improving the efficiency of the CIs.

Finally, Supplementary Appendix Figures \ref{fig:MC_hprime_800} and \ref{fig:MC_hprime_2400}
 show the effect of the tuning parameter choice $H'$ on the performance of the weighted procedure.
 For smaller $\rho$, it may be beneficial  to ``under-smooth'' ($H'=p$) to improve coverage probability at small $h$. For more persistent data, however, this under smoothing results in overly wide CIs. Over-smoothing, $H'=60p$, becomes problematic for very persistent data because of accumulated bias in LP estimates at long horizons. We recommend a rule of thumb $H'=5p$. As the sample size grows from 800 to 2400, the effects of the choice of $H'$ on coverage probability become weaker, and higher smoothing results in shorter CIs, which is consistent with the asymptotic efficiency property.

\subsection{Identification using VIV in a DSGE model}
 The Supplementary Appendix Section \ref{sec:dsge} studies the finite-sample performance of the proposed VIV
 local projection procedure and the associated dependent wild bootstrap confidence
intervals in data generated by a well-known medium-scale New Keynesian model \cite{smets2007shocks}. The goal of this exercise is twofold. First, we illustrate that the VIV moments can correctly recover structural impulse responses to a monetary policy shock in a realistic model.
Second, we compute the coverage and average length of the confidence sets based on the dependent wild bootstrap in a misspecified VAR(8) model when the true model is VARMA(3,2).
Using a small scale Monte Carlo study, we find that LP IRF estimates identified using VIV can successfully estimate the true IRFs to a monetary policy shock.
Weighted IRFs allow one to considerably reduce the length of the confidence intervals while maintaining good coverage probability, especially for highly persistent variables estimated in levels.

\subsection{Empirical application to US monetary policy }\label{sec:monetaryshock}

Our last empirical exercise, illustrating the usefulness of the DWB, reconsiders the classical question of identifying the impact of monetary policy shocks in the US.
The identification of monetary policy shocks in SVAR models goes back to \cite{sims1980comparison}.
This application has been extensively studied, with many identification schemes proposed, including the classical Cholesky decomposition \citep{sims1980comparison,christiano1999monetary}, external  instrumental variables \citep[for example,][]{romer1989does,romer2004new,gertler2015monetary}, sign restrictions \citep[for example,][]{uhlig2005effects,gafarov2018delta,antolin2018narrative}, high-frequency identification \cite[for example,][]{faust2004identifying}, and regime-switching models \citep[for example,][]{sims2006were,lutkepohl2017structural}, to name a few.\footnote{For foundational context and a detailed examination of prior empirical efforts, we refer the reader to seminal works by \cite{christiano1999monetary}, \cite{boivin2010has}, and \cite{ramey2016macroeconomic}, which provide comprehensive reviews on the subject.}
Such an abundance of empirical studies  allows us to evaluate the practical implications of the novel weighted LP estimators and the DWB confidence sets in both strongly identified  (Cholesky identification) and more agnostic but weakly identified  (VIV identification) schemes.

We focus on the reduced-form SVAR model with six variables studied in seminal papers, including, but not limited to \cite{christiano1999monetary}, \cite{uhlig2005effects}, \cite{antolin2018narrative}.
Our specification includes six variables for the U.S.: real GDP, GDP deflator, commodity price index, non-borrowed reserves, the 1-year bond rate, and the federal funds rate, as described in Table \ref{tab:data_sources} in Supplementary Appendix \ref{sec:app_monetary_policy}. The sample for the baseline covers the period from January 1950 to December 2024, with monthly data.  Our specification includes a constant and a polynomial time trend of order up to $k\in\{0,6\}$ and the number of lags $p=12$.\footnote{
For example, \citet{christiano1999monetary} uses a constant but does not include a time trend, with lag lengths of 6 and 12, while \citet{uhlig2005effects} and \citet{antolin2018narrative} use 12 lags without a constant or deterministic trend.
Supplementary Appendix \ref{sec:app_extra_figs} reports results for $k=0$.
}

\subsubsection{Efficiency gains from optimal smoothing}

Consider the effect of IRF smoothing in the classical recursive (Cholesky) identification.
Figure \ref{fig:IRF_CEEsmooth} compares impulse response functions with    non-smoothed  and    weighted (smoothed) pointwise bands for the benchmark specification for selected variables.\footnote{For the sake of space, we report impulse responses of only three variables in the main text; figures with six impulse responses are available upon request.
Supplementary Appendix \ref{sec:app_extra_figs} shows weighted LP IRFs for the full set of variables in selected specifications. }

The recursively identified structural IRF of real GDP to a monetary policy shock is well-interpreted, with a 12 bp reduction in response to a 25 bp federal funds rate shock at a 24 month horizon. This scheme, however, exhibits a price puzzle, i.e., a temporary statistically significant increase in the GDP deflator level  in response to a monetary tightening shock.
For commodity prices, the puzzle is relatively short lived and statistically significant only for the weighted estimates.
We will revisit these findings using alternative identification approaches in Section \ref{sec:hetiv_id}

What is the effect of optimal smoothing? It leaves short-horizon dynamics almost unchanged.
Compared to the raw (non-smooth) LP estimates, the weighted IRF results in a stronger impact of the shock on GDP   but a weaker impact on inflation at longer horizons.
Notice that the confidence bounds become tighter as a result of the smoothing.
This tightening occurs due to efficiency gains from imposing model-based shape restrictions on the LP coefficients.

\begin{figure}[H]
\begin{center}
 \includegraphics[width=\linewidth,trim= 100pt 350pt  80pt 20pt,clip]{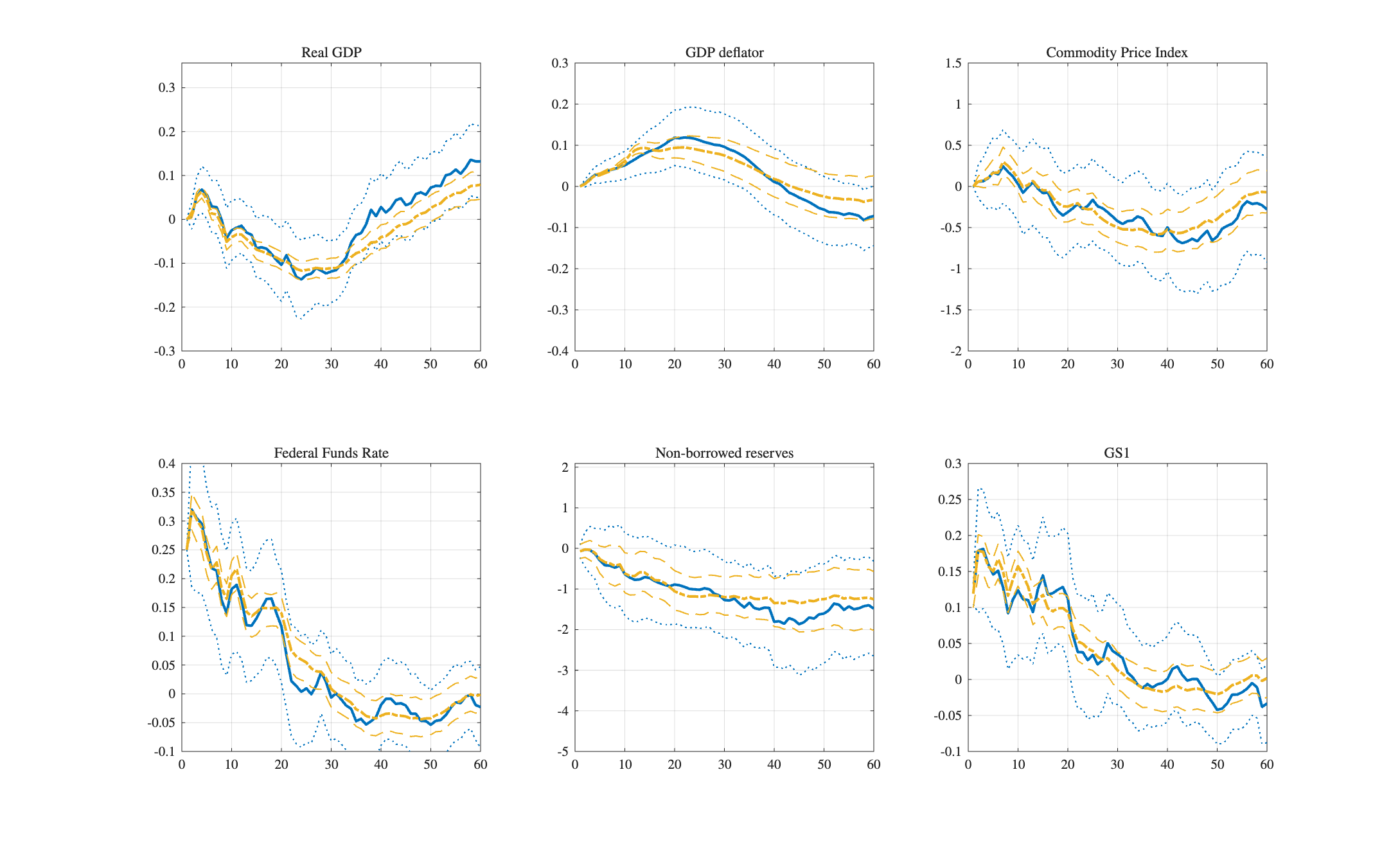}
\vspace{-1cm}
\caption{LP IRF estimates under Cholesky identification (solid lines) with corresponding   $90\%$  pointwise DWB confidence bands (dotted and dashed lines).  Non-smoothed estimates (blue lines) and weighted estimates (orange lines).
 }\label{fig:IRF_CEEsmooth}
\end{center}
\end{figure}
Supplementary Appendix Figure \ref{fig:IRF_CEEnoTrend} shows that the inclusion of a flexible trend in the main specification eliminates the long run effects of monetary policy shocks on real GDP but does not substantially affect the price dynamics.
This figure also shows that the weighted LP confidence bounds are often tighter than the corresponding conventional homoscedastic residual bootstrap bounds computed recursively for VAR models estimated in levels using least squares.
These findings are in line with our Monte Carlo results for near-unit and unit root designs shown in Figure \ref{fig:MC_2400_main}.

\subsubsection{More efficient sup-$t$ simultaneous confidence bounds}

Next, we compare the DWB-based sup-$t$ confidence bands with simpler Bonferroni bands that do not use information about the correlation between the IRF estimators across horizons.
We focus on the weighted LP estimators.
Figure \ref{fig:IRF_CEEsupt} reports pointwise and simultaneous confidence bands after smoothing.

\begin{figure}[H]
\begin{center}
 \includegraphics[width=\linewidth,trim= 100pt 350pt  80pt 20pt,clip]{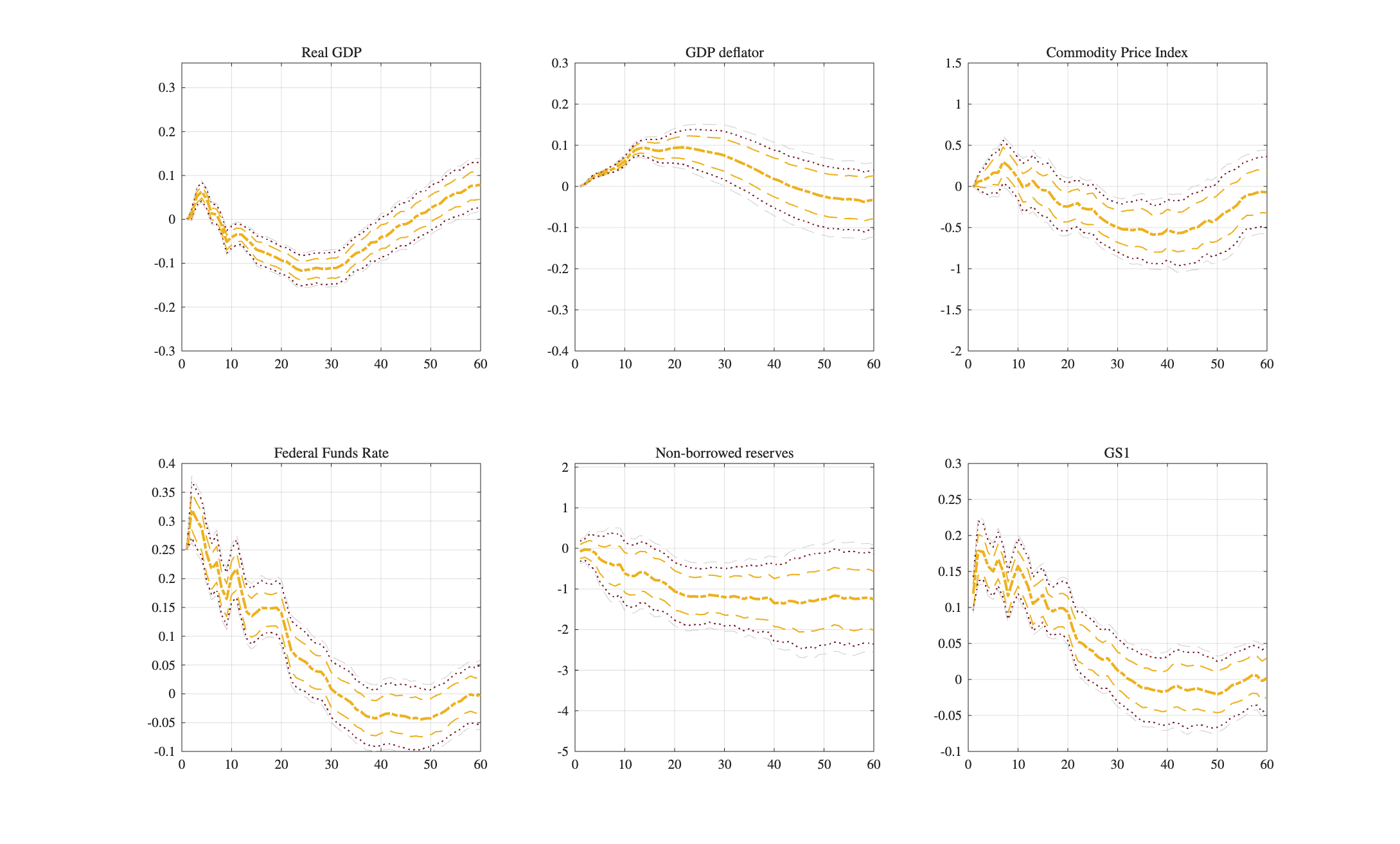}
 \vspace{-1cm}
\caption{Comparison of $90\%$ confidence sets for smoothed LP IRFs under Cholesky identification, $k=6$. Point estimates (solid orange lines), pointwise confidence bounds (dashed orange lines), sup-$t$ simultaneous bounds (purple dots), Bonferroni simultaneous bounds (dashed gray).}\label{fig:IRF_CEEsupt}
\end{center}
\end{figure}
\vspace{-1cm}
The main observation here, similar to \cite{montiel2019simultaneous}, is that there is substantial tightening of the simultaneous confidence bounds for IRFs based on sup-$t$ when compared to the so-called union or Bonferroni bounds \citep{inoue2023significance}.
Unlike our DWB sup-$t$ bounds, the Bonferroni  approach assumes the worst-case dependence between the individual components of the IRF estimates, which results in more conservative bounds.
In contrast to earlier studies, our paper is the first to propose a bootstrap method for the computation of sup-$t$ critical values in the presence of an arbitrary number of unit roots in general multivariate VAR models.

\subsubsection{Identification using VIV}
\label{sec:hetiv_id}

Following the insight from   \cite{rigobon2004impact} that the volatility of monetary policy shocks is higher on the dates of the FOMC meetings,
we propose  external instruments $Z_t$ that are allowed to shift the conditional variance of the monetary policy shock while remaining orthogonal to other shocks and lagged information.
Differently from \cite{rigobon2004impact}, our instruments are not necessarily binary; instead, we consider various instrument intensities.

The VIV identification approach has two practical advantages. First, it avoids the  ``labeling'' problem common to purely heteroscedasticity-based identification because $Z_t$ targets the variance of a specific shock of interest (monetary policy). Second, using a count/intensity instrument at the monthly frequency provides additional time-series variation relative to a binary event dummy, which is useful in practice when event incidence is sparse.
The use of our VIV $Z_t$ should be distinguished from the standard ``external instruments'' in the high-frequency identification tradition: rather than using intra-day asset-price surprises as proxies for the shock itself, we use event-based intensity measures to shift the second moment of the policy shock and then conduct inference on the resulting stacked structural object using the DWB.\footnote{High-frequency identification examples include \cite{cochrane2002fed} \cite{gurkaynak2005sensitivity}, \cite{gertler2015monetary}, and \cite{bauer2022reassessment}

}

We focus exclusively on the dates of FOMC meeting announcements, including meetings, telephone and video conferences, unscheduled meetings, and sequential day meetings.
Our baseline instrumental variable indicates the number of meetings that occur in a given month, including telephone conferences and unscheduled meetings.\footnote{See Supplementary Appendix \ref{sec:app_monetary_policy} for details.} Our instrument remains agnostic about the sign and size of the policy change.
Unlike many other papers that restrict  the FOMC meetings sample starting from 1994,
\footnote{This is because the Fed released no public announcements about monetary policy decisions prior to 1994. Although some scholars, for example, \cite{swanson2023speeches}, extended the sample back to 1988.} we extend our sample back to 1950 and forward to 2024.
Supplementary Appendix Figure \ref{fig:FOMC_events} and Figure \ref{fig:FOMC_days} summarize the time-series variation in the number of FOMC events and in the number of meeting days per month, distinguishing scheduled meetings from the full set that also includes unscheduled meetings and conference calls.

A natural concern that one might have is that the occurrence (or timing) of meetings may be endogenous to macroeconomic conditions.
To address this concern, we consider an alternative external instrument based on meeting-related surprises in trading activity (share volume or dollar volume). Here we follow the high-frequency identification literature, specifically \cite{kuttner2001monetary}, as we construct daily surprises. Next, we aggregate them into monthly series, as illustrated in Supplementary Appendix Figure \ref{fig:usd_log_vol_pre1_post1_FALSE_TRUE}. Intuitively, this instrument isolates changes in trading activity concentrated in narrow windows around FOMC events.

It is worth noting that the hypothesis  $\beta_{Z1}=0$  is rejected in favor of the alternative $\beta_{Z1}>0$ (Assumption \ref{ass:instuments}) at 5\% significance level.
The corresponding t-test statistic is $t=1.83$,  with a one-sided p-value  equal to $0.034$ based on DWB.
In other words, the variance of the federal funds rate forecast errors is higher in the months with more FOMC meetings.
This confirms earlier findings in the high-frequency identification literature that markets react to FOMC meeting announcements \citep[see, for example,][]{faust2004identifying}.
At the same time, since the $t$-statistic is  below the conventional cut-off for weak IV tests, we should report weak IV robust confidence bounds together with the simple DWB bounds.

Figure \ref{fig:IRF_Rigobon} reports VIV impulse responses together with two inference procedures: (i) pointwise DWB CI and (ii) Anderson-Rubin (AR) CI to address the potential weak instrument problem. For comparison, we also plot Cholesky-based point estimates.
Similarly to recursive identification, one can see a significant drop in output, which remains significant from 20 to 50 months using both pointwise confidence bands as well as AR confidence intervals.
Interestingly, the price puzzle survived even when we apply AR CI, with point estimates being even larger than those for the recursive approach.
Overall, both asymptotically normal and AR-type confidence sets include the point estimates corresponding to the recursive approach. In other words, the classical Cholesky-type identification is consistent with the new VIV approach, which does not impose a priori zero impact restrictions on real GDP, the GDP deflator, and commodity prices.

\begin{figure}[H]
\begin{center}
 \includegraphics[width=\linewidth,trim= 100pt 350pt  80pt 20pt,clip]{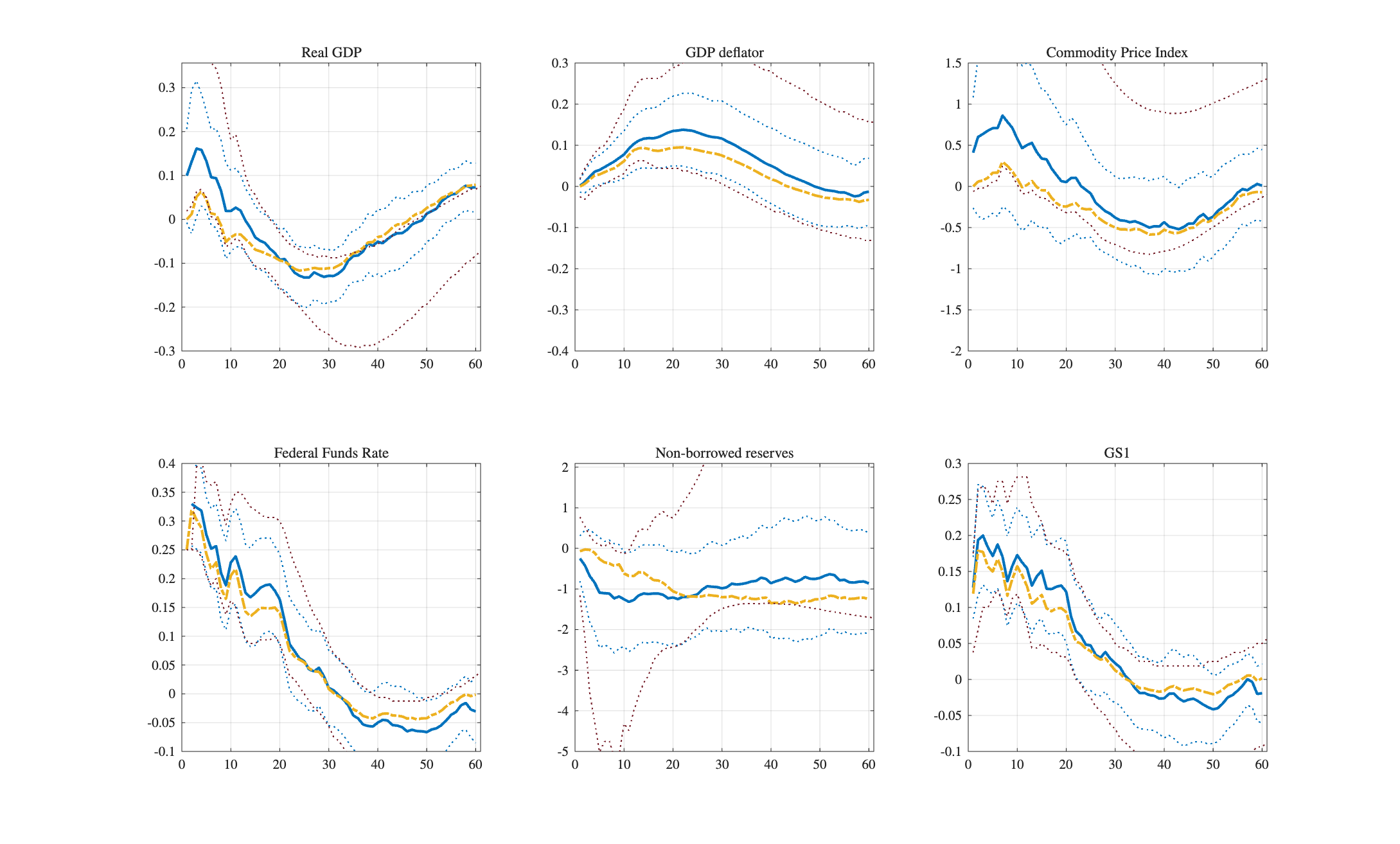}
 \vspace{-1cm}
 \caption{Estimates of monetary policy shocks using FOMC meeting count as VIV based on weighted LP approach. VIV point estimates (blue solid line), pointwise confidence bounds (blue dots),  pointwise AR confidence bounds (purple dots), Cholesky point estimates (solid gold). Confidence coefficient is  $90\%$, $k=6$.
}\label{fig:IRF_Rigobon}
\end{center}
\end{figure}

Figure \ref{fig:IRF_Volumes} repeats the analysis using the trading-activity surprise instrument that places a higher weight on the meeting dates with large jumps in the volume of financial trading. The FOMC meetings that were not followed by a jump in trading are less likely to introduce large policy shocks.
Since the variation-in-trading-volume VIV is only present in the subset of FOMC meetings, this VIV is weaker than the simple meeting count
The null hypothesis  $\beta_{Z1}=0$  is rejected with a larger p-value of 0.08. As a result, AR-bounds are infinite at a 90\% confidence level for most variables, and we report 80\% confidence sets instead.
Two observations are worth mentioning. First, the output response remains statistically significant even when the weak IV problem is accounted for, albeit at a lower confidence level. Second, the price becomes statistically insignificant upon impact,  but it is still within the confidence sets at medium horizons.

To summarize, a 12 bp reduction in real GDP at approximately a 24 month horizon is common among the three alternative identification schemes. The price puzzle, on the contrary, is more sensitive to the choice of identification. At the same time, across the three schemes, the price puzzle cannot be rejected at medium horizons when using the full dataset from 1950 to 2024.

\begin{figure}[H]
\begin{center}
 \includegraphics[width=\linewidth,trim= 100pt 350pt  80pt 20pt,clip]{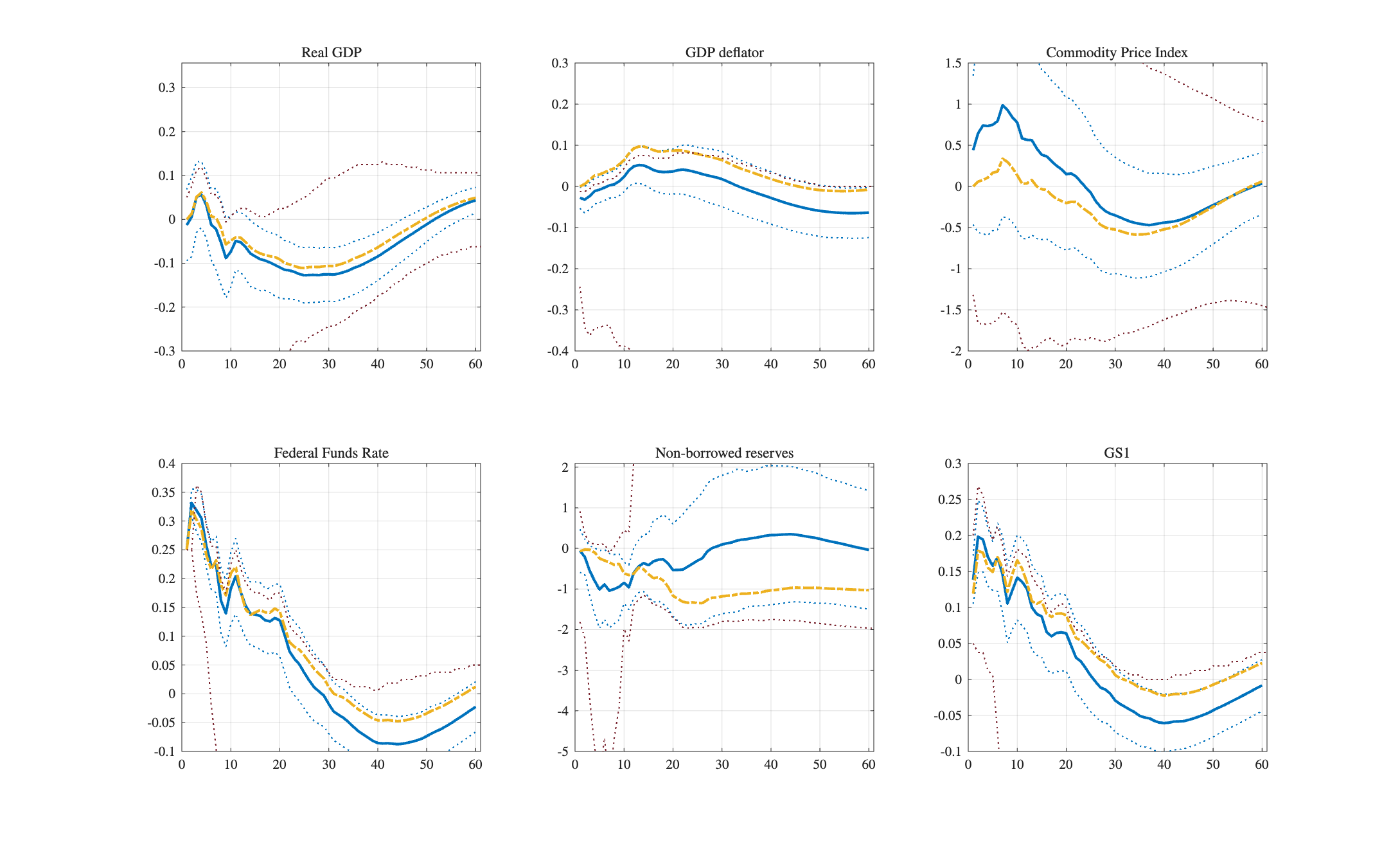}
\vspace{-1cm}
\caption{ Estimates of monetary policy shocks using trading volume jumps  as VIV based on weighted LP approach. VIV point estimates (blue solid line), pointwise confidence bounds (blue dots),  pointwise AR confidence bounds (purple dots), Cholesky point estimates (solid gold). Confidence coefficient is  $80\%$, $k=6$.
} \label{fig:IRF_Volumes}
\end{center}
\end{figure}

\section{Conclusion} \label{sec:conclusion}

This paper shows how to perform robust frequentist inference  on many practically relevant parameters   in general potentially nonstationary SVAR models using dependent wild bootstrap applied to local projection estimators of IRFs.
Our main insight is that the local projection estimator applied to  data in levels automatically prewhitens the data and makes the joint CLT and DWB applicable.
Through the delta-method, this allows us to estimate the asymptotic distribution of any smooth function of the entire   collection of LP estimates over multiple horizons and the external moment conditions.
Joint inference over multiple horizons, as opposed to horizon specific inference, allows us to: (i) derive efficient smoothed LP estimators and estimate their asymptotic distribution for potentially nonstationary data; (ii) perform sup-$t$ tests of misspecification and construct efficient simultaneous confidence bands; (iii) incorporate external moments such as VIV to make inferences on structural IRFs.
We illustrate the statistical validity and efficiency gains from these procedures in Monte Carlo experiments, both in the AR(1) model and in VAR(8)  approximation to a VARMA(3,2) representation of \cite{smets2007shocks} DSGE model, as well as using a six-variable monetary SVAR model with a novel VIV identification.
Efficient smoothing reduces the length of confidence intervals, particularly for the IRFs at long horizons. Additionally, weak-IV robust intervals based on smoothed LP are informative about statistically significant real effects of monetary policy shocks while highlighting the fact that the price puzzle is sensitive to the choice of identification.

\part*{ Appendix}
\appendix

\section{Proofs of the main results}

\subsection{Proof of Theorem \ref{thm:LPconsistency}, LP consistency} \label{sec:app_proof_LP}

\textbf{Step 1:} Use the Frisch-Waugh-Lovell theorem (Theorem 4.2 in \cite{lovell1963seasonal}) to obtain a representation
\begin{equation*}
    \widehat C_h^\prime =     (\sum_{t=p}^{T-h} \hat \eta_t  \hat \eta_{t}^\prime )^{-1}\sum_{t=p}^{T-h} \hat \eta_t (\hat \eta^{(h)}_{t+h} )^\prime .
\end{equation*}

\textbf{Step 2.} By Lemma \ref{lem:residuals} in Supplementary Appendix \ref{sec:app_aux_lemma}
\begin{align*}
    \widehat C_h^\prime &=     (\sum_{t=p}^{T-h} \hat \eta_t  \hat \eta_{t}^\prime )^{-1}\sum_{t=p}^{T-h} \hat \eta_t (\hat \eta^{(h)}_{t+h} )^\prime = \big(\frac{1}{T}\sum_{t=p}^{T-h}   \eta_t   \eta_{t}^\prime +  O_p(T^{-1})\big)^{-1}( \sum_{s=0}^h \frac{1}{T}\sum_{t=p}^{T-h}   \eta_t ( C_s\eta_{t+h-s})^\prime   +  O_p(T^{-1}))\\
    &= C_h^\prime +  \big(\frac{1}{T}\sum_{t=p}^{T-h}   \eta_t   \eta_{t}^\prime +  O_p(T^{-1})\big)^{-1}(  \frac{1}{T}\sum_{t=p}^{T-h} \sum_{s=0}^{h-1}  \eta_t ( C_s\eta_{t+h-s})^\prime   +  O_p(T^{-1})) + O_p(T^{-1}).
\end{align*}
Since $\eta_t$ is stationary and mixing, and thus the sequence $\eta_t\eta_{t+h}$ is ergodic with finite fourth moments, we can apply the ergodic theorem to obtain the desired result.

The same proof of consistency applies to $\hat\Sigma_\eta$.

\subsection{Proof of  Theorem \ref{thm:IVconsistency}, consistency of IV estimator}\label{sec:app_proof_IV}

Consider an individual term in $\hat \beta_Z$,
\begin{align}
      \hat\eta_{it}\hat\eta_{1t}(Z_t - \bar Z_T) =&   \eta_{it} \eta_{1t}(Z_t - \bar Z_T) + (\hat\eta_{it}-\eta_{it}) \eta_{1t}(Z_t - \bar Z_T) \notag\\
       &+   \eta_{it}(\hat\eta_{1t}-\eta_{1t})(Z_t - \bar Z_T) + (\hat\eta_{it}-\eta_{it})(\hat\eta_{1t}-\eta_{1t})(Z_t - \bar Z_T).\label{eq:estimatedCovZ}
\end{align}

First, note that by using the
same argument as in Step 3 in the proof of Lemma \ref{lem:residuals},
\begin{align*}
    \bar Z_T  \sum_{t=p}^T  (\hat\eta_{it}-\eta_{it}) \eta_{jt}  =  O_p(1).
\end{align*}
Next, by Assumption \ref{ass:innovations}.5, $ Z_t  \eta_{t} $ is zero mean and is uncorrelated with $\Tilde{X}_{t-1}^\prime$. Using the same argument as in Step 2 of the proof of Lemma \ref{lem:residuals}, $\sum_{t=p}^{T}  Z_t  \eta_{t}  \Tilde{X}_{t-1}^\prime    \Upsilon_T^{-1}   $ is a tight random sequence by  \citet[Theorem 4.1]{de2000functional}.
So
\begin{align*}
      (\sum_{t=p}^T  (\hat\eta_{t}-\eta_{t}) \eta'_{t}Z_t )' = \sum_{t=p}^{T}  Z_t  \eta_{t}  \Tilde{X}_{t-1}^\prime    \Upsilon_T^{-1}    ( \Upsilon_T^{-1} \sum_{r=p}^{T} \Tilde{X}_{r-1} \Tilde{X}_{r-1}^\prime \Upsilon_T^{-1})^{-1}  \sum_{r=p}^{T} \Upsilon_T^{-1}  \Tilde{X}_{r-1}\eta'_{r} =O_p(1).
\end{align*}

The last term in  \eqref{eq:estimatedCovZ} can be bounded by the Cauchy-Schwarz theorem,
\begin{align*}
    | \sum_{t=p}^T  (\hat\eta_{it}-\eta_{it})(\hat\eta_{1t}-\eta_{1t})(Z_t - \bar Z_T)  |\leq  \sqrt{ \sum_{t=p}^T  (\hat\eta_{it}-\eta_{it})^2 \sum_{t=p}^T (\hat\eta_{1t}-\eta_{1t})^2(Z_t - \bar Z_T)^2  } = O_p(T^{1/4}),
\end{align*}
since
\begin{align*}
   \sum_{t=p}^T (\hat\eta_{1t}-\eta_{1t})^2(Z_t - \bar Z_T)^2 = \sum_{t=p}^T  (\hat\eta_{1t}-\eta_{1t})^2Z^2_t - 2 \bar Z_T \sum_{t=p}^T  (\hat\eta_{1t}-\eta_{1t})^2Z_t  +  \bar Z_T^2  \sum_{t=p}^T  (\hat\eta_{1t}-\eta_{1t})^2 .
\end{align*}
and  by Assumptions \ref{ass:innovations}.2 and \ref{ass:innovations}.3  and Lemma \ref{lem:residuals}
\begin{align*}
        \max_{t=p,\dots,T}|Z_t|    & =  \sqrt[4]{ \max_{t=p,\dots,T}|Z_t|^4}\leq T^{1/4} \sqrt[4]{ \sum_{t=p}^T \frac{|Z_t|^4}{T}} =O_p(T^{1/4}),\\
       \max_{t=p,\dots,T}|Z_t|^2    & =  O_p(T^{1/2}),\\
     \sum_{t=p}^T  (\hat\eta_{1t}-\eta_{1t})^2 &=  O_p(1).
\end{align*}
By Assumption \ref{ass:innovations}, the time series $\eta_{t} \eta_{1t} $ and $Z_t$ satisfy CLT for mixing processes. Therefore,
\begin{align*}
    \sum_{t=p}^T   \eta_{t} \eta_{1t}(Z_t -  \bar Z_T) =& \sum_{t=p}^T  ( \eta_{t} \eta_{1t}-\E\eta_{t} \eta_{1t})(Z_t - \E Z_t) \\
    &+ \frac{1}{\sqrt{T-p}}\sum_{t=p}^T( Z_t - \E Z_t)  \frac{1}{\sqrt{T-p}} \sum_{t=p}^T  ( \eta_{t} \eta_{1t}-\E\eta_{t} \eta_{1t})\\
    =&\sum_{t=p}^T  ( \eta_{t} \eta_{1t}-\E\eta_{t} \eta_{1t})(Z_t - \E Z_t) + O_p(1).
\end{align*}
To summarize,
\begin{align*}
  \hat\beta_Z =   \frac{1}{T-p}   \sum_{t=p}^T  (\hat\eta_{t}\hat\eta_{1t}-\overline{\hat\eta_{t}\hat\eta_{1t}})(Z_t - \bar Z_T) =&  \frac{1}{T-p}   \sum_{t=p}^T  ( \eta_{t} \eta_{1t}-\E\eta_{t} \eta_{1t})(Z_t - \E Z_t)     + O_p(T^{-\frac{3}{4}}).
\end{align*}
By ergodicity of $( \eta_{t} \eta_{1t}-\E\eta_{t} \eta_{1t})(Z_t - \E Z_t)$, we get $\hat\beta_Z \pTo \beta_Z$.
The statement of the theorem then follows from Theorems   \ref{thm:LPconsistency} and  \ref{thm:identification}.

\subsection{Proof of Theorem \ref{thm:inferenceCH}, asymptotic normality of the LP estimators}\label{sec:app_proof_CLT}

By Assumption \ref{ass:innovations}, $ \xi_t$ is a strong mixing process of size at least $-2(1+\epsilon)/\epsilon$ with zero mean and finite $2+\epsilon$ moments (condition DMR is satisfied, see p. 12 in \cite{rio2017asymptotic}).
By Theorem 4.2 in \cite{rio2017asymptotic} and the Cramer-Wold theorem, there exists the long run variance matrix $\Omega$ for $ \xi_t $ and
\begin{equation}
    \frac{1}{\sqrt{T-p-H}}\sum_{t=p}^{T-H}  \xi_t \weakto N(0, \Omega). \label{eq:CLTmixing}
\end{equation}

By Theorems \ref{thm:LPconsistency} and \ref{thm:IVconsistency},
\begin{align*}
     \sqrt{T-p} (\hat \Sigma_\eta - \Sigma_\eta) &=  \frac{1}{\sqrt{T-p}}\sum_{t=p}^{T}  ( \eta_t   \eta_{t}^\prime -\Sigma_\eta)+  O_p(T^{-\frac{1}{2}}),  \\
   \sqrt{T-p-h} ( \widehat C_h^\prime - C_h^\prime)&=  \big(\frac{1}{T-p-h}\sum_{t=p}^{T-h}   \hat\eta_t   \hat\eta_{t}^\prime  \big)^{-1}(  \frac{1}{\sqrt{T-p-h}}\sum_{t=p}^{T-h} \sum_{s=0}^{h-1}  \eta_t ( C_s\eta_{t+h-s})^\prime    ) + O_p(T^{-\frac{1}{2}}),\\
    \sqrt{T-p} (\hat \beta_Z  - \beta_Z) &=  \frac{1}{\sqrt{T-p}}\sum_{t=p}^{T}
    \left[(\eta_{t} \eta_{1t}-\E\eta_{t}\eta_{1t})(Z_t - \E Z_t)-\beta_Z\right]
    +  O_p(T^{-1/4}) .
\end{align*}

Since l.h.s. of \eqref{eq:CLTmixing} differs from \eqref{eq:CLTmainText} by some  $O_p(T^{-\frac{1}{4}})$ terms,
the statement of the theorem follows from the continuous mapping theorem \citep{mann1943stochastic}.

\subsection{Proof of Theorem \ref{thm:hac}, consistency of HAC standard errors}\label{sec:app_proof_hac}

\textbf{Step 1.} To show consistency of the whole matrix $\widehat\Omega_T$, we need to establish consistency of HAC estimators for all combinations $  a' \xi_t$.
Then the HAC estimator for particular $a$ is defined as
     $\widehat\Omega_T(a) \bydef \sum_{\ell=-B_T}^{B_T} \frac{K(\frac{\ell}{B_T})}{T-\ell}  \sum_{t=p}^{T-\ell} (a^\prime\hat \xi_t) (  \hat \xi_{t+\ell}^\prime a)$ .
This estimator has an infeasible analog
     $\widetilde\Omega_T(a)  \bydef   \sum_{\ell=-B_T}^{B_T} \frac{K(\frac{\ell}{B_T})}{T-\ell}  \sum_{t=p}^{T-\ell} (a^\prime  \xi_t) (    \xi_{t+\ell}^\prime a) .$

Since $\xi_t $ is strong mixing with $4+\epsilon$ moment, Theorem 1(a) in \cite{andrews1991heteroskedasticity} implies that $\widetilde\Omega_T(a)$ is consistent for $\Omega(a)=a'\Omega a$.

\textbf{Step 2.} What remains to show is $\widehat\Omega_T(a)-\widetilde\Omega_T(a) \pTo 0$.
Since by Theorem \ref{thm:LPconsistency}, $\hat \Sigma_\eta\pTo\Sigma_\eta $ and $\widehat C_h\pTo C_h $, we just need to establish HAC consistency for vectors
\begin{equation*}
   \hat\zeta_t = \begin{pmatrix}

             \hat \zeta^{(1)}_t   \\
         \cdots\\
             \hat \zeta^{(H)}_t  \\
                               \hat \zeta^{(H+1)}_t  \\
           \hat \zeta^{(H+2)}_t
         \end{pmatrix}= \begin{pmatrix}
         vec(    (\hat\eta_t\hat\eta'_{t+1} -\overline{\hat\eta_t\hat\eta'_{t+1} }) )\\
         \cdots\\
         vec(  (\hat\eta_t\hat\eta'_{t+H}-\overline{\hat\eta_t\hat\eta'_{t+H}}))\\
         vec( \hat\eta_t  \hat\eta_{t} ^\prime - \hat\Sigma_\eta )\\
         ( \hat\eta_{t} \hat\eta_{1t}-\overline{\hat\eta_{t} \hat\eta_{1t}})(Z_t - \bar Z_T)- \hat\beta_Z
         \end{pmatrix}.
\end{equation*}
Let $\zeta_t$ denote the population analog of $\hat\zeta_t$. Using
$vec(AXB)=(B^\prime\otimes A)vec(X)$, the original horizon-specific scores are
linear transformations of the auxiliary scores: for each $h=1,\dots,H$,
$\xi_t^{(h)}=\sum_{s=0}^{h-1}(C_s\otimes\Sigma_\eta^{-1})
\zeta_t^{(h-s)}$ and $\hat\xi_t^{(h)}=\sum_{s=0}^{h-1}
(\widehat C_s\otimes\hat\Sigma_\eta^{-1})\hat\zeta_t^{(h-s)}$,
where $C_0=\widehat C_0=I_n$.

The last component, $\hat \zeta^{H+2}_t$, has a different structure from the other components since it is multiplied by $(Z_t - \bar Z_T)$.
From the proof of Theorem \ref{thm:IVconsistency}, we know that
$\max_{t=1,\dots T}|Z_t - \bar Z_T| = O_p(T^{1/4})$. So, the consistency of the HAC components, including the linear combinations of $\hat \zeta^{H+2}_t$ can be proven similarly to the other components.
So, to save space,  we only consider vector $a=vec(a_0,a_1,\dots,a_{H},0_n)$ that puts 0 on the last $n$ components and all $a_i\in R^{n^2}$.
 \begin{align}
 \hat \zeta_t^\prime a a^\prime \hat \zeta_{t+\ell} -   \zeta_t^\prime a a^\prime  \zeta_{t+\ell} = &\sum_{h,s=0}^H (\hat\zeta^{(h)}_t-\zeta^{(h)}_t)' a_h a'_{s} \hat\zeta^{(s)}_{t+\ell} \notag\\
 & +  \sum_{h,s=0}^H ( \zeta^{(h)}_t)' a_h a'_{s}( \hat\zeta^{(s)}_{t+\ell}-\zeta^{(s)}_{t+\ell}) \notag\\
 &+  \sum_{h,s=0}^H (\hat\zeta^{(h)}_t-\zeta^{(h)}_t)' a_h a'_{s}( \hat\zeta^{(s)}_{t+\ell}-\zeta^{(s)}_{t+\ell}). \label{eq:estimationHACerror}
 \end{align}
The first  term after summation over $t$ becomes
\begin{align*}
     & \sum_{t=p}^{T-\ell}\sum_{h,s=0}^{H} (\hat\zeta^{(h)}_t-\zeta^{(h)}_t)' a_h a_{s} \hat\zeta^{(s)}_{t+\ell}  = \\
     &=    \sum_{t=p}^{T-\ell}\sum_{i,j=1}^k\sum_{h,s=0}^{H} e'_i(  \hat \eta_t  (\hat\eta^{(h)}_{t+h})'-\overline{\hat \eta_t (\hat\eta^{(h)}_{t+h})'}-\eta_t    \eta^{(h)}_{t+h}  + \E \eta_t  (\eta^{(h)}_{t+h})' )' a_{ih} \\
     &\times a_{js}'   ( \eta_{t+\ell} (\eta^{(s)}_{t+\ell+s})'- \E \eta_{t+\ell} (\eta^{(s)}_{t+\ell+s})')e_j\\
    & = \sum_{i,j=1}^k\sum_{h,s=0}^{H}    \sum_{t=p}^{T-\ell} e'_i(\hat \eta_t  \hat \eta^{(h)}_{t+h} -\eta_t    \eta^{(h)}_{t+h}) ' a_{ih} a_{js}'    ( \eta_{t+\ell} (\eta^{(s)}_{t+\ell+s})'- \E \eta_{t+\ell}  (\eta^{(s)}_{t+\ell+s})')e_j\\
     & - \sum_{i,j=1}^k\sum_{h,s=0}^{H}    e'_i (\overline{\hat \eta_t (\hat\eta^{(h)}_{t+h})'}  - \E \eta_{t}  (\eta^{(h)}_{t+h})' ) a_{ih}  \sum_{t=p}^{T-\ell}    a_{js}'    ( \eta_{t+\ell} (\eta^{(s)}_{t+\ell+s})'- \E \eta_t (\eta^{(s)}_{t+\ell+s})')e_j
\end{align*}
The last term is $O_p(1)$ by Lemma \ref{lem:residuals} and the FCLT for mixing processes \citep{de2000functional} applied to $\eta_{t+\ell} (\eta^{(s)}_{t+\ell+s})'- \E \eta_t   (\eta^{(s)}_{t+\ell+s})'$.

  Next, consider one element of the sum in the first term.  By \eqref{eq:etahdif1}-\eqref{eq:etahdif3} we get
\begin{align*}
  &  \sum_{t=p}^{T-\ell} e'_i(\hat \eta_t   (\hat\eta^{(h)}_{t+h})' -\eta_t    (\eta^{(h)}_{t+h})') a_{ih} \\
& =  \sum_{t=p}^{T-\ell}    a_{js}'    ( \eta_{t+\ell} (\eta^{(s)}_{t+\ell+s})'- \E \eta_{t+\ell}  (\eta^{(s)}_{t+\ell+s})')e_j e'_i \eta_{t}  \Tilde{X}_{t-1}^\prime    \Upsilon_T^{-1}   \Lambda_{T}^h a_{ih} \\
& + 2 \sum_{t=p}^{T-\ell}    a_{js}'    ( \eta_{t+\ell} (\eta^{(s)}_{t+\ell+s})'- \E \eta_{t+\ell}  (\eta^{(s)}_{t+\ell+s})')e_j e'_i \eta^{(h)}_{t+h}  \Tilde{X}_{t-1}^\prime    \Upsilon_T^{-1}   \Lambda_{T}^0 a_{ih}
\end{align*}
where $ \Lambda_{T}^h \bydef ( \Upsilon_T^{-1} \sum_{r=p}^{T-h} \Tilde{X}_{r-1} \Tilde{X}_{r-1}^\prime \Upsilon_T^{-1})^{-1}  (\sum_{r=p}^{T-h} \Upsilon_T^{-1}  \Tilde{X}_{r-1}(\eta^{(h)}_{r+h})') =O_p(1) $  by the argument in the proof of Lemma \ref{lem:residuals}.
Both terms in the r.h.s of the displayed formula above are $O_P(\sqrt{T-B_T})$ by the Cauchy-Schwarz inequality since, for example, for any constant vector $b$,
\begin{align*}
    & |\sum_{t=p}^{T-\ell}    a_{js}'    ( \eta_{t+\ell} (\eta^{(s)}_{t+\ell+s})'- \E \eta_{t+\ell}  (\eta^{(s)}_{t+\ell+s})')e_j e'_i \eta_{t}  \Tilde{X}_{t-1}^\prime    \Upsilon_T^{-1}   b  |\\
     & \leq \sqrt{b'\left(\sum_{t=p}^{T-\ell} \Upsilon_T^{-1} \Tilde{X}_{t-1}\Tilde{X}_{t-1}^\prime \Upsilon_T^{-1}\right)b}
     \sqrt{\sum_{t=p}^{T-\ell}\left(a_{js}'\left[\eta_{t+\ell}(\eta^{(s)}_{t+\ell+s})'-\E\eta_{t+\ell}(\eta^{(s)}_{t+\ell+s})'\right]e_j e_i'\eta_t\right)^2}.
\end{align*}
The r.h.s. is $O_p(\sqrt{T})$ by proof of Step 2 of Lemma \ref{lem:residuals} and Markov inequality implied by the existence of $6+\epsilon$ moments of the stationary mixing sequence $\eta_t$.

The sums of the remaining two terms in \eqref{eq:estimationHACerror} are also $ O_p(\sqrt{T})$, which can be shown using the same techniques.

\textbf{Step 3.}
Putting all bounds together gives us
\begin{equation*}
    \widehat\Omega_T-\widetilde\Omega_T = \frac{B_T}{T-B_T} O_p(\sqrt{T-B_T}) = O_p(\frac{B_T}{\sqrt{T}}).
\end{equation*}
Since by assumption of the theorem $B_T^2/T \to 0$, we get the desired result.

 \subsection{Proof of Theorem \ref{thm:bootstrap}, wild bootstrap}\label{sec:app_proof_bootstrap}

By construction, $ \sqrt{ T-p-H } (\hat\theta^s -\hat\theta) \sim N(0,\hat{\Omega}^*_T)$ for some matrix $\hat{\Omega}^*_T$ conditional on the data $y_1,...,y_T$.
Let us find the covariance matrix $\hat{\Omega}^*_T=E[(T-p-H)(\hat\theta^s -\hat\theta)^2|y_1,...,y_T]$. Consider
\begin{align}
   (T-p-H)\widehat{\Omega}^*_T &=   \sum_{t,t'=p}^{T-H} \hat\xi_t  \hat\xi'_{t'}  \E \nu^s_t \nu^s_{t'}  =  \sum_{t,t'=p}^{T-H} \hat\xi_t  \hat\xi'_{t'} K( \frac{|t-t'|}{B_T} )= \sum_{\ell=-B_T}^{B_T}\sum_{t=p+\ell}^{T-H} \hat\xi_t  \hat\xi'_{t+\ell} K( \frac{ \ell }{B_T} ) \notag\\
    &=   \sum_{\ell=-B_T}^{B_T}K( \frac{ \ell }{B_T} )(T-H-p-\ell) \overline{\hat\xi_t  \hat\xi'_{t+\ell} } \notag\\
     &=  (T-H-p )  \sum_{\ell=-B_T}^{B_T}K( \frac{ \ell }{B_T} )\overline{\hat\xi_t  \hat\xi'_{t+\ell} } -  \sum_{\ell=-B_T}^{B_T}\ell K( \frac{ \ell }{B_T} )   \overline{\hat\xi_t  \hat\xi'_{t+\ell} },\label{eq:DWBproof}
\end{align}
where $K(x) = \max\{1-|x|,0\}$ is a Bartlett kernel. Since second term in \eqref{eq:DWBproof} is $O_P( B^2_T)$, we get    $  \widehat{\Omega}^*_T = \widehat{\Omega}_T + O_P(\frac{B^2_T}{T}).$

Since $\frac{B^2_T}{T}\to 0$ , by Theorem \ref{thm:hac} we get  $\widehat{\Omega}_T \pTo \Omega $. So by Theorem \ref{thm:inferenceCH},
\begin{align*}
& \lim_{T\to\infty} \sup_{x\in R, \|a\|=1}  |P(\sqrt{T}a'(\hat \theta^s-\hat \theta)<x|y_1,...,y_T)-P(\sqrt{T}a'(\hat \theta- \theta)<x) |\\
\leq &
  \lim_{T\to\infty}\sup_{x\in R, \|a\|=1}   |\Phi( \frac{x}{\sqrt{a'\hat{Q}_Ta}})- \Phi(  \frac{x}{\sqrt{a'\Omega a}})| \\
  & +   \lim_{T\to\infty} \sup_{x\in R, \|a\|=1}  | P(\frac{\sqrt{T}a'(\hat \theta- \theta)}{\sqrt{a'\Omega a}}  <x) - \Phi(   \frac{x}{\sqrt{a'\Omega a}})| = 0 .
\end{align*}

\subsection{Proof of Theorem \ref{thm:identification}}\label{app:proof_ID}

The reduced-form shocks have the following representation for any $k=1,\dots,n,$
\begin{align*}
     \eta_{1t}\eta_{kt}&= b_{k1}b_{11}\varepsilon^2_{1t} + \sum_{i\neq j}b_{ki}b_{1j}\varepsilon_{it}\varepsilon_{jt}+ \sum_{i\neq 1} b_{ki}b_{1i} \varepsilon^2_{it} .\\
\end{align*}
Then, by Assumption \ref{ass:instuments},
\begin{align*}
   \beta_Z   &=  b_{\cdot1}b_{11}\E\varepsilon^2_{1t}(Z_t - \E Z_t) + \sum_{i\neq j}b_{\cdot i}b_{1j}\E\varepsilon_{it}\varepsilon_{jt}(Z_t - \E Z_t)+ \sum_{i\neq 1} b_{\cdot i}b_{1i} \E\varepsilon^2_{it}(Z_t - \E Z_t)\\
    &=  b_{\cdot 1}b_{11}\E(\varepsilon^2_{1t}-1)(Z_t - \E Z_t)  = b_{\cdot 1}b_{11} \rho_1.
\end{align*}

 The statement of the Theorem follows immediately.

\renewcommand{\refname}{References}

\clearpage
\hypersetup{pageanchor=false}
\setcounter{page}{1}
\renewcommand{\thepage}{S\arabic{page}}
\hypersetup{pageanchor=true}
\part*{Supplementary Appendix}
\appendix
\renewcommand{\thesection}{S.\Alph{section}}
\renewcommand{\theHsection}{supplementary.\Alph{section}}
\setcounter{figure}{0}
\renewcommand{\thefigure}{S\arabic{figure}}
\renewcommand{\theHfigure}{supplementary.\arabic{figure}}
\setcounter{table}{0}
\renewcommand{\thetable}{S\arabic{table}}
\renewcommand{\theHtable}{supplementary.\arabic{table}}

\section{Auxiliary lemma } \label{sec:app_aux_lemma}
\begin{lem}\label{lem:residuals}
Suppose that Assumptions   \ref{ass:innovations} and \ref{ass:SSW} hold. Consider the LS residuals $\hat \eta^{(H)}_{t}$ from $H-$step forward autoregression,
\begin{align}\label{eq:hlagVARApp}
     y_{t+H} &=  V^{(H)} \mu_t + \sum_{i=1}^p A^{(H)}_i y_{t-i} + \eta_{t+H}^{(H)},
\end{align}
where $\eta_{t}^{(H)} = \sum_{s=0}^H C_s \eta_{t-s} $. Then, for each $h=0,1,2,\dots,H,$, the following coupling holds:
    \begin{equation}
       \sum_{t=p}^{T-h} \hat \eta_t (\hat \eta^{(h)}_{t+h} )^\prime -    \sum_{t=p}^{T-h}  \eta_t (\eta^{(h)}_{t+h} )^\prime =  O_p(1) . \label{eq:residualSums}
    \end{equation}

\end{lem}

\begin{proof}

The proof proceeds in 3 steps. First, we show that \eqref{var1} can be decomposed into linearly independent univariate AR processes using Jordan decomposition. Then, we show that the LS estimators for individual processes have well-defined asymptotic distributions. Finally, we use the invariance of LS estimates with respect to linear transformations to estimate the difference in \eqref{eq:residualSums}.

\textbf{Step 1.}
Notice that $\mu_t=(1,t,\dots,t^k)'$ satisfies a first-order deterministic difference equation
\begin{equation*}
\mu_t = Q_\mu \mu_{t-1},
\end{equation*}
where $Q_\mu $ is a lower triangular matrix
\begin{equation*}
   Q_\mu \bydef \begin{bmatrix}
1 & 0 & 0 & 0 & \cdots & 0  & 0\\
1 & 1 & 0 & 0 & \cdots & 0  & 0\\
1 & 2 & 1 & 0 & \cdots & 0  & 0\\
\vdots &\vdots  &   &   &   & \vdots  & \vdots\\
1 & {k \choose 2 }  & \cdots & \cdots & \cdots &{k \choose k-1 } & 1
\end{bmatrix} .
\end{equation*}

Take any solution $W$  to the following linear matrix equation
\begin{equation*}
       W   =  \sum_{i=1}^p A_i   W  Q^i_\mu +  V.
\end{equation*}
Then  $y^\circ_t=y_t -  W\mu_t$ satisfies
\begin{equation*}
      y^\circ_t =   \sum_{i=1}^p A_i  y^\circ_{t-1} + \eta_t ,
\end{equation*}

We can now introduce $Y^{\circ}_t = vec(y^\circ_t,\dots,y^\circ_{t-p+1})$.
Then
\begin{align*}
    \begin{pmatrix}
        Y^{\circ}_t \\
        \mu_t
    \end{pmatrix} =    \begin{bmatrix}Q_y & 0_{np\times k}\\
    0_{k\times np}& Q_\mu\end{bmatrix} \begin{pmatrix}
        Y^{\circ}_{t-1} \\
        \mu_{t-1}
    \end{pmatrix} + \begin{pmatrix}
        U_{t} \\
         0
    \end{pmatrix}
\end{align*}
with the companion matrix
\begin{equation*}
    Q_y \bydef \begin{bmatrix}
 A_1 & A_2 &\dots & A_p  \\
I &0 & \cdots  &0\\
\dots &\dots &\dots &\dots \\
0 &0 & \cdots &I\\
\end{bmatrix}
\end{equation*}
and $ U_{t} = vec(\eta_t, 0,\dots,0)$.

Matrix  $Q_y  $   can be represented in Jordan normal form
\begin{align*}
    Q_y = P_{y}J_yP_y^{-1},
\end{align*}
where $P_{y}$ is an invertible matrix and $J_y$ is a block diagonal matrix.\footnote{See, for example, \cite{gantmakher2000theory-supp}. }
Then, following ideas in \cite{tsay1990asymptotic-supp} and \cite{sims1990inference-supp}, we get VAR(1) representation
\begin{align}
      \begin{pmatrix}
       P_y^{-1} Y^{\circ}_t \\
        \mu_t
    \end{pmatrix} =    \begin{bmatrix}J_y & 0_{np\times k}\\
    0_{k\times np}&Q_\mu\end{bmatrix}       \begin{pmatrix}
       P_y^{-1} Y^{\circ}_{t-1} \\
        \mu_{t-1}
    \end{pmatrix} + \begin{pmatrix}
        \tilde U_{t} \\
         0
    \end{pmatrix},\label{eq:VAR1jordan}
\end{align}
with $ \tilde U_{t} =  P_y^{-1} U_{t}$.
Matrix $J_y$ has a block-diagonal form $J_y= diag \{J_{y,1},\dots,J_{y,g}\}$,  where each block $J_{y,i}$ is a $n_i\times n_i  $ matrix
\begin{equation*}
    J_{y,i} =\begin{bmatrix}
\lambda_i & 1       & 0      & \cdots  & 0 \\
0       & \lambda_i & 1      & \cdots  & 0 \\
\vdots  & \vdots  & \vdots & \ddots  & \vdots \\
0       & 0       & 0      & \lambda_i & 1      \\
0       & 0       & 0      & 0       & \lambda_i
\end{bmatrix} .
\end{equation*}
 with $\lambda_i$ being  $i$-th eigenvalue  of $Q_y$ (not necessarily distinct)  with geometric multiplicity $n_i$.
 Suppose, without loss of generality, that the blocks $J_{y,i}$ are arranged in decreasing order according to the absolute values of the corresponding $\lambda_i$. If a unit root is present, then start with $\lambda_1=1$.
 By Assumption \ref{ass:SSW}, there are no roots larger than 1 in absolute value.

Now let us find a row permutation matrix $P_x$ that rearranges rows of $  P_y^{-1} Y^{\circ}_{t-1}$  to get
     \begin{equation*}
         \Tilde{X}_t \bydef  \begin{pmatrix}
     P_x  P_y^{-1} Y^{\circ}_t \\
        \mu_t
    \end{pmatrix}  \bydef F(L) \nu_t,
     \end{equation*}
where
\begin{equation}
    F(L) = \begin{bmatrix}
F_{11}(L) & 0 & 0 & 0 & \cdots & 0  & 0\\
0 & F_{22} & 0 & 0 & \cdots & 0  & 0\\
0 & F_{32} & F_{33} & 0 & \cdots & 0  & 0\\
0 & F_{42} & F_{43} & F_{44} & \cdots & 0  & 0\\
\vdots &\vdots  &   &   &   & \vdots  & \vdots\\
0 & F_{G,2}  & \cdots & \cdots & \cdots & F_{G,G}  & 0\\
0 & 0  & \cdots & \cdots & \cdots & 0 & I_k
\end{bmatrix} \label{eq:SSWmat}
\end{equation}
and $\nu_t = (\eta'_t,(\upsilon^{1}_t(\lambda_1))^\prime,\dots,(\upsilon^{g_1}_t(\lambda_1))^\prime,(\upsilon^{1}_t(\lambda_2))^\prime,\dots,(\upsilon^{g_2}_t(\lambda_2))^\prime,\dots, \\(\upsilon^{1}_t(\lambda_b))^\prime,\dots,(\upsilon^{g_b}_t(\lambda_b))^\prime, 1,t, \dots,t^{k})^\prime$ with $\upsilon^{j}_t(\lambda)\bydef\lambda\sum_{s=1}^t\upsilon^{j-1}_s(\lambda)$, $\upsilon^{0}_t(\lambda)\bydef\eta_t \lambda^t$,  $i=1,\dots,b$, $\lambda_i$ are unit roots   of \eqref{var1} and  $g_i$ are their corresponding maximum geometric multiplicities.
The blocks $F_{i,j}$ corresponding to two different unit-roots $\lambda_i$ and $\lambda_j$ are 0.
The lag matrix $F_{11}(L)$ is obtained by inversion of the stationary rows of \eqref{eq:VAR1jordan} and therefore $\sum^{\infty}_{j=0}j^g|F_{11j}|<\infty$ and $\sum_{j=0}^\infty F_{11j}F_{11j}^\prime$ is not singular.

Following \cite{sims1990inference-supp} and \cite{tsay1990asymptotic-supp}, one can arrange the properly chosen powers of $T^{1/2}$ in a diagonal matrix $\Upsilon_T$ such that $ \sum_{t=p}^{T} \Upsilon_T^{-1}  \Tilde{X}_{t-1}$ has a bounded variance (for example, for stationary components it is $T^{1/2}$, for components containing $\upsilon^{j}_t(\lambda_2)$ it is $T^{(1+j)/2}$ etc.).

\textbf{Step 2.} We would like to show now that for all $h=0,1,2,\dots,H$, there exists a diagonal matrix $\Upsilon_T$ such that
\begin{align*}
    & \sum_{r=p}^{T-h} \Upsilon_T^{-1}  \Tilde{X}_{r-1}(\eta^{(h)}_{r+h})^\prime = \sum_{s=0}^h  \sum_{r=p}^{T-h} \Upsilon_T^{-1}  \Tilde{X}_{r-1} C_s\eta_{r+h-s} \overset{d}{\to} \sum_{s=0}^h \phi^s,\\
   &   \Upsilon_T^{-1} \sum_{r=p}^{T-h} \Tilde{X}_{r-1}\Tilde{X}_{r-1}^\prime \Upsilon_T^{-1} \overset{d}{\to}\tilde V  ,
\end{align*}
where $\phi^s$ are tight random vectors and $\tilde V$ is  a tight random matrix.

Consider $h=0$ first. The proof for deterministic, stationary, and ordinary unit root components follows the proofs of Lemmas 1 and 2 in \cite{sims1990inference-supp} with a modification required to account for conditionally heteroscedastic shocks $\eta_t$.
Namely, the proof of Lemma 1 \cite{sims1990inference-supp} requires the replacement of the Functional CLT for sums of  $\upsilon^1_t$ in part (a) and Theorem 2.4(ii) from \cite{chan1988limiting-supp} in part (c) for case $p=1$.
Under Assumption \ref{ass:innovations}, Theorem 3.1 \citep{de2000functional-supp} applied to $\upsilon^1_t$ gives the required FCLT under the mixing processes assumption. Theorem 4.1 of \cite{de2000functional-supp} for any pair of components of $\upsilon^1_t$ replaces Theorem 2.4(ii) of \cite{chan1988limiting-supp}  for the mixing sequence case (the MDS assumption on $\eta_t$ ensures that we do not need to recenter the sums).
For complex unit roots, the limiting distribution is given in Theorem 4.4 of \cite{tsay1990asymptotic-supp}, which also requires substituting the FCLT for martingale differences with mixing white noise \citep[Theorem 4.1]{de2000functional-supp} and the strong law of large numbers used from \cite{chan1988limiting-supp} for MDS with Corollary 3.2 in \cite{rio2017asymptotic-supp}.
For our purposes, we do not need the explicit form of $\phi$ and $\tilde V$.

For $h>0$, the existence of the limiting distributions follows the proof of Theorem 4.5 of \cite{tsay1990asymptotic-supp}.

\textbf{Step 3.} The regression model \eqref{eq:hlagVARApp} can be written as
    \begin{equation*}
        y_{t+h} = \Theta^{(h)} X_{t-1} + \eta^{(h)}_{t+h},
    \end{equation*}
where $$\Theta^{(h)}  \bydef (A_1^{(h)}, A_2^{(h)} , \dots ,  A_p^{(h)} ,  V^{(h)})$$ and $$X_{t-1}\bydef
 vec( y_{t-1} ,y_{t-2} ,\dots ,y_{t-p},1,(t-1),\dots (t-1)^k). $$

Each period $t=p,\dots,(T-h)$ we have  \begin{equation*}
     \hat \eta^{(h)}_{t+h} -  \eta^{(h)}_{t+h}   =     (\Theta^{(h)}-\hat \Theta^{(h)}) X_{t-1}= - X_{t-1}^\prime(\sum_{r=p}^{T-h} X_{r-1}X_{r-1}^\prime)^{-1}\sum_{r=p}^{T-h} X_{r-1} \eta_{r+h}^{(h)}.
\end{equation*}
By Step 1, there exists a linear transformation $P$ such that  $X_{t} = P^{-1}\Tilde{X}_t$. So we can get a representation
\begin{align*}
     &X_{t-1}^\prime(\sum_{r=p}^{T-h} X_{r-1}X_{r-1}^\prime)^{-1}\sum_{r=p}^{T-h} X_{r-1} \eta^{(h)}_{r+h} \\  =& \Tilde{X}_{t-1}^\prime  (P^{-1})' (P^{-1}\sum_{r=p}^{T-h} \Tilde{X}_{r-1} \Tilde{X}_{r-1}^\prime(P^{-1})')^{-1}\sum_{r=p}^{T-h} (P^{-1}) \Tilde{X}_{r-1}\eta^{(h)}_{r+h} \\
    =&   \Tilde{X}_{t-1}^\prime    ( \sum_{r=p}^{T-h} \Tilde{X}_{r-1} \Tilde{X}_{r-1}^\prime )^{-1}\sum_{r=p}^{T-h}  \Tilde{X}_{r-1}\eta^{(h)}_{r+h} .
\end{align*}

Taking weighted sums over $t$ using  matrix $\Upsilon_T$ from Step 2, we get representations
{\small
\begin{align}
( \sum_{t=p}^{T-h}       (\hat \eta^{(h)}_{t+h} -  \eta^{(h)}_{t+h} ) \eta_{t}^\prime )^\prime &=\sum_{t=p}^{T-h}  \eta_{t}  \Tilde{X}_{t-1}^\prime    \Upsilon_T^{-1}    ( \Upsilon_T^{-1} \sum_{r=p}^{T-h} \Tilde{X}_{r-1} \Tilde{X}_{r-1}^\prime \Upsilon_T^{-1})^{-1}  \sum_{r=p}^{T-h} \Upsilon_T^{-1}  \Tilde{X}_{r-1}(\eta^{(h)}_{r+h})^{\prime}  , \label{eq:etahdif1}\\
(\sum_{t=p}^{T-h}(\hat \eta_t-\eta_t) (\eta^{(h)}_{t+h})^\prime )^\prime & = \sum_{t=p}^{T-h}   \eta^{(h)}_{t+h}  \Tilde{X}_{t-1}^\prime    \Upsilon_T^{-1}    ( \Upsilon_T^{-1} \sum_{r=p}^{T} \Tilde{X}_{r-1} \Tilde{X}_{r-1}^\prime \Upsilon_T^{-1})^{-1}  \sum_{r=p}^{T} \Upsilon_T^{-1}  \Tilde{X}_{r-1}\eta_r^{\prime}  , \notag\\
\sum_{t=p}^{T-h} (\hat \eta_t -\eta_t) (\hat \eta^{(h)}_{t+h}-  \eta^{(h)}_{t+h})^  \prime & = \sum_{t=p}^{T-h}   \eta^{(h)}_{t+h}  \Tilde{X}_{t-1}^\prime    \Upsilon_T^{-1}   ( \Upsilon_T^{-1} \sum_{r=p}^{T} \Tilde{X}_{r-1} \Tilde{X}_{r-1}^\prime \Upsilon_T^{-1})^{-1}  \sum_{r=p}^{T} \Upsilon_T^{-1}  \Tilde{X}_{r-1}\eta_r^{\prime}. \label{eq:etahdif3}
\end{align}
}
By Step 2, these sums converge in distribution to some random matrices. Convergence in distribution implies that these sums are tight random sequences   \cite[Lemma 4.8]{kallenberg1997foundations-supp}.  This means that the sums are $O_p(1)$.

To complete the proof, notice
\begin{align*}
     & \sum_{t=p}^{T-h} \hat \eta_t (\hat \eta^{(h)}_{t+h} )^\prime-  \sum_{t=p}^{T-h}  \eta_t (\eta^{(h)}_{t+h} )^\prime   \\
     &= \sum_{t=p}^{T-h} \eta_t(\hat \eta^{(h)}_{t+h}-  \eta^{(h)}_{t+h})^  \prime+   \sum_{t=p}^{T-h}(\hat \eta_t-\eta_t) (\eta^{(h)}_{t+h})^\prime+  \sum_{t=p}^{T-h} (\hat \eta_t -\eta_t) (\hat \eta^{(h)}_{t+h}-  \eta^{(h)}_{t+h})^  \prime \\
     &= O_p(1).
\end{align*}

\end{proof}
\newpage
\section{Additional figures}

\begin{figure}[!htbp]
\centering
\vspace{-0.4cm}
\includegraphics[width=0.7 \linewidth]{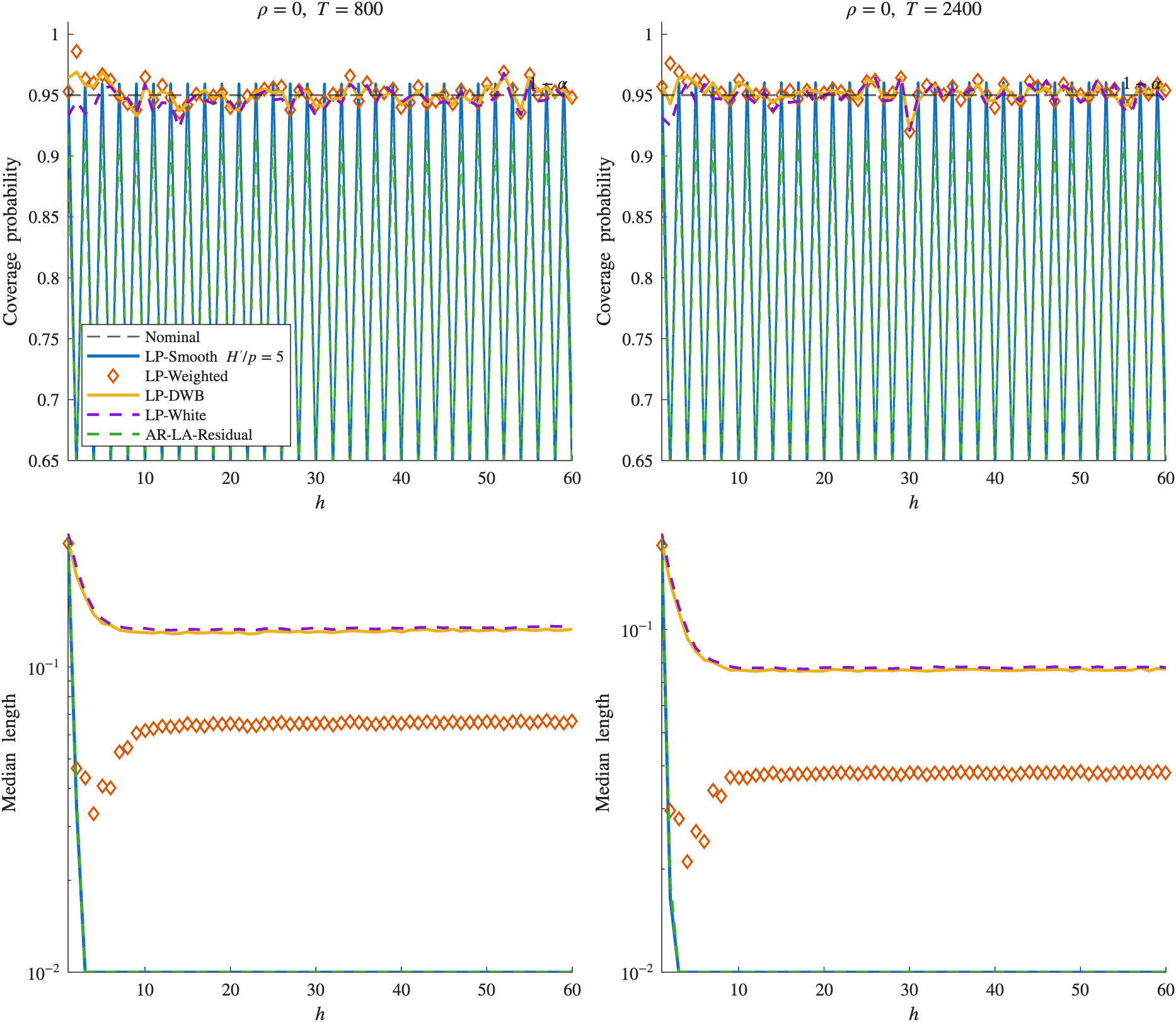}
\vspace{-0.4cm}
\caption{ {Coverage probability and median interval length in AR(1) design
with asymmetric conditional heteroscedasticity ($\gamma=1$) for large sample ($T=2400$) and $\rho=0$.}}
\label{fig:MC_rho_0}
\vspace{0.3cm}
\includegraphics[width=0.7 \linewidth]{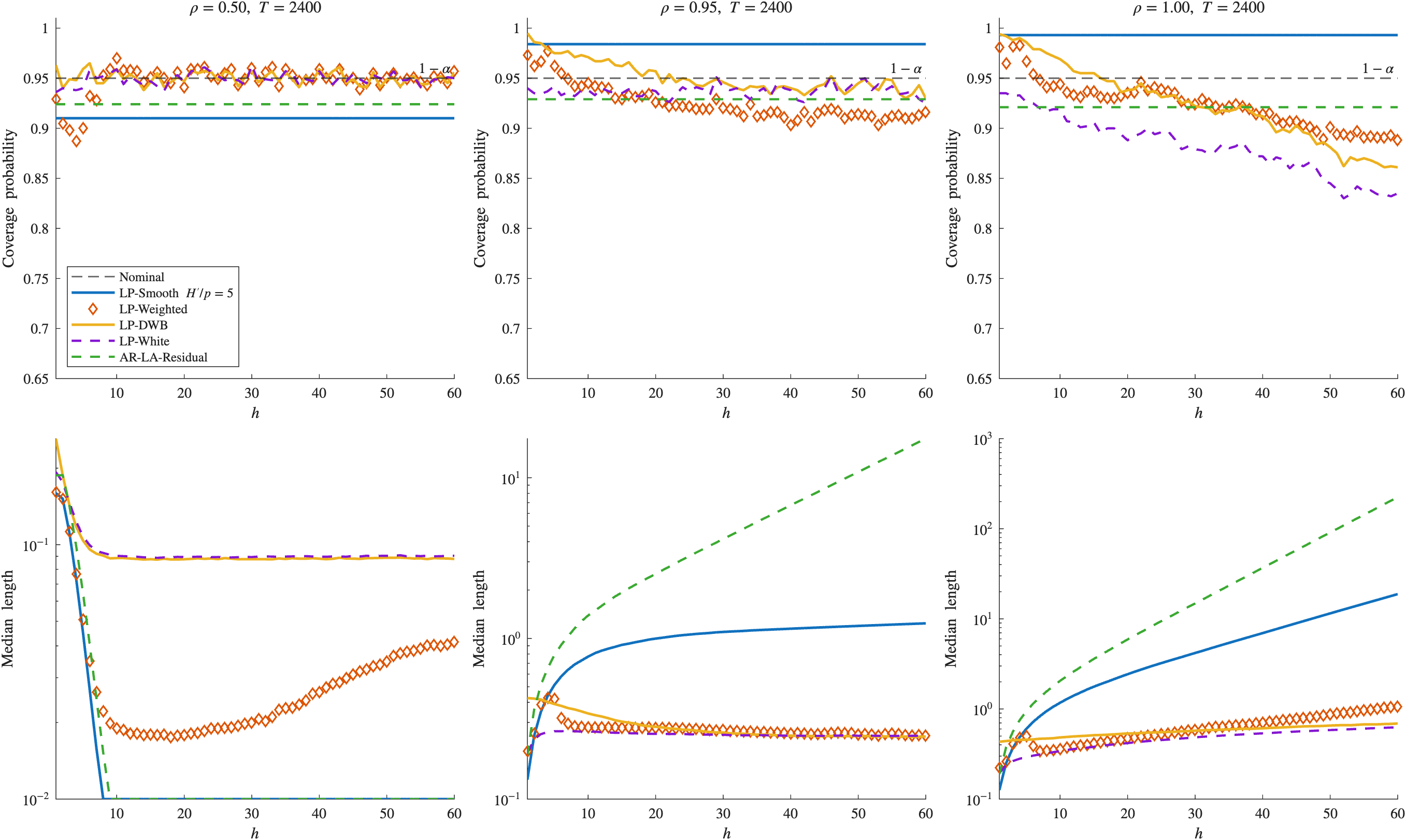}
\vspace{-0.2cm}
\caption{ {Coverage probability and median interval length in AR(1) design
with homoscedastic innovations ($\gamma=0$) in large sample $T=2400$.}}
\label{fig:MC_gamma_0}
\end{figure}

\newpage

\begin{figure}[!htbp]
\centering

\includegraphics[width=0.95\linewidth]{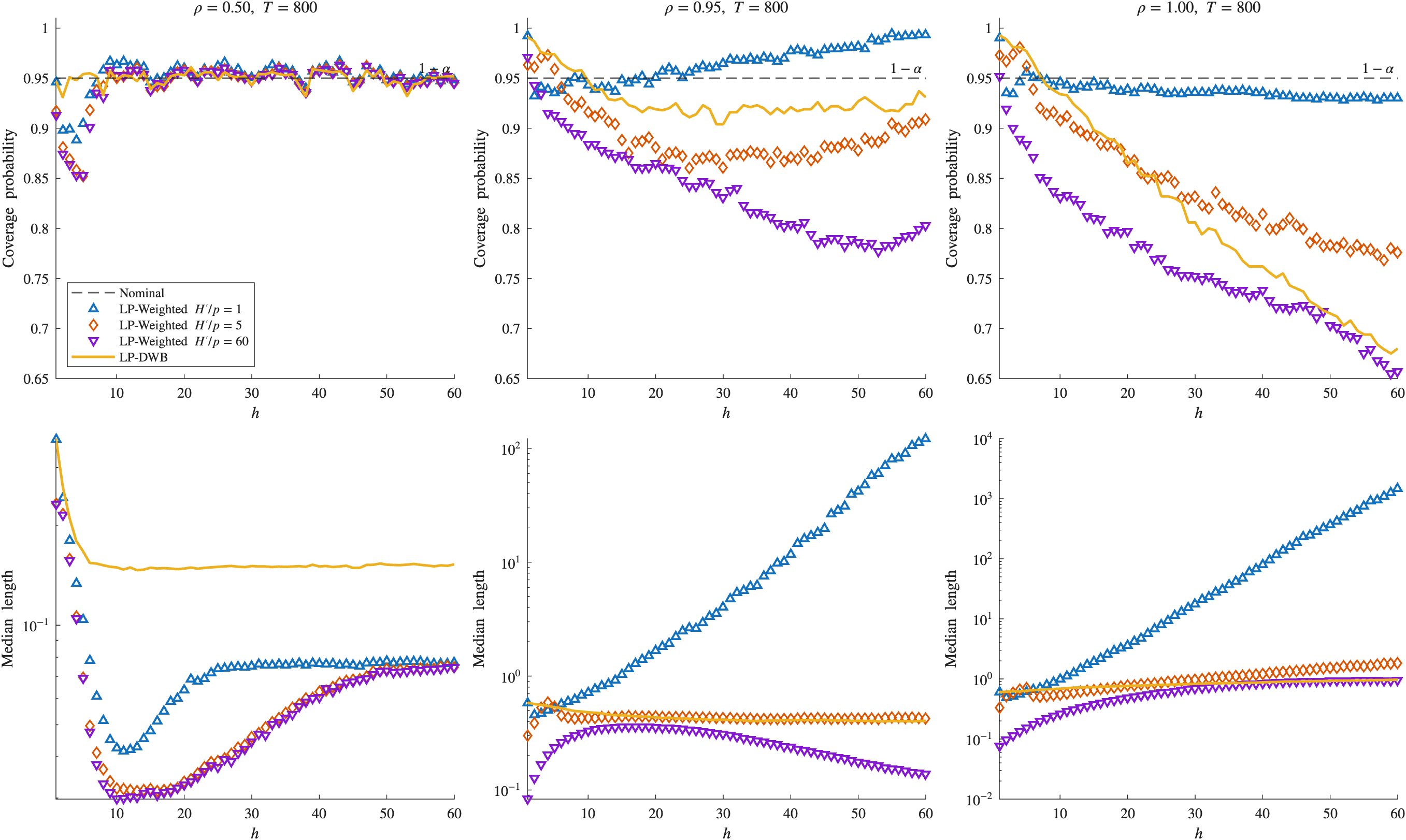}
\vspace{-0.25cm}

\caption{ {Coverage probability and median interval length in AR(1) design
with asymmetric conditional heteroscedasticity ($\gamma=1$) for medium size sample ($T=800$).}}
\label{fig:MC_hprime_800}

\vspace{0.25cm}

\includegraphics[width=0.95\linewidth]{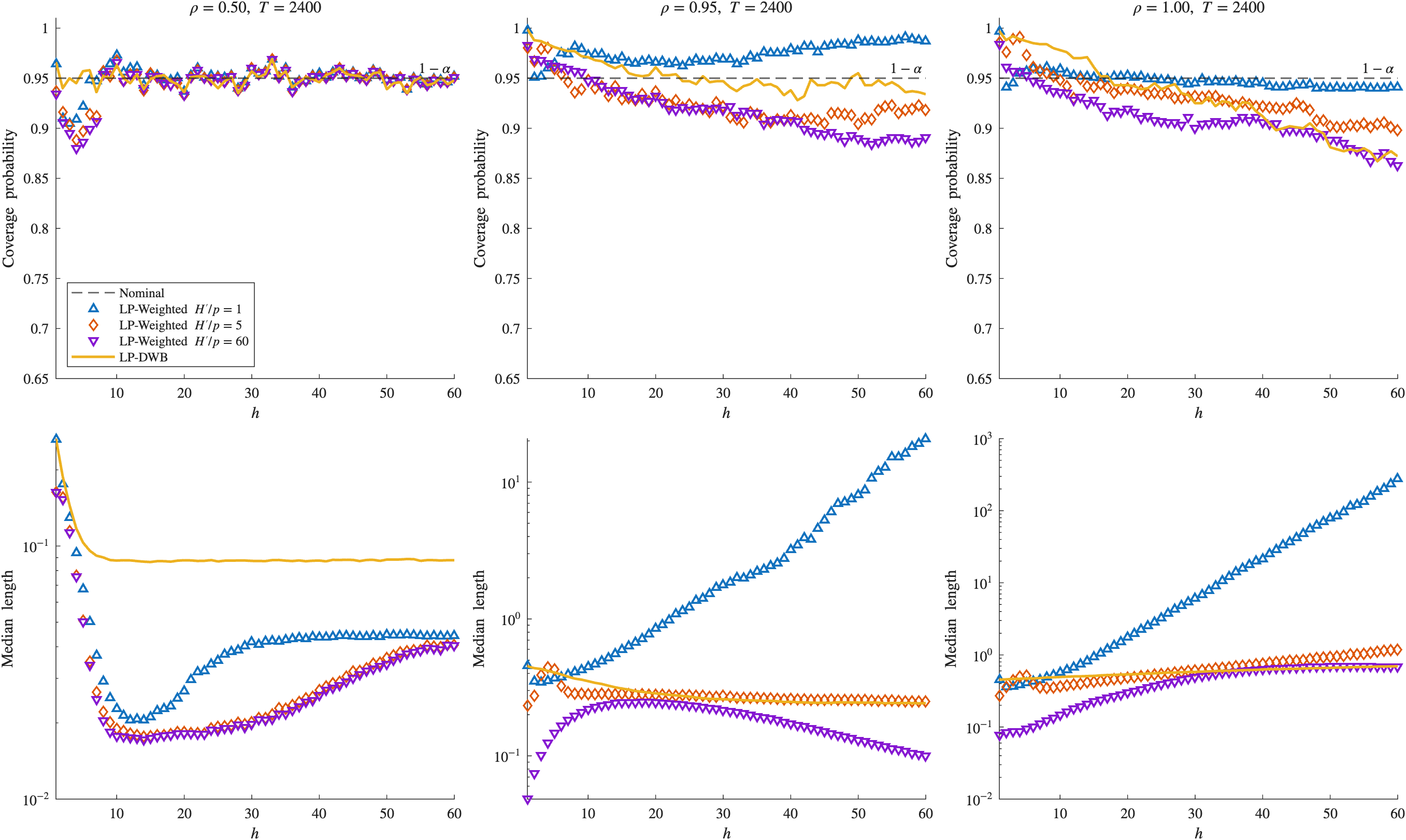}
\vspace{-0.25cm}

\caption{ {Coverage probability and median interval length in AR(1) design
with asymmetric conditional heteroscedasticity ($\gamma=1$) for large sample ($T=2400$).}}
\label{fig:MC_hprime_2400}

\end{figure}

\section{Simulation experiment based on a DSGE model of the U.S. economy} \label{sec:dsge}

In this section, we assess the finite-sample performance of the proposed VIV local projection procedure and the associated dependent wild bootstrap confidence intervals in data generated by a well-known medium-scale New Keynesian model with a wide range of nominal and real rigidities of \citet{smets2007shocks-supp} (SW hereafter).
The goal of this exercise is twofold.
First, we illustrate that VIV moments can recover structural impulse responses  to a monetary policy shock in a realistic model. Second, we compute the coverage and average length of the confidence sets based on the dependent wild bootstrap in a misspecified VAR model when the true model is VARMA.

The state vector in the full linearized DSGE model has a dimension of about 20 variables.
As in the original paper, we introduce seven structural macroeconomic shocks into the model, including the monetary policy shock.
Typically,  the model is estimated using the Kalman filter, with the observation vector consisting of seven variables (the number of variables cannot be less than the number of shocks): the log difference of real GDP, real consumption, real investment, real wage, log hours worked, the log difference of the GDP deflator, and the federal funds rate.
The filter is used to estimate the linear state space model that corresponds to the solution of the log-linearized DSGE models.

The unobserved state vector follows a SVAR model. Since the observation vector has a smaller dimension than the state vector (7 vs. 20), the observation vector is not generally a finite-order VAR.
However, in a fairly general case, the vector of observed variables can be represented as a VARMA(3,2) process, and, assuming the invertibility of the MA part, a VAR of finite order can act as an approximation for the corresponding  VAR($\infty$) \citep{morris2016varma-supp}.

In this exercise, we will approximate the underlying VARMA(3,2) process with VAR(p), $p=8$ model to mimic 2 years of lagged quarterly data.
Since the true DGP is  only an \emph{approximate} VAR(8) process,  the proposed inference procedures are only approximately correct.
Hence, this design can be used to study the finite sample properties of DWB in misspecified VAR models.

We generate samples of various lengths from the log-linearized SW DSGE model, calibrated using the posterior mean estimates obtained in \citep{smets2007shocks-supp}.
The only exception is the standard deviation parameter of the monetary policy shock, for which we introduce two regimes described by a Markov chain.
We consider a symmetric case where the probability of remaining in one of the regimes is 70\%, and the probability of switching to the other regime is 30\%.
We introduce an auxiliary instrumental variable $Z_t$, which takes the value 0 in the first regime and 1 in the second.
The value of the standard deviation of the monetary policy shock, $\sigma_{r,t}$, is set according to the following equation:
\begin{equation*}
    \sigma_{r,t} =\sigma_{r,0}(1+3Z_{t})
\end{equation*}
where $\sigma_{r,0}$ is set to the value of the posterior mean estimate for the standard deviation of the monetary policy shock in  \citet{smets2007shocks-supp}. This specification assumes that the standard deviation of the monetary shock increases by a factor of four in the second regime.
We assume that the econometrician observes the binary process $Z_t$ and uses it as a VIV for the monetary policy shock.

Next, using the simulated dataset of seven observable variables and the instrument $Z_t$, we apply our identification method to extract the structural monetary policy shock and compute the corresponding IRFs.
 We include in the vector of observed variables in the local projection model with 8 lags: log deviations from the steady-state level for real GDP, household consumption, investment, labor, wages, the GDP deflator index,  and the nominal interest rate.\footnote{ We use the GDP deflator index rather than growth rates in the econometric exercise to emphasize that the method can handle both stationary and nonstationary variables.}
 Thus, the dataset consists of six stationary variables and one nonstationary variable.

  Figure  \ref{fig:SW500_raw500} shows the IRFs for the seven variables with the corresponding DWB 95\% confidence sets (pointwise and sup-$t$).
We compare the IRFs generated from the theoretical SW model (VARMA) to those obtained from the estimated VAR model to validate the accuracy and effectiveness of our method in capturing the dynamics prescribed by the SW framework.
The contemporaneous impulse responses identified with VIV are well-aligned with the true theoretical IRFs, which validates our identification and estimation approach.
For subsequent horizons, the true IRFs lie within joint confidence bounds for all variables and across all horizons.

\begin{figure}[!htbp]
\centering
\vspace{-0.3cm}
\includegraphics[width=1\linewidth]{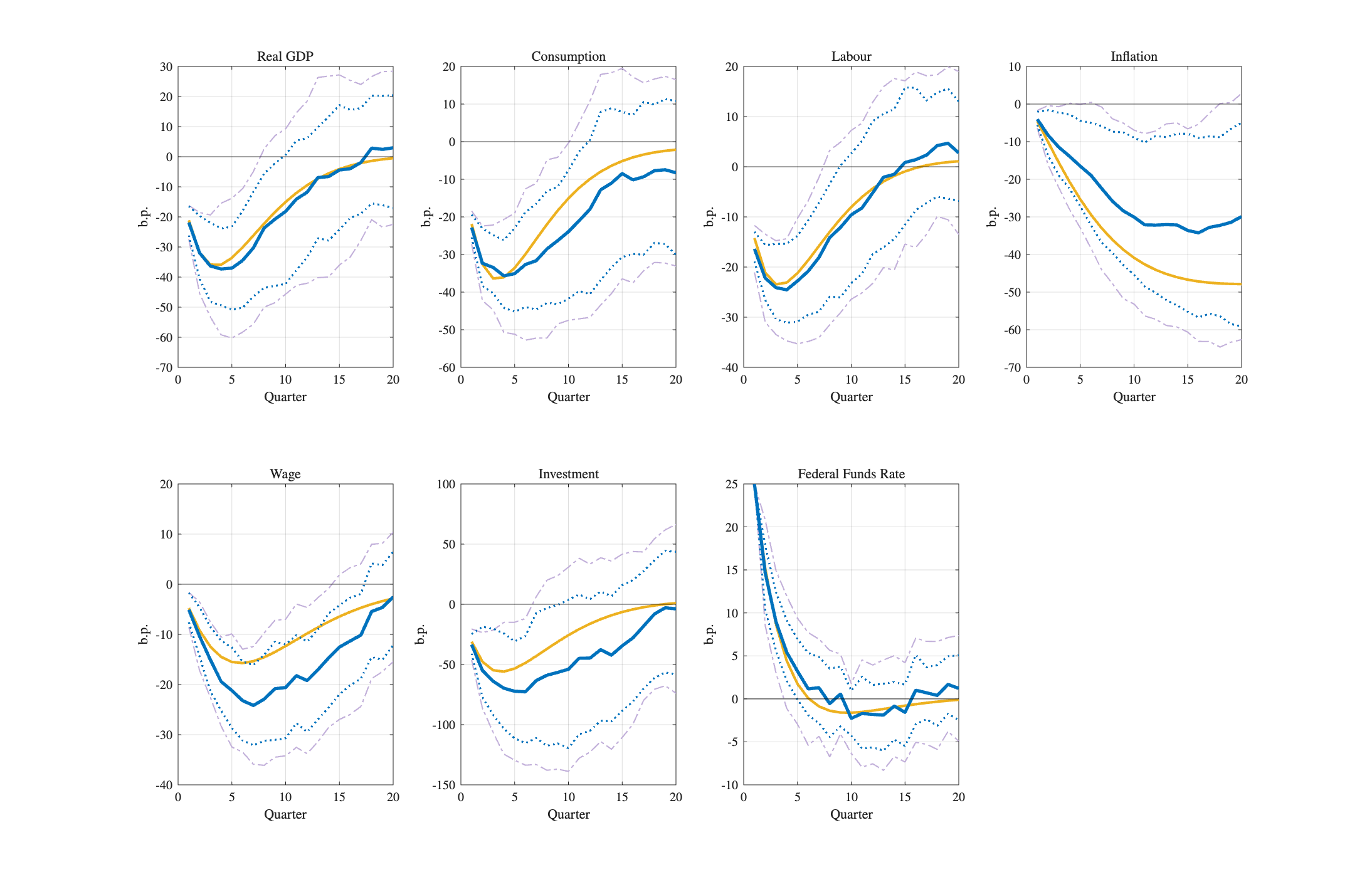}

\vspace{-0.4cm}

\caption{IRFs to a monetary policy shock identified through VIV, $T=500$. LP estimators (blue lines) with corresponding 95\% pointwise (blue dots) and simultaneous sup-$t$ bands (dot-dash lines), true IRFs from the SW model (gold lines).}
\label{fig:SW500_raw500}

\vspace{0.3cm}

\includegraphics[width=1\linewidth]{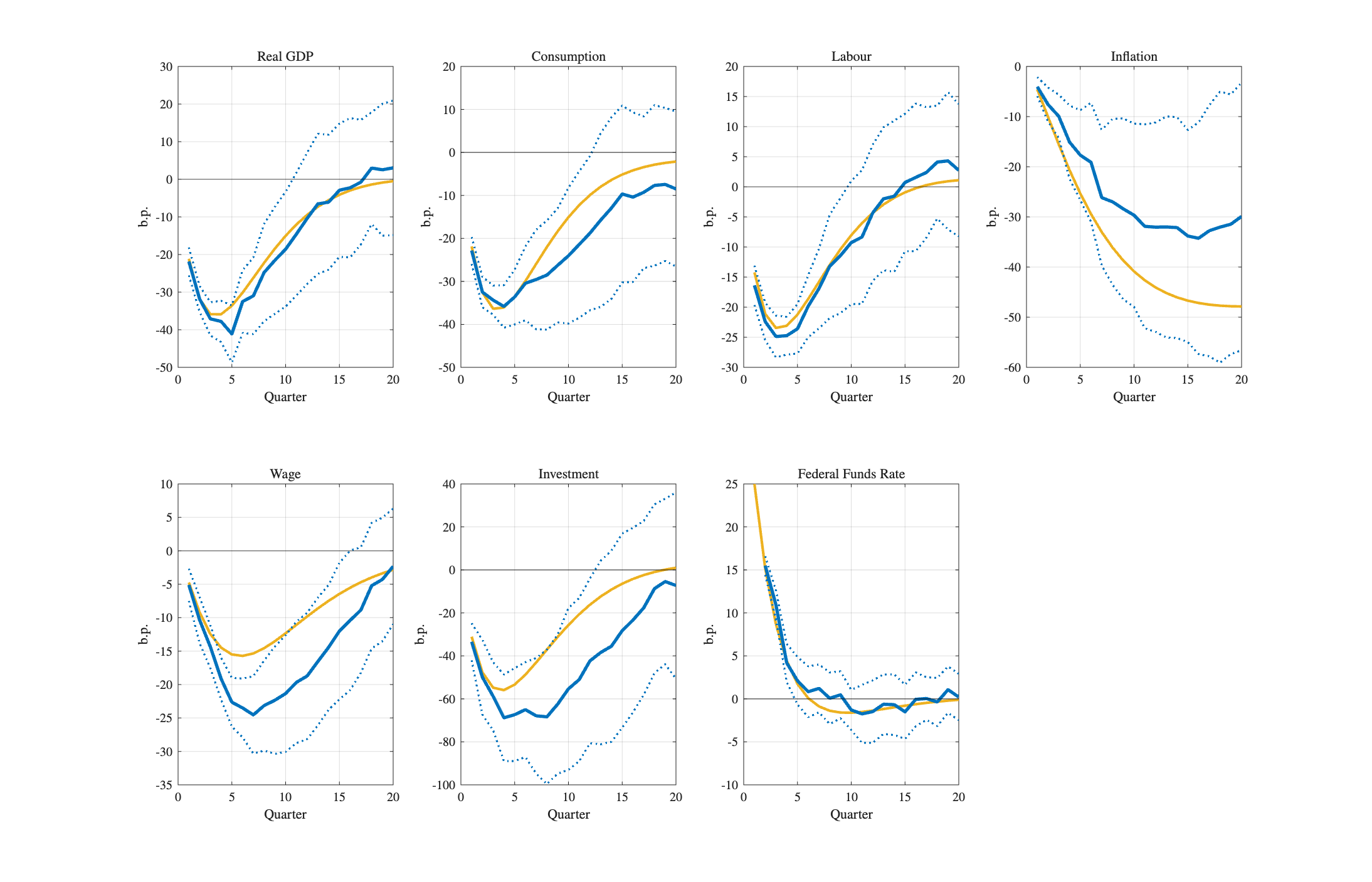}

\vspace{-0.4cm}

\caption{IRFs to a monetary policy shock identified through VIV, $T=500$. Weighted LP estimators (blue lines) with corresponding 95\% pointwise bands (blue dots), true IRFs from the SW model (gold lines).}
\label{fig:SW500_weighted500}

\end{figure}

The weighted IRF estimator that matches the VAR(8)-based  and the corresponding LP IRF estimators for 20 horizons is shown in Figure \ref{fig:SW500_weighted500}.
Optimal weighting preserves the approximate shape of LP-estimates while reducing the width of the corresponding confidence sets.
As we increase the sample size to 1500 periods, the LP estimates get even closer to the true values.\footnote{The corresponding figure is available upon request.}
This result shows that LP estimators based on finite VAR(8) can reasonably approximate the underlying VARMA(3,2) model.
Moreover, the illustrative figures show that the simultaneous confidence bounds successfully cover the true IRFs in all the figures, including the nonstationary GDP deflator index

Of course, one cannot judge the quality of the confidence sets based on a single sample. Next, we conduct a small scale Monte Carlo study  to evaluate the finite-sample coverage properties of the proposed confidence intervals more systematically.
To do so, we simulate a much longer realization from the SW DSGE model and partition it into 100 non-overlapping subsamples of equal length, 500.
In each subsample, we re-estimate the local projection model, construct the corresponding confidence intervals, and compare them to the theoretical impulse responses implied by the DSGE model.
This allows us to compute empirical coverage rates, defined as the fraction of subsamples in which the theoretical impulse response lies inside the estimated confidence interval, as well as the corresponding average interval lengths.

Supplementary Appendix Table \ref{tab:dsge_mc_var_h_selected} reports coverage probabilities and average interval lengths for selected representative variables at selected horizons.
The Monte Carlo results complement the representative figures by replacing single-realization evidence with repeated-sampling evidence on finite-sample coverage.
The raw procedure typically achieves better coverage at short horizons, with wider intervals, while the smoothed and weighted procedures often shorten intervals at some cost to coverage.
This tradeoff is more pronounced for stationary variables.
The only unit root variable, GDP deflator index, in contrast, demonstrates almost uniform improvement in both coverage and the average length of confidence sets.

\begin{table}[!ht]
\centering

\caption{Monte Carlo coverage and average interval length by variable and selected horizons in DSGE simulations based on the SW model.}
\label{tab:dsge_mc_var_h_selected}

\begin{tabular}{llcccccc}
\hline
 & & \multicolumn{3}{c}{Coverage probabilities} & \multicolumn{3}{c}{Average lengths} \\
\cline{3-5}\cline{6-8}
Variable & Hor.& Raw & Smoothed & Weighted & Raw & Smoothed & Weighted \\
\hline
Federal Funds Rate & 1  & 1.00 & 1.00 & - & 0.00  & 0.00  & -   \\
  & 5  & 0.87 & 0.62 & 0.64 & 6.73  & 4.42  & 4.48  \\
  & 10 & 0.92 & 0.85 & 0.88 & 6.64  & 6.15  & 5.44  \\
  & 15 & 0.93 & 0.78 & 0.80 & 6.85  & 5.53  & 5.13  \\
  & 20 & 0.94 & 0.91 & 0.89 & 6.83  & 5.57  & 5.05  \\
  \hline
Real GDP           & 1  & 0.93 & 0.96 & 0.96 & 9.20  & 9.31  & 9.31  \\
            & 5  & 0.93 & 0.71 & 0.75 & 28.16 & 13.89 & 14.22 \\
            & 10 & 0.90 & 0.85 & 0.86 & 32.37 & 26.52 & 26.76 \\
            & 15 & 0.86 & 0.95 & 0.87 & 34.09 & 35.95 & 33.06 \\
           & 20 & 0.88 & 0.93 & 0.86 & 35.76 & 38.56 & 34.92 \\
           \hline

GDP deflator index          & 1  & 0.92 & 0.96 & 0.96 & 3.49  & 3.40  & 3.40  \\
           & 5  & 0.97 & 0.91 & 0.93 & 22.64 & 15.09 & 15.73 \\
          & 10 & 0.91 & 0.90 & 0.90 & 37.54 & 35.76 & 34.22 \\
          & 15 & 0.87 & 0.92 & 0.89 & 48.53 & 55.85 & 46.91 \\
           & 20 & 0.82 & 0.91 & 0.84 & 56.55 & 76.94 & 56.08 \\

\hline
\end{tabular}

\vspace{0.3cm}
\begin{minipage}{0.92\textwidth}
 \footnotesize
\textit{Notes.}
Federal Funds Rate IRF is normalized to 1. The results are based on $N_{MC}=100$ simulations of time series  with length  $T=500$ each.
The omitted variables and horizons have comparable coverage probabilities. Raw, Smoothed, and Weighted refer to the three versions of LP estimators. Smoothing is done over 20 horizons.
\end{minipage}
\end{table}

Overall, our results indicate that the proposed VIV local projection procedure paired with DWB can recover the DSGE-implied monetary policy impulse responses with good coverage in empirically realistic finite samples, even though the observable process is generated by a VARMA model rather than by an exact finite-order VAR.
This result suggests that the robustness theory for LP estimators established  in  \cite{Montiel2024double-supp} for stationary VARMA models carries over to nonstationary models.
In addition to the standard LP estimates, we find that  the  model-based efficient weighted estimator shows robustness to model misspecification at long horizons but can result in visible biases for shorter horizons for stationary variables.
One natural and promising way to improve the performance of weighted estimates for stationary variables at short horizons is to perform smoothing using the VARMA(3,2) model instead of  the VAR(8) model.
We leave this extension for future research.

\section{Additional materials for empirical application}\label{sec:app_monetary_policy}

\subsection{Data sources}

Table \ref{tab:data_sources} describes the variables and respective sources used in the main specification. The monthly series for real GDP and the GDP deflator are interpolated from quarterly data using the \cite{chow1971best-supp} methodology. We use the industrial production index for the real GDP interpolation and the consumer price index for the GDP deflator.

\begin{table}[h]
\centering
\caption{Description of Variables and Data Sources}

\begin{tabular}{lccl}

\toprule
\textbf{Variable} & \textbf{Source} & \textbf{Seasonally } & \textbf{Mnemonics} \\
                  &                 &  \textbf{Adjusted}   &                         \\
\midrule
Real GDP & FRED & Yes & GDPC1 \\
GDP Deflator & FRED & Yes & GDPDEF  \\
Total Reserves & FRED & Yes & TRARR \\
Nonborrowed Reserves & FRED & Yes &  BOGNONBR \\
Federal Funds Rate & FRED & No & FEDFUNDS  \\
Commodity Price Index & S\&P Dow Jones & No & DJCIT  \\
Consumer Price Index & FRED & Yes & CPIAUCSL  \\
Industrial Production Index & FRED & Yes & INDPRO  \\
Market Yield on U.S. Treasury Securities & FRED & No & GS1 \\
 at 1-Year Constant Maturity  &   &   &  \\
\bottomrule
\end{tabular}
\label{tab:data_sources}

{\footnotesize {\it Notes:} Time series are available on the following website, \href{https://fred.stlouisfed.org/series/}{https://fred.stlouisfed.org/series/} and \\ \href{https://www.spglobal.com/spdji/en/indices/commodities/dow-jones-commodity-index/} {https://www.spglobal.com/spdji/en/indices/commodities/dow-jones-commodity-index/} }
\end{table}

\subsection{Instrument construction}

\subsubsection{Baseline instruments}

As a baseline instrument (see Figure \ref{fig:IRF_Rigobon} in the main text), we consider the number of FOMC meetings held in a given month, including telephone conferences and unscheduled meetings.\footnote{We construct instrumental-variable series using the FRASER website of the Federal Reserve Bank of St. Louis https://fraser.stlouisfed.org/title/federal-open-market-committee-meeting-minutes-transcripts-documents-677?browse=1980s} As a second instrument we use the number of meeting-days in a given month.\footnote{Results with this instrument are available upon request.} Figures \ref{fig:FOMC_events} and \ref{fig:FOMC_days} report both instruments for all meetings and for scheduled meetings only.
\begin{figure}[!ht]
\begin{center}
\includegraphics[width=0.9\linewidth]{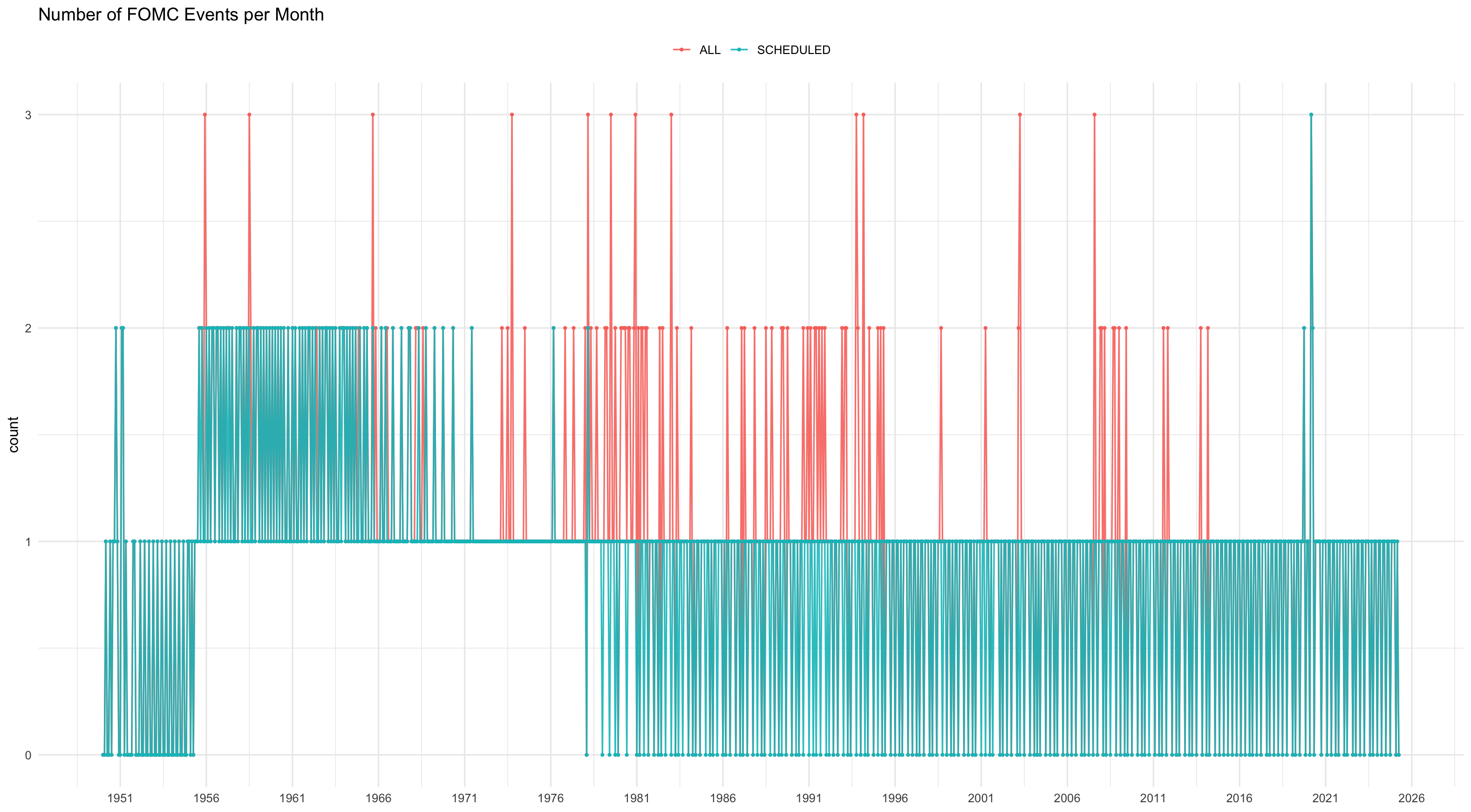}
\end{center}
\caption{Number of FOMC events per month.}\label{fig:FOMC_events}
{\footnotesize {\it Notes:} The blue series counts scheduled meetings; the red series counts all events, including unscheduled meetings and conference calls. Sample 1950 - 2025.}

\begin{center}
\includegraphics[width=0.9\linewidth]{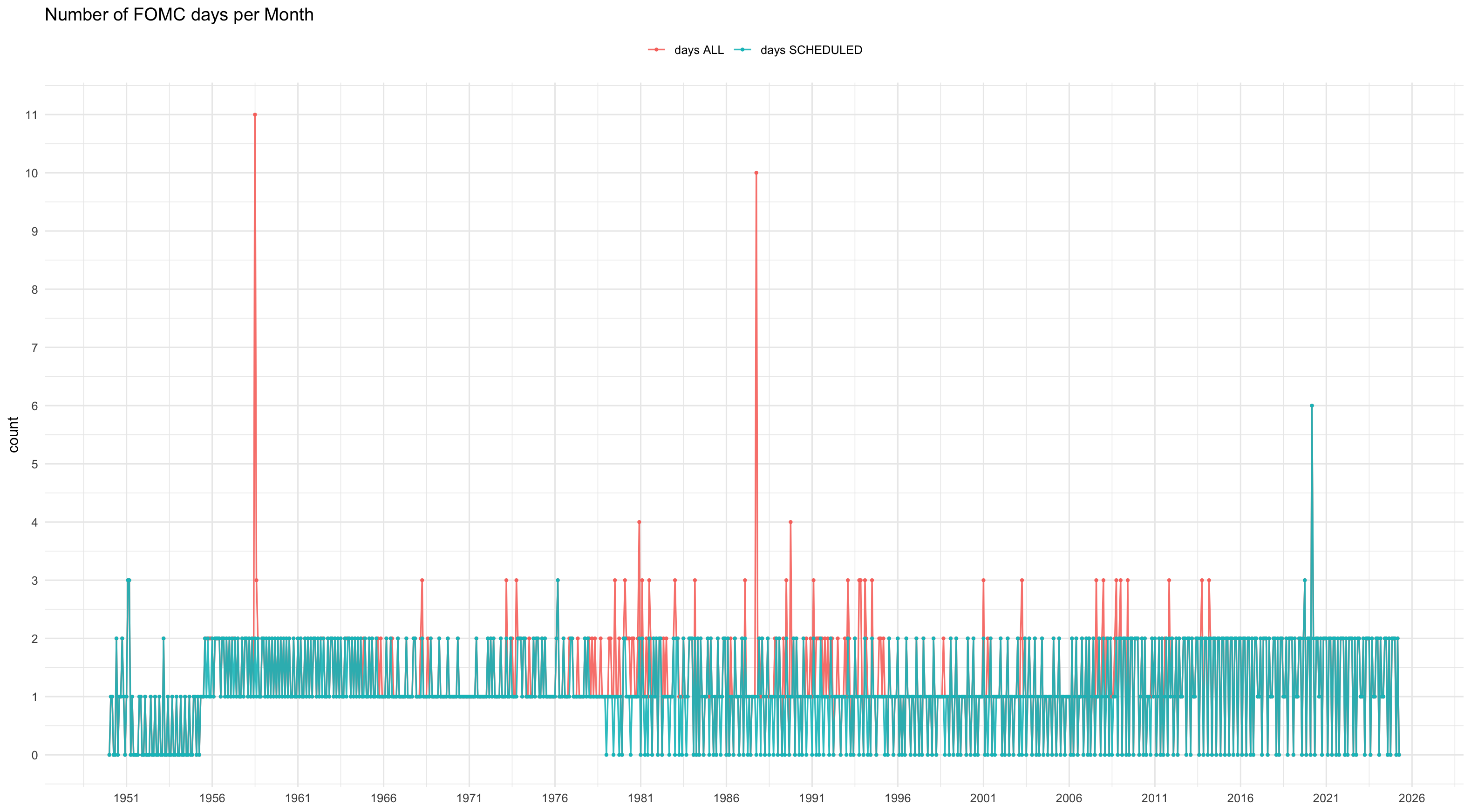}
\end{center}
\caption{Number of FOMC meeting-days per month.}\label{fig:FOMC_days}
{\footnotesize {\it Notes:} The blue series counts scheduled meeting-days; the red series counts all meeting-days, including unscheduled events and conference calls. Sample 1950-2025.}
\end{figure}

\subsubsection{Volume-surprise instruments}

We measure meeting-related surprises in equity-market trading activity using trading-day windows around FOMC meetings. Let $t$ index trading days, and let
\[
v_t = \log(\mathrm{Vol}_t),
\qquad
d_t = \log(\mathrm{Vol}_t \cdot P_t),
\]
where $\mathrm{Vol}_t$ denotes S\&P~500 share volume and $P_t$ the closing price from Yahoo Finance. We compute surprises for $x_t \in \{v_t,d_t\}$ and report both. We assemble two meeting calendars: \textbf{ALL} (all meetings) and \textbf{SCHED} (scheduled ex ante). Calendar days are mapped to trading days. In the few cases where a meeting falls on a non-trading day (weekend, holiday, or exchange closure), we exclude that event from the equity-based surprise calculation.\footnote{In our sample, there are 10 such instances; examples include 1952-03-01 (Saturday), 1965-11-02 (NYSE closed for Election Day), 2001-09-13 (post-9/11 exchange closure), and 2020-03-15 (Sunday).}
If a meeting spans consecutive trading days, we collapse them into a \emph{meeting block} $[s_b,e_b]$ of consecutive trading indices with the meeting indicator equal to one.

\paragraph{Pre/post windows and placement.}
For each block $b$ we form pre- and post-windows. We allow two anchoring conventions:
(i) \emph{final-day anchoring} \textbf{final}, where the anchor is the last day $e_b$; and
(ii) \emph{block anchoring} \textbf{block}, where the anchor spans the full block $[s_b,e_b]$.
The pre-window is taken \emph{before} the anchor (we exclude the anchor from the pre window by default); the post window runs from the anchor through $e_b+\texttt{post}$. We enforce non-overlapping windows by construction. To avoid forward-looking assignment, we \emph{place} the surprise at the end of the post window, $e_b+\texttt{post}$.\footnote{For robustness, we also place it at the anchor.}

Given $x_t$, the \emph{event surprise} for block $b$ is the pre/post difference in mean logs,
\[
S^{\mathrm{event}}_b(x)
\equiv
\overline{x}\!\left(\mathcal{P}^{\text{post}}_b\right)
-
\overline{x}\!\left(\mathcal{P}^{\text{pre}}_b\right),
\]
which we record on a single trading day (the placement day). This yields a sparse series $S^{\mathrm{event}}_t(x)$: zero on non-placement days and equal to $S^{\mathrm{event}}_b(x)$ on the placement day. Because $x_t$ is in logs, $S^{\mathrm{event}}_b(x)$ is a percent change.

Complementing this, we also construct a \emph{daily} surprise series that does not use meeting dates. For every trading day $t$,
\[
S^{\mathrm{daily}}_t(x)
\equiv
\overline{x}\!\big(\{t,\ldots,t+\texttt{post}\}\big)
-
\overline{x}\!\big(\{t-\texttt{pre},\ldots,t-1\}\big),
\]
with the same inclusion choices for the anchor day. This fully populated series asks whether activity around $t$ is unusually high relative to its immediate pre window.

\paragraph{Estimation and interpretation.}
For each specification, we estimate
\[
S_t = \alpha + \beta D_t + \varepsilon_t,
\]
on trading-day data with Newey-West standard errors.

We consider three dummies $D_t$:
(i) a \emph{placement} dummy equal to one on the specification's placement dates,
(ii) an \emph{ALL-calendar} dummy for all FOMC meetings, and
(iii) a \emph{SCHEDULED-calendar} dummy for scheduled meetings only.
The coefficient $\beta$ is the difference between the average surprise on $D_t{=}1$ days and the unconditional mean of $S_t$. With event surprises, this recovers the mean meeting day spike; with daily surprises, it quantifies how much larger the average surprise is on meeting days than on typical days.

\begin{figure}[H]
\begin{center}
\includegraphics[width=0.9\linewidth]{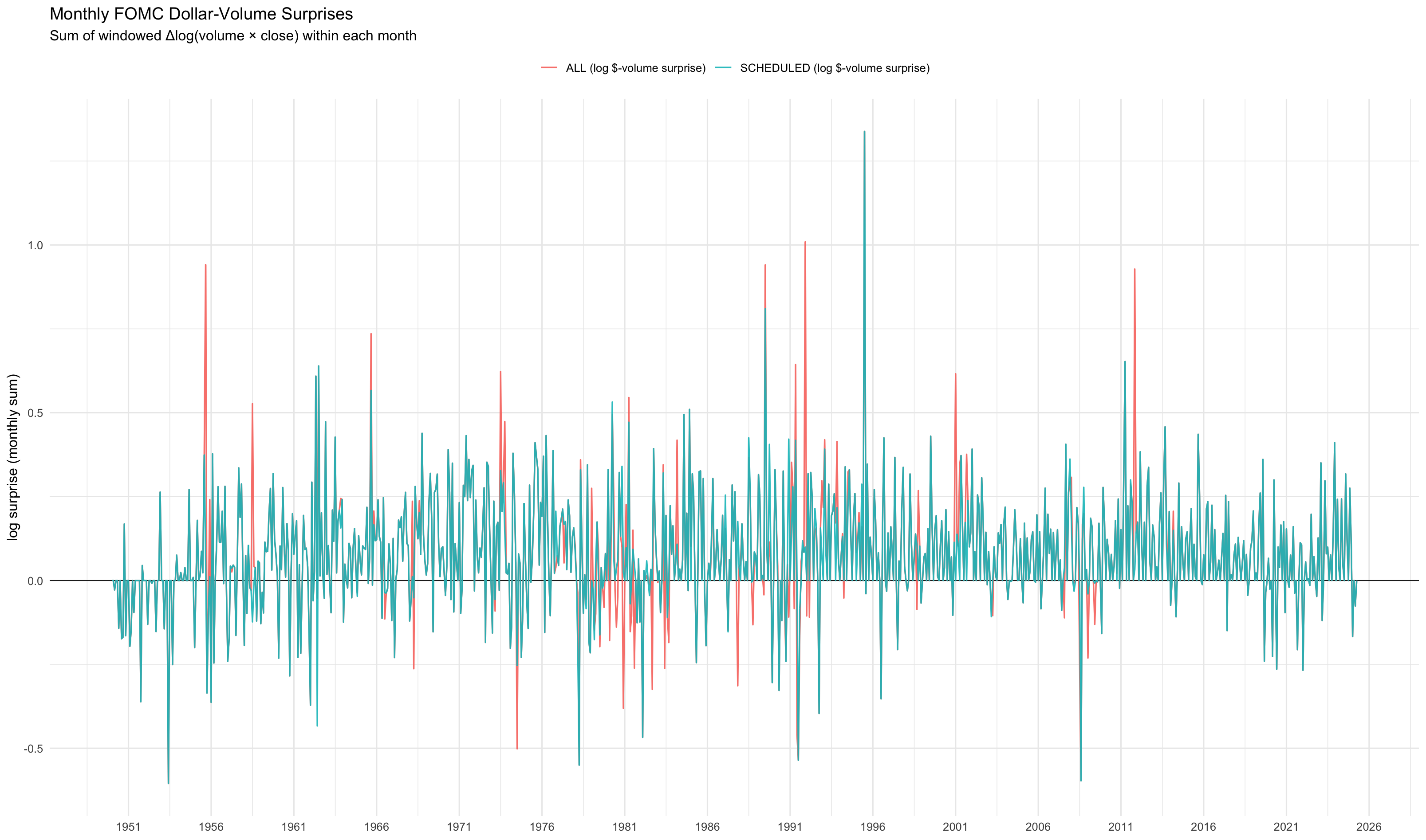}
\end{center}
\caption{Monthly FOMC trading-activity surprises (dollar volume).}\label{fig:usd_log_vol_pre1_post1_FALSE_TRUE}
{\footnotesize {\it Notes:} Monthly aggregation of event-day surprises constructed from pre/post windows around FOMC meeting blocks. Sample 1950-2025.}
\end{figure}

\paragraph{Specification grid.}
The design admits discrete choices that control the focus around meetings: which meeting set (ALL vs.\ SCHED), anchor mode (``final'' vs.\ ``block''), window sizes $\texttt{pre},\texttt{post}\in\{0,\ldots,K\}$, whether to include the anchor day in the post window, what placement date to choose(``anchor'' vs.\ ``end''). We perform a grid search over these choices and over base series $x_t\in\{v_t,d_t\}$. Table~\ref{tab:grid} summarizes the discrete choices for the grid. Our baseline reporting is full-sample: we present $\hat\beta$ with Newey--West $t$-statistics for both all FOMC meetings and scheduled only, as well as for series ($v_t$ and $d_t$), and various window designs.

\begin{table}[h!]
\centering
\caption{Grid of window and placement choices}
\label{tab:grid}
\begin{tabular}{ll}
\toprule
Dimension & Values \\
\midrule
Meeting set & \texttt{ALL} (all meetings), \texttt{SCHED} (scheduled) \\
Anchoring mode & \texttt{final} (anchor $=e_b$), \texttt{block} ($[s_b,e_b]$) \\
Pre window length & $\texttt{pre} \in \{0,1,\ldots,K\}$ (default $K\!\le\!5$) \\
Post window length & $\texttt{post} \in \{0,1,\ldots,K\}$ (default $K\!\le\!5$) \\
Include anchor in pre & FALSE/TRUE \\
Include anchor in post & TRUE/FALSE \\
Placement day & \texttt{end} ($e_b{+}\texttt{post}$) [default], \texttt{anchor} \\
Base series & $x_t \in \{\log \mathrm{Vol}_t,\;\log (\mathrm{Vol}_t\!\cdot\! P_t)\}$ \\
\bottomrule
\end{tabular}
\end{table}

\paragraph{Out of sample selection. Robustness.}
To avoid overfitting, we split the sample into two parts around February 4, 1994 (the start of modern FOMC statement practices) and recompute Newey--West $t$-statistics in-sample (pre-1994) and out-of-sample (post-1994). We rank specifications both by the size of the coefficient $\beta$ out-of-sample and by the out-of-sample $t$-statistic on $\beta$, retaining only cases with $\hat\beta>0$ and $t-stat>0$ both in- and out-of-sample; ties are broken by parsimony (smaller $\texttt{pre}+\texttt{post}$) and then by preferring dollar- to share-volume.

\paragraph{Illustration (timeline).}
Figure~\ref{fig:timeline} depicts a two-day meeting ($s_b$ to $e_b$) with $\texttt{pre}=1$, $\texttt{post}=2$, exclusion of the anchor both from the pre window and the post window, and placement at the post end. The event surprise is the difference between the mean of two post days and the mean of one pre day, recorded at $e_b{+}2$. Figure~\ref{fig:timeline2} is different only by inclusion of the anchor to the $\texttt{post}$ period, which changes the calculation of surprise. Now the event surprise is the difference between the mean of two post days and one anchor day, and the mean of one pre day, recorded at $e_b{+}2$.

\begin{figure}[t!]
\centering
\begin{tikzpicture}[>=Stealth,scale=1]

\draw[->] (-0.5,0) -- (9.0,0);

\foreach \x/\lab/\dy in {
  0/{$s_b{-}1$}/2mm,
  1.5/{$s_b$}/2mm,
  4.5/{$e_b$}/2mm,
  6/{$e_b{+}1$}/2mm,
  7.5/{$e_b{+}2$}/2mm
}{
  \draw (\x,0.12) -- (\x,-0.12);
  \fill (\x,0) circle (1.2pt);
  \node[below=\dy] at (\x,-0.12) {\lab};
}

\draw [decorate,decoration={calligraphic brace,amplitude=6pt,mirror,raise=6pt}]
  (0,-0.02) -- (1.5,-0.02)
  node[midway,below=20pt] {\small PRE};

\draw [decorate,decoration={calligraphic brace,amplitude=6pt,raise=6pt}]
  (1.5,0.02) -- (4.5,0.02)
  node[midway,above=10pt] {\small Meeting block $[s_b,e_b]$};

\draw [decorate,decoration={calligraphic brace,amplitude=6pt,mirror,raise=6pt}]
  (4.5,-0.02) -- (7.5,-0.02)
  node[midway,below=20pt] {\small POST};

\draw[->,thick] (7.5,0.85) -- (7.5,0.14);
\node[above] at (7.5,0.85) {\small placement at $e_b{+}\texttt{post}$};
\end{tikzpicture}
\caption{Pre/post windows for a two-day meeting with $\texttt{pre}=1$, $\texttt{post}=2$; single-day gaps equal, block spans two steps. Anchor is excluded both from the pre window and the post window.}
\label{fig:timeline}
\end{figure}

\begin{figure}[t!]
\centering
\begin{tikzpicture}[>=Stealth,scale=1]

\draw[->] (-0.5,0) -- (9.0,0);

\foreach \x/\lab/\dy in {
  0/{$s_b{-}1$}/2mm,
  1.5/{$s_b$}/2mm,
  4.5/{$e_b$}/2mm,
  6/{$e_b{+}1$}/2mm,
  7.5/{$e_b{+}2$}/2mm
}{
  \draw (\x,0.12) -- (\x,-0.12);
  \fill (\x,0) circle (1.2pt);
  \node[below=\dy] at (\x,-0.12) {\lab};
}

\draw [decorate,decoration={calligraphic brace,amplitude=6pt,mirror,raise=6pt}]
  (0,-0.02) -- (1.5,-0.02)
  node[midway,below=20pt] {\small PRE};

\draw [decorate,decoration={calligraphic brace,amplitude=6pt,raise=6pt}]
  (1.5,0.02) -- (4.5,0.02)
  node[midway,above=10pt] {\small Meeting block $[s_b,e_b]$};

\draw [decorate,decoration={calligraphic brace,amplitude=6pt,mirror,raise=6pt}]
  (3,-0.02) -- (7.5,-0.02)
  node[midway,below=20pt] {\small POST};

\draw[->,thick] (7.5,0.85) -- (7.5,0.14);
\node[above] at (7.5,0.85) {\small placement at $e_b{+}\texttt{post}$};
\end{tikzpicture}
\caption{Pre/post windows for a two-day meeting with $\texttt{pre}=1$, $\texttt{post}=2$; single-day gaps equal, block spans two steps. Anchor is excluded from the pre window and included in the post window.}
\label{fig:timeline2}
\end{figure}

\subsection{Additional figures}  \label{sec:app_extra_figs}

\begin{figure}[H]
\begin{center}
 \includegraphics[width=\linewidth ]{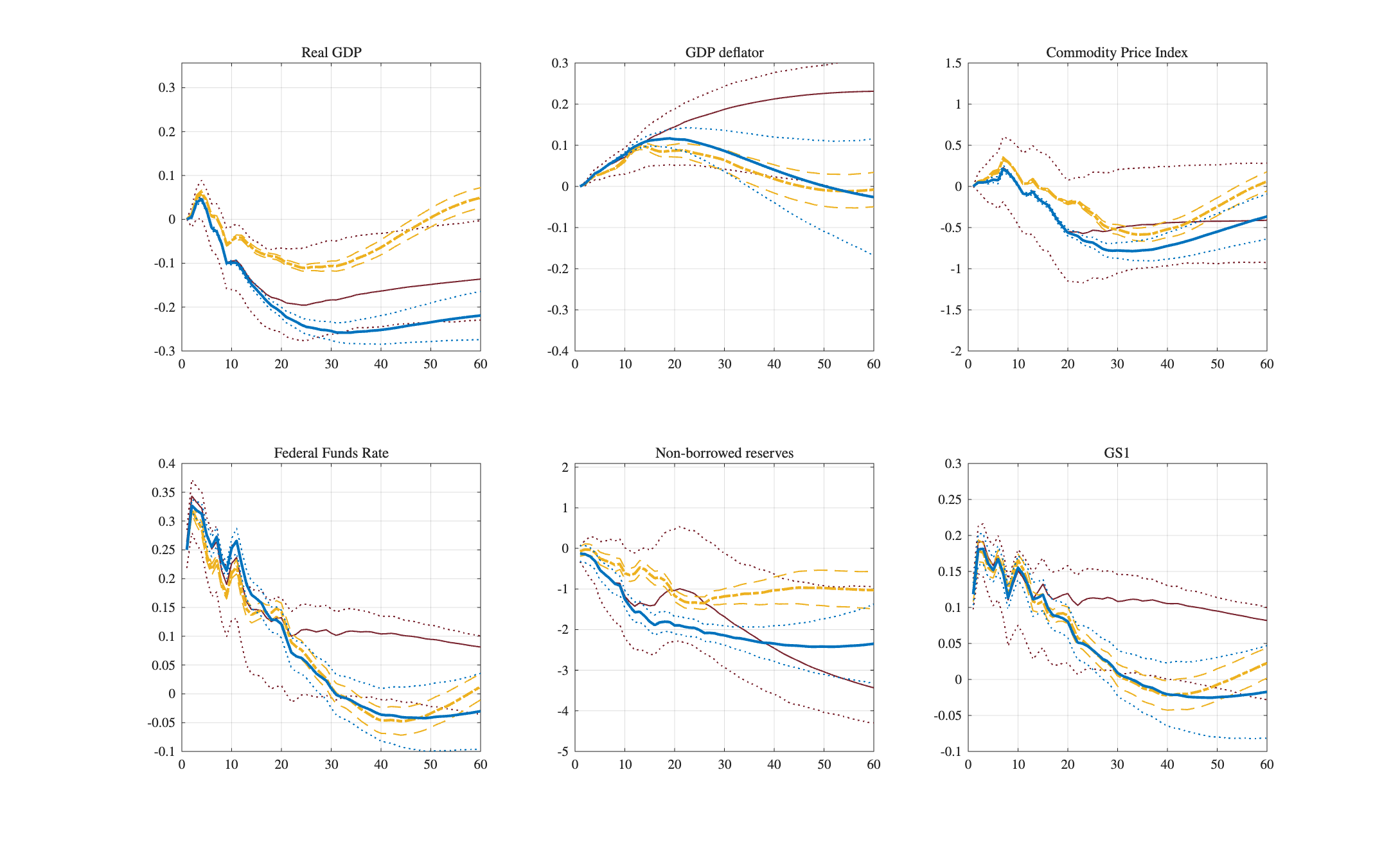}
\vspace{-1cm}
\caption{ Effect of a monetary policy shock on macroeconomic variables under various specifications. }
  \label{fig:IRF_CEEnoTrend}
\end{center}
\vspace{-0.5cm}
{\footnotesize {\it Notes:} Weighted LP IRF estimates under Cholesky identification  with corresponding   $90\%$  pointwise DWB confidence bands (dotted and dashed lines) for baseline specification ($k=6$, orange dashed line) and constant only specification ($k=0$, blue solid line).  Purple solid line and the corresponding dotted bounds are based on conventional recursive VAR IRF estimation  under homoscedasticity assumption with residual bootstrap with constant only, $k=0$. }

\end{figure}

\clearpage
\renewcommand{\refname}{Supplementary Appendix References}

\end{document}